\documentclass[12pt,twoside]{article}
\usepackage[utf8]{inputenc}
\usepackage{amsmath,graphicx,microtype,amssymb,verbatim,multirow,url,bbm,rotating,booktabs,amsthm}
\usepackage{float,hhline,afterpage,soul,pgfplots,pdfpages,subcaption}

\usepackage{pgfplots}
\pgfplotsset{compat=1.18}

\usepackage[font={normal}]{caption}
\usepackage{libertine}
\usepackage{color} 
\usepackage{sgame}
\usepackage{enumitem}
\usepackage{comment}
\usepackage{subcaption}
\usepackage{tikz}
\usetikzlibrary{positioning}
\usetikzlibrary{arrows.meta}
\usetikzlibrary{fit}

\makeatletter
	\renewcommand{\abstract}[1]{\def \@abstract {#1}}
	\newcommand{\jelcodes}[1]{\def \@jelcodes {#1}}
	\newcommand{\keywords}[1]{\def \@keywords {#1}}
	\newcommand{\thanknotes}[1]{\def \@thanknotes {#1}}
	\newcommand{\contact}[1]{\def \@contact {#1}}
	\newcommand{\shortauthor}[1]{\def \@shortauthor {#1}}
	\newcommand{\shorttitle}[1]{\def \@shorttitle {#1}}
\makeatother

\jelcodes{}
\keywords{}
\thanknotes{}
\shortauthor{}
\shorttitle{}
\abstract{}

\newcommand{\abs}[1]{\left\lvert{#1}\right\rvert} 

\usepackage[paperwidth=8.5in, paperheight=11in]{geometry}
\usepackage{setspace}
\usepackage[hang,flushmargin]{footmisc} 
\usepackage[authoryear]{natbib}

\makeatletter
\usepackage{fancyhdr}
\makeatother

\newcommand\blfootnote[1]{%
  \begingroup
  \renewcommand\thefootnote{}\footnote{#1}%
  \addtocounter{footnote}{-1}%
  \endgroup
}

\usepackage{authblk}
\makeatletter
\def \maketitle { 
	\thispagestyle{empty}
	\vspace*{0.1in}
	\blfootnote{\textsc{Contact.} \@contact.  \@thanknotes\\}
	\begin{center}
	\begin{minipage}{5.2in}
	\begin{center}
	{\large {\textbf{\@title}}}
	
	\vspace{0.2in}
	
	{\textsc{\@author}}
	
	\vspace{0.2in}
	
	{\@date}
	\end{center}
	
	\ifx\@abstract\@empty
	\relax
	\else
	{\small{\textsc{Abstract.} \@abstract}}
	\fi
	
	\ifx\@keywords\@empty
	\relax
	\else
	\vspace{0.2in}
	
	{\small\textsc{Keywords.} \@keywords.}
	\fi
	
	\ifx\@jelcodes\@empty
	\relax
	\else
	{\small\textsc{JEL Codes.} \@jelcodes.}
	\fi
	
	\end{minipage}
	\end{center} }
\makeatother

\makeatletter
\def\@seccntformat#1{\csname the#1\endcsname.\ }
\makeatother

\usepackage{sectsty}
\allsectionsfont{\noindent\normalsize}
\sectionfont{\centering\normalsize\uppercase}
\paragraphfont{\textnormal}

\DeclareMathOperator*{\argmin}{arg\,min}

\providecommand{\U}[1]{\protect\rule{.1in}{.1in}}

\newtheorem{theorem}{Theorem}

\newtheorem{subtheorem}{Theorem}[theorem]

\newtheorem{lemma}{Lemma}

\newtheorem{proposition}{Proposition}

\newtheorem{assumption}{Assumption}

\newtheorem{prop}[proposition]{Proposition}

\usepackage{hyperref}

\makeatletter
\renewcommand\@biblabel[1]{}
\makeatother

\hypersetup{hidelinks}

\begin{document}

\title{\Large Organization Design for Complex Worlds} 
\author{\textsc{Jonathan Libgober}}

\affil{University of Southern California}

\date{\today}

\abstract{I study the role of \emph{horizontal complexity}---defined as the variation in actions that similar tasks require---in organization design. A continuum of workers each choose an action to adapt to a local state that follows a Gaussian process across locations. Headquarters can group workers into \emph{teams}, simplifying the attention problem it faces in coordinating across the organization, at the cost of inhibiting adaptation. Tools from spectral theory identify how horizontal complexity affects team size: roughly speaking, variation left unresolved by headquarters expands teams, while variation concentrated along dimensions that headquarters absorbs through attention shrinks them.}

\keywords{Organizational hierarchy, Karhunen-Lo\'eve decomposition, complexity}

\shortauthor{Jonathan Libgober}
\shorttitle{}

\contact{\textsc{libgober.economics@gmail.com}}
\thanknotes{I am deeply grateful to Vittorio Bassi for several inspiring conversations, without which this paper would not have been written. I thank Arjada Bardhi, Krishna Dasaratha, Duarte Gon{\c c}alves, Navin Kartik, Nicolas Lambert, Luciano Pomatto, Eduard Talamas, Omer Tamuz and Nate Yoder for helpful comments. I used ChatGPT, Gemini, and Claude for AI-assisted research assistance. This assistance included exploratory proof discussion and algebra checks, as well as literature searches and language editing. Refine.ink was similarly used for proofreading. All arguments and exposition were selected, developed, and verified by me; I take full responsibility for the final manuscript and any errors.}

\setcounter{page}{1}

\maketitle

\vspace{20mm}

\begin{center}
\end{center} 

\normalsize

\setcounter{page}{0}

\newpage

\section{Introduction}

\footnotesize

\begin{quote} 
Our business is strong. Gross profit continues to grow, we continue to serve more and more customers, and profitability is improving. But something has changed. We're already seeing that the intelligence tools we're creating and using, paired with smaller and flatter teams, are enabling a new way of working which fundamentally changes what it means to build and run a company. And that's accelerating rapidly.\\ ~~
\\
---February 26, 2026 post on X by Jack Dorsey announcing a restructuring of Block.
\end{quote}

\normalsize

This paper studies the relationship between an organization's design and the \emph{horizontal complexity} of the environment it faces. I use ``horizontal'' to distinguish this notion from task difficulty, which provides a natural ``vertical'' dimension of complexity. Horizontal complexity instead measures how correlated the ideal executions across different tasks are as a function of task similarity. If, fixing similarity, this interdependence is weaker, then the environment would be described as more horizontally complex---even if task difficulty is itself unchanged. 

To illustrate this form of complexity and its interaction with organization design, consider the Manhattan Project, the United States' effort to develop the atomic bomb during World War II. A major difficulty of managing a project of its scale---at its peak, the Manhattan Project employed approximately 130,000 individuals---was the need to coordinate a disparate set of scientific, engineering, and production tasks. The tasks formed a \emph{landscape} of related decisions, horizontally differentiated by topic but not obviously in terms of difficulty. For instance, different researchers studied different properties of plutonium, with some of these properties in turn influencing more distant choices related to the bomb assembly method \citep{hoddeson1993critical}. A key challenge faced by the project's leaders---General Leslie Groves and J. Robert Oppenheimer central among them---was how to aggregate this information and integrate it into an overarching vision \citep{thorpe2000oppenheimer}. One difficulty here was managing coordination given limitations on what each team could be informed about \citep{borpujari2025adaptive}. Several features of the institutional design of the Manhattan Project grew out of the need to effectively manage and process this information.

This underlying problem is one that organizations frequently confront, even if the stakes are typically not quite so dramatic. Many production processes involve a bundled set of disparate tasks varying in relatedness: for instance, planes, cars, AI products, and point-of-sale systems are all the sum of various interrelated and interdependent components, with management facing the goal of interpreting the organization's opportunities and deficiencies to develop an overarching vision. The Manhattan Project is one salient example where an organizational hierarchy developed as part of managing a project with a high degree of horizontal complexity. This paper seeks to understand the relationship between horizontal complexity and hierarchy, with an eye toward how changes in the environment that influence the former become reflected in the latter.   

The framework I propose builds on a line of work in economic theory, dating back to at least \citet{MarschakRadner1972}, on the tradeoffs organizations face when balancing the necessity of letting workers \emph{adapt} to their local circumstances effectively while still \emph{coordinating} information across dispersed locations. The baseline model focuses on \emph{team size} as the primary margin of organization design. Team size influences the tradeoff between adaptation and coordination, facilitating the latter but inhibiting the former. Larger teams economize on the points of contact headquarters must track, but raise the cost of adapting each action to local conditions (for instance, by making the execution of each worker's task dependent on the other workers). I consider other organization design margins in extensions, including richer hierarchy layers.

My main results identify how complexity influences team size. While horizontal complexity \emph{generally} has an ambiguous impact on team size, the framework enables a sharp characterization of \emph{how}  horizontal complexity enters into this decision. The simple intuition is that larger teams help organizations manage complexity
that falls outside the dimensions headquarters attends to when coordinating the
organization. An increase in complexity primarily along those dimensions makes larger teams less effective at managing complexity, in a relative sense. An increase in complexity outside of those dimensions has the opposite impact. This decomposition illustrates why, for instance, a technology that simplifies worker tasks can increase team size but one that simplifies management's coordination efforts can decrease it.

Toward this end, I model uncertainty using Gaussian processes and apply tools from spectral theory to decompose the resulting state process. The use of Gaussian processes to describe a rich collection of interdependent states places my paper in a tradition following the influential work of \citet{Callander2011}, which deployed this framework in a search-and-learning problem. The reason the Gaussian process framework is particularly well-suited is that it allows (1) a tractable description of the joint correlation across a rich set of locations, while (2) correlation varies depending on the ``closeness'' across different locations in a natural way. It is therefore a natural framework to use to study horizontal complexity as defined above. Moreover, as in \citet{Callander2011} and the subsequent literature, I will exploit a significant degree of tractability facilitated by the continuous model to speak to the relationship between organization design (i.e., optimal team size) and horizontal complexity. 

In the model I propose, the Gaussian process determines the state process that each worker in the unit interval seeks to match. To explain the mechanics of the model, here I focus on the model's discrete approximation where $n+1$ workers are located on a fine grid $\{0,1/n, \ldots, 1\}$, with states being multivariate normal random vectors; this alternative is described formally in Appendix \ref{app:discrete}.  Members of a team are constrained to be adjacent, and team sizes are assumed to be regular---i.e., all the same size, as in \citet{DesseinSantos2006}.\footnote{An advantage of the continuum limit, however, is that unlike discrete models, I can consider comparative statics in team size that do not require the number of workers to be divisible by team size. The model treats team size as a real number, which as I describe could correspond to having some heterogeneity in the sizes of adjacent teams; e.g., a team size of 2.5 can be approximated by alternating teams of size 2 and 3.}

The organization's adaptation loss is the average loss across workers in their individual decision problems. The model posits that these losses are exogenously larger for larger teams, reflecting that each worker's action must be coordinated with more teammates, which magnifies losses across workers on the same team: if one team member is worse off, so are other team members.\footnote{Appendix \ref{app:discrete} describes how to build this feature into the discrete model. The only challenge in doing so endogenously in the continuum model is that each team is infinitesimally small; as I view these issues as a detour, the main text only focuses on the continuum setting directly with the understanding that it reflects a particular limit of such a discrete model. } The actions of all members of the team determine the \emph{team state}, which is the average of all team member actions. The choice of team size in the discrete approximation determines the number (or more precisely in the continuous limit, mass) of team-level states headquarters must track.

To model coordination, I assume that headquarters must predict the average action within each team. A loss is incurred per team that depends on the distance between this prediction and the true average action. This formulation reflects the idea that headquarters' coordination activities such as giving instructions happen at the team level, and that within-team communication is cheap relative to team-to-headquarters communication. The challenge is that headquarters has a fixed \emph{attention capacity} and can acquire only information about the team state process up to a fixed budget. I measure this amount using mutual information, but allow headquarters to observe any signal process jointly distributed with the team state process. The solution to this attention problem is classical, taking the form of the \emph{water-filling algorithm} \citep{CoverThomas2006}, previously applied in economics to a rational inattention problem by \cite{KoszegiMatejka2020}. The idea is to represent the correlated state process as a collection of independent components and then allocate attention to the components responsible for the most variation. Posterior variance is equalized across the components that receive attention.

To see why changes in horizontal complexity yield ambiguous predictions on team size in the baseline model, recall that team size balances the tradeoff between adaptation and coordination loss. Multiplying the variance of the individual states by a constant multiplies the organization's loss by a constant. But this increased loss is the same for adaptation and coordination, so the underlying tradeoff is unchanged. Hence, team size is unchanged as well. But some changes in the underlying state distribution increase coordination loss by less than adaptation loss. In particular, increases in variance along dimensions that headquarters pays attention to increase adaptation loss by a larger factor than coordination loss, since attention dampens the impact on coordination. Such increases therefore lead to more siloed teams. By contrast, increases in variance along dimensions outside headquarters' attention problem are not similarly dampened; these increases therefore strengthen the relative value of larger teams.

This dichotomy, between the components which significantly influence the attention problem and those that do not, drives the relationship between complexity and team size. The technical innovations yielding this insight involve extending the aforementioned results for (discrete-location) multivariate normal vectors to (continuous-location) Gaussian processes. These tools are useful because they identify which components of the organization's state process headquarters attends to.  Fundamentally, these extensions are possible because a similar decomposition applies to Gaussian processes, thanks to a result known as the \emph{Karhunen-Lo\'eve theorem}. The tools from spectral theory that I apply essentially translate the key ideas from the finite-dimensional case to the infinite-dimensional case. I show that, for any finite attention budget, the optimal solution allocates attention to only finitely many dimensions and none to the rest. This contrasts with the finite-dimensional case, where every dimension eventually enters the attention problem as the budget grows.

Despite this reduction, the continuum formulation allows me to convey the main economic message of the paper more transparently than the discrete-location formulation would allow. The familiar reason is that the continuum formulation enables simple, transparent comparative statics based on calculus.\footnote{Note that in the continuum limit, this is both a continuum of workers \emph{and} a continuum of teams, the latter of which has mass inversely related to team size. Section \ref{sec:divisions} describes a version with discrete \emph{divisions}, where the organization faces the decision of whether to merge the divisions or not.}  But the continuum formulation does more than translate comparative statics into calculus. A finite grid provides only a fixed collection of task distances, making it awkward to study how correlation changes while distance is held fixed. Indexing tasks on a continuum holds the location space fixed while allowing the covariance kernel to vary, so that changes in relatedness at any given distance can be isolated.

The tools from spectral theory I apply also help interpret the relationship between horizontal interdependence and organization design. Section \ref{sec:teamsize} proposes sample path ``jaggedness'' as one possible local measure of complexity. Using the spectral properties of the Gaussian process, I show that, under some conditions on the eigenfunctions, decreasing the eigenvalue decay rate while holding total variance fixed weakly increases both jaggedness and team size for all attention levels, strictly so away from the relevant edge cases. This result has a natural empirical implication, with the caveat that a given environmental change may not match a pure change in jaggedness as I describe it. Still, it suggests that researchers should be careful with associating volatility itself with complexity, as what matters may be how this complexity is spread out across the organization---namely, whether workers face correlated or idiosyncratic uncertainty. 

In extensions I show how these insights translate into other margins of organization design, such as training or decentralization decisions. For now, two are worth highlighting more thoroughly. First, I describe how team size roughly corresponds to hierarchy depth---while these different margins may have different costs in practice, my results on the relationship between horizontal complexity and team size also translate to hierarchy expansion more broadly. Second, I show how the ability to more efficiently simplify an organization's environment---say, due to artificial intelligence technologies---naturally leads to larger teams. The simple reason is that the technology has a larger impact on the components that headquarters pays attention to, since these tend to be the most significant ones (which is why they receive attention in the first place). The \emph{unattended} components therefore become relatively \emph{more} important, increasing team size by the logic above. This result highlights the potential dual impact of such technologies. Consistent with the Jack Dorsey quote above, I find that an increase in the attention budget itself decreases team size. But my framework shows that the way such changes influence an organization overall should depend on where the impact is concentrated. This paper provides a unified framework to think about such changes systematically, showing that the resolution of these conflicting forces should depend on how much residual complexity lies outside the dimensions
to which management allocates attention.

\section{Literature} 

The main tradeoff this paper studies, between adaptation and coordination, is inspired by the theoretical literature on organization design. The framework introduced in \cite{DesseinSantos2006} defines adaptation losses as those that depend on the primary action each worker takes and his local state, while coordination losses depend on how well each worker knows about the primary actions of other workers. The organization design problem in \cite{DesseinSantos2006} consists of the assignment of tasks to workers, where assigning more tasks to a given worker would increase the cost associated with carrying out each task, but would lessen communication frictions across workers (and hence facilitate coordination).\footnote{My results in Section \ref{sec:divisions} also speak to the question of when organizations are better off merging versus separating divisions, which is related to the question of centralization that some work following \cite{DesseinSantos2006} has also considered---notably, \cite{Rantakari2008GoverningAdaptation} and \cite{AlonsoDesseinMatouschek2008}.   }

While my main focus is also on this aspect of organization design, two other ingredients are shared with the literature on organization design following \cite{MarschakRadner1972}. The first is the possibility of direct interdependence between tasks, studied in a general linear-quadratic framework by  \cite{CalvoArmengolDeMartiPrat2015} and applied to modular production technologies and the \emph{mirroring hypothesis} by \cite{MatouschekPowellReich2025}. However, a key property of these papers is that states themselves are \emph{in}dependent, with this state interdependence being a fundamental component of my notion of horizontal complexity. I show how allowing for correlation between states provides nuanced insights for how the \emph{form} of complexity influences the shape of organizations. But on the other hand, the production technology I study does not feature such dependence, so that unlike in those papers, interdependence largely relates to the coordination problem.

The second is the presence of a centralized attention problem. While the notion that managers generally face attention constraints is a fundamental component of \cite{MarschakRadner1972} and in the following literature (e.g., \cite{GeanakoplosMilgrom1991Hierarchies, Radner1993DecentralizedInformationProcessing,BoltonDewatripont1994,VanZandt1999}), in my model this is actively chosen by headquarters alongside the rest of the organization design problem. In this sense, the problem resembles \cite{DesseinGaleottiSantos2016} and \cite{DesseinSantos2021}, where the allocation of attention is also a design choice of the organization.

While the baseline model follows much theoretical work on organization in treating all workers or tasks horizontally, this paper will also allow for richer hierarchies, which can sometimes balance adaptation and coordination more effectively than single-layer organizations. An influential example that illustrates how hierarchies can emerge to manage task difficulty is \cite{Garicano2000}. But in \cite{Garicano2000}, tasks can be directly ordered by difficulty; thus, changes that increase productivity as in \cite{GaricanoRossiHansberg2012OrganizingGrowth} would more appropriately reflect \emph{vertical complexity} rather than horizontal.\footnote{Another paper studying hierarchy design in such a horizontal model is \cite{IchihashiLiZou2025}, which instead focuses on the question of how to \emph{assign} tasks to agents to incentivize effort. } This framework has proven useful in applications; \cite{IdeTalamas2025AIKnowledgeEconomy} study the impact of artificial intelligence on organizational hierarchy in a model building on \cite{Garicano2000}. Methodologically, however, team size in my paper essentially parameterizes the ``span-of-control'' as in that literature, bundling lower-level problems into a single higher-level coordination unit.

The particular approach to modeling the coordination loss borrows heavy inspiration from the literature on rational inattention. Technically, I solve the coordination problem by applying ideas from the water-filling algorithm, described in \cite{CoverThomas2006}. \cite{KoszegiMatejka2020} build on similar ideas to describe mental budgeting. The basic insight here is that the posterior variance is minimized along the principal components of the random vector the agent faces.\footnote{The use of principal-components analysis in economic theory by now has significant precedent and is familiar in several lines of work; see, for instance, \cite{CalvoArmengolDeMartiPrat2015,GaleottiGolubGoyal2020,GaleottiEtAl2025}.} Relative to this literature, I consider a setting with a \emph{continuum of states}. While there is some precedent for modelling using continuum formulations in rational inattention problems (e.g., \cite{MackowiakWiederholt2009}), I am not aware of past work in the economics literature on attention where the state itself is a function on the continuum. Having said that, some recent papers have also found that eigenvalue and eigenvector methods can enable tractable analysis in continuum models---notably, \cite{Miyashita2024} and \cite{FarboodiKohXia2024}.

Lastly, I formalize the complexity of the world faced by the organization using the formal structure pioneered by \cite{Callander2011}. The influence of this modelling device in economic theory has been substantial, as surveyed by \citet{BardhiCallander2026LearningCorrelatedWorld}. A partial list includes voting \citep{BardhiBobkova2023LocalEvidence}, communication \citep{callander2021power,dong2024communication}, coordination problems \citep{DallAra2025Coordination}, and contracting \citep{Bardhi2024Attributes}. Relative to this literature, a key technical novelty is identifying the usefulness of the Karhunen-Lo\'eve theorem in assessing ``complexity.'' I show that team size itself reacts to ``local complexity'' more than ``global complexity,'' to the extent that the latter is measured by the process' variance and the former reflects the ``jaggedness'' of the sample paths.

\section{Model}
\label{sec:model}

The model considers an organization choosing a continuum of actions to match a stochastic process defined on that continuum. In Appendix \ref{app:discrete}, I describe a discrete-location version that converges to the continuum model as the grid becomes arbitrarily fine.

\medskip 

\noindent \textbf{Setting} An organization consists of a continuum of \emph{tasks}, $x \in [0,1]$, with each task associated with a \emph{worker}. Associated with each task location $x$ is a state denoted  $W_{0}(x)$. Tasks are interdependent, so that $W_{0}(x)$ is a stochastic process correlated across space. In particular, I assume that $W_{0}(x)$ is a mean-zero Gaussian process with symmetric covariance kernel $C(x,x')$. I assume that $C(x,x')$ is positive definite and continuous. I let $\sigma^{2}(x)$ denote the variance of this stochastic process, and assume that $\int_{0}^{1} \sigma^{2}(x)dx < \infty$. I refer to workers undertaking these tasks as the \emph{ground layer} of the organization.  Below, I describe the individual decision problem that each worker in the ground layer faces. The goal of the organization is to balance losses due to coordination and adaptation. 

\medskip

\noindent \textbf{Hierarchy} Workers can be placed into \emph{teams}. To simplify the analysis, in this paper I focus on the case of \emph{symmetric organizations}, where all teams are the same size, as in \cite{DesseinSantos2006}. While I will enrich the model to allow for the organization to design more elaborate hierarchies, for now I omit this possibility and defer this to Section \ref{sec:hierarchies}. The organization chooses teams (and later hierarchies) as a way of optimally balancing two different kinds of losses, which I refer to as \emph{adaptation loss} and  \emph{coordination loss}.

\medskip

\noindent \textbf{Adaptation Loss} Each worker chooses an action $a_{x}$ to best match their own state $W_{0}(x)$. However, matching this state is costly---not only that, doing so is more costly the larger the group size is. Specifically, the loss a worker at location $x$ faces is: 

\begin{equation*} 
\ell_{x}(a_{x}, W_{0}(x)) = f(k) \cdot [a_{x}^{2} + (a_{x} - W_{0}(x))^{2}]. 
\end{equation*}

\noindent I assume that the function $f(k)$ is strictly increasing and differentiable in $k$, equal to 1 at $k=1$ with $\lim_{k \rightarrow \infty} f(k)=\infty$. I have in mind that the choice of an action $a_{x}$ requires coordination with other group members. If $a_{x}$ is a level of precision in some input, and if each point of contact in the team must be contacted, then the marginal cost of $a_{x}$ should increase with the team size. Appendix \ref{app:discrete} shows how this microfoundation can generate losses that scale multiplicatively with team size in the discrete-location version of the model, but I treat it as an assumption since the model focuses on the continuous-location limit. 

I mention that my model will allow team size to be any real number $k$; this assumption follows \citet{Garicano2000} and the subsequent literature as a way of continuously measuring span of control in a large organization.\footnote{\citet{Garicano2000} assumes that the organization is sufficiently large that integer constraints can be ignored. \citet{FuchsGaricanoRayo2015} provide a related pointwise-matching interpretation, in which noninteger team sizes approximate matches between small positive masses of producers and consultants. See also \citet{GaricanoRossiHansberg2012OrganizingGrowth} and \citet{IdeTalamas2025AIKnowledgeEconomy}.} The same limiting team geometry could be obtained by assuming that, within any subinterval, some teams have size $\lfloor k \rfloor$ and some teams have size $\lceil k \rceil$---e.g., $k=7/3$ implies $1/3$ of teams are size 3 and $2/3$ are size 2. In the discrete approximation as the worker grid becomes fine, this formulation yields the same limiting mass of teams and the same limiting team-state process described below. I take $f(k)$ to measure the average adaptation cost associated with a span of control $k$, rather than deriving it mechanically by averaging a fixed cost schedule defined only at integer team sizes. I work with continuous $k$ and a smooth $f$ so that optimal team size can be characterized transparently by balancing marginal benefit and marginal cost.

I assume that $a_{x}$ is chosen to minimize the loss at location $x$. Under this assumption, it is immediate that: \begin{equation*} a_{x} = \frac{W_{0}(x)}{2}~~~ \text{ and } ~~~\ell_{x} \left(\frac{W_{0}(x)}{2}, W_{0}(x) \right)= \frac{W_{0}(x)^{2} f(k)}{2}. \end{equation*}

\noindent Abusing notation, I write $l(W_{0}(x))$ (without the $x$ subscript) to denote $\frac{W_{0}(x)^{2}f(k)}{2}$, the realized loss. 

A leading case of interest is one where $f(k)=1+\beta(k-1)$, reflecting the idea that the loss-per-team member of coordinating with others is $\beta$. This reflects the benchmark studied in \cite{DesseinSantos2006} in the special case where the ``primary task'' at location $x$ has a fixed weight while loss for ``bundled tasks'' is linear-per-task. I use this specification to derive sharp predictions in some later sections, but for now only mention that in principle nonlinear scaling is allowed. Notice that in this formulation, the action itself is not influenced by team size.

\medskip 

\noindent \textbf{Team States} The organization's coordination loss will depend on what the organization knows about the \emph{team states} across the organization. Heuristically, each team state can be associated with a ``team leader'' who is not a member of the ground layer of the organization and whose state is the average action of all team members. Appendix \ref{app:discrete} walks through the interpretation of the team state as the ``average action of team members'' in the discrete approximation of the model and derives the resulting specification of the covariance operator of the team-state process $W_{1}(x)$, presented below.  I defer this from the main text as the main analysis is all done in the continuum limit, but briefly summarize how this exercise motivates my specification for $W_{1}(x)$. As the worker grid is made fine with team size fixed, each team's width shrinks to zero; by the (mean-square) continuity of the task-state process, the average of team members' task states converges to the state at a representative point (Appendix \ref{app:discrete}); the resulting team-state process is supported on an interval of mass $1/k$ rather than 1.

These observations motivate my formal assumption: that the team-state process is $W_{1}(u)=W_{0}(ku)/2$ on $[0,1/k]$, so that the covariance between $W_{1}(u)$ and $W_{1}(u')$ is $\frac{1}{4}C(ku, ku')$ for all $u, u' \in [0,1/k]$. The division of the task state by 2 reflects that $a_{x} = \frac{W_{0}(x)}{2}$, and that the team-state concerns worker actions rather than states. Mathematically, larger $k$ rescales the location index of the task process while reducing the mass of teams; the underlying task process on $[0,1]$ is unchanged.  

The two extreme cases are worth commenting on. When $k=1$, then groups consist of individuals and hence the group state is simply $\frac{1}{2}W_{0}(x)$. But as $k \rightarrow \infty$, headquarters tracks vanishingly few points of contact, so the coordination burden becomes negligible. Choosing $k$ large is beneficial for coordination, but limits adaptation, while choosing $k$ small makes coordination more difficult, but facilitates adaptation. 

\medskip

\noindent \textbf{Coordination Loss} Headquarters incurs a loss depending on how much information it has about the team states. I formalize this in a manner similar to the rational-inattention literature, using mutual information. Headquarters chooses a signal $S$ that is informative about $W_{1}$. A signal $S$ is \emph{admissible} if it is a real-valued stochastic process on $[0,1/k]$, defined jointly with $W_{1}$ on a common probability space and jointly measurable in the underlying state $\omega$ and location $x$.\footnote{That is, the map $(\omega, x ) \mapsto S(\omega,x)$ is assumed to be jointly measurable, where $(\Omega, \mathcal{F}, \mathbb{P})$ is the probability space generating all uncertainty and $\omega \in \Omega$ is a sample point. The common probability space may be enlarged, if necessary, to accommodate independent signal noise. Formally, I also assume $S$ is separable in the sense that $\mathcal{F}_{S}:=\sigma(S(x) : x \in [0,1/k])$ is (modulo null events) generated by the values of $S$ at a countable collection of locations. Conditioning on $S$ means conditioning on $\mathcal{F}_{S}$; a priori, no Gaussianity or path-continuity restrictions are imposed on $S$.}

I define the mutual information between $W_{1}$ and $S$ to be: 

\begin{equation} 
I \left( W_{1};S \right)= \sup_{m < \infty, x_{1}, \ldots, x_{m} \in [0,1/k]} I(\{W_{1}(x_{1}), \ldots, W_{1}(x_{m})\};\{S(x_{1}), \ldots, S(x_{m})\}). \label{eq:MIDef}
\end{equation}

\noindent The constraint headquarters faces is that $S$ is restricted so that $I(W_{1}; S) \leq \rho$, where $\rho$ is exogenous. Headquarters then attempts to match the team states in the population. Letting $b(x) = \mathbb{E}[W_{1}(x) \mid S]$, the coordination loss incurred is: 

\begin{equation*} 
\int_{0}^{1/k} (b(x) - W_{1}(x))^{2} dx.
\end{equation*}

\noindent An important feature is that this loss is incurred \emph{per team} rather than per worker. I interpret this as a headquarters-level loss: headquarters chooses a mission, target, or intervention for each team, and the cost depends on how well informed it is about that team's state. This feature is in the spirit of \citet{Cremer1980PartialTheory}, where organizational boundaries determine the units over which local adjustment is possible and the residual coordination loss is incurred per unit.\footnote{In \citet{Cremer1980PartialTheory}, the primitive units are shops and the chosen groups are ``services.'' Transfers across services must be planned before uncertainty is realized, while allocations within a service can adapt ex post. The fully integrated organization is used only as an infeasible benchmark; here, full centralization is feasible but suboptimal due to adaptation costs.} \citet{GeanakoplosMilgrom1991Hierarchies} also obtain this feature, in a model where hierarchies economize on scarce managerial attention. The quadratic loss is a tractable way to capture the cost of acting on team-level information imperfectly.

\medskip 

\noindent \textbf{Summary and Timing} To review, the timing of the model is as follows: 

\begin{itemize} 
\item The organization decides on the organizational structure (the size of each group), as well as which signal $S$ will be subsequently observed subject to the aforementioned constraints.

\item States are realized, and actions are chosen as described above. 

\item Headquarters observes $S$ and chooses $b(x)$ for all $x$. 

\item The organization's expected payoff is expected output $\Pi$, minus the coordination and adaptation losses:

\begin{equation*} 
\Pi-\overbrace{\mathbb{E} \left[\int_{0}^{1/k} (b(x) - W_{1}(x))^{2}dx \right]}^{(\text{Coordination})} - \overbrace{ \mathbb{E} \left[\int_{0}^{1} l(W_{0}(x))dx \right]}^{(\text{Adaptation})}. 
\end{equation*}

\end{itemize}

\subsection{Warm-Up: Zero Attention Capacity}
\label{subsec:noattention} 

To help isolate the role of headquarters' information constraint, I first walk through the solution to the team design problem in the case where $\rho=0$. Here, the organization's incentives for larger teams are strongest, reflecting the idea that within-team communication has the highest importance when headquarters is fully in the dark about the organization's activities. Substituting in the solution for $a_{x}$ above, using the linearity of expectation and that $W_{1}(x)$ is mean zero (so that $b(x)=0$), I have that the objective can be written as: 

\begin{equation*} 
\Pi-\int_{0}^{1/k}\mathbb{E} \left[ W_{1}(x)^{2} \right]dx - \frac{f(k)}{2} \int_{0}^{1} \mathbb{E} \left[W_{0}(x)^{2} \right]dx .
\end{equation*}

\noindent Again using that both $W_{1}(x)$ and $W_{0}(x)$ are mean 0, I have that both integrals in the objective are simply variances. I have: 

\begin{equation*}
\int_{0}^{1/k}\mathbb{E} \left[ W_{1}(x)^{2} \right]dx= \frac{1}{4}\int_{0}^{1/k} C(kx,kx) dx = \frac{1}{4k} \int_{0}^{1}C(u,u)du = \frac{1}{4k} \int_{0}^{1} \sigma^{2}(x)dx. 
\end{equation*} 

\noindent Thus, $k$ is chosen to minimize: 

\begin{equation*} 
 \left( \frac{1}{4k} + \frac{f(k)}{2} \right) \int_{0}^{1} \sigma^{2}(x)dx. 
\end{equation*}

\noindent This illustration shows that the optimal team size has a simple form, which in particular depends only on $f$ (and does not depend on any properties of the process $W_{0}(x)$). The optimal team size manages the tradeoff between coordination loss, which is decreasing in $k$, and adaptation loss, which is increasing in $k$. The first-order condition for $k$ implies team size, if interior, should satisfy: 

\begin{equation*} 
1=2k^{2}f'(k). 
\end{equation*}

\noindent Furthermore, convexity of $f$ implies that as long as $f'(1) < 1/2$, the optimal team size will indeed be strictly greater than 1.

On the other hand, this argument also yields a stark prediction that as long as $\int_{0}^{1} \sigma^{2}(x)dx > 0$, team size is constant. Thus, as long as the underlying process is nondegenerate, the details do not influence the organizational structure. Once the attention constraint is introduced, this property will fail, and one of this paper's contributions is to illustrate \emph{which} characteristics matter and why.

\section{Using Spectral Theory to Solve the Attention Problem}  \label{sec:spectral}

I now walk through the solution to the organization's attention problem. The underlying structure used will also drive the solution to the organization design problem, but for the moment I take this as given. The solution makes use of tools from spectral theory which will be useful for several extensions as well.   Throughout the rest of the paper, I assume $\rho > 0$. 

A key difference between the attention problem in Section \ref{sec:model} and those from past work is the fact that there is a continuum of random variables, rather than a discrete number. For the case of multivariate normal random variables, the solution to the problem of minimizing posterior variance subject to a constraint on mutual information takes the form of the well-known \emph{water-filling} algorithm \citep{CoverThomas2006}, applied to a rational inattention setting by \cite{KoszegiMatejka2020}. The idea is the following: 

\begin{itemize} 
\item For a random vector $(\theta_{1}, \ldots, \theta_{n}) \sim \mathcal{N}(0, \Sigma^{2})$, note that there exists $\lambda_{1}, \ldots, \lambda_{n} \in \mathbb{R}_{+}$ and $v_{1}, \ldots, v_{n} \in \mathbb{R}^{n}$, such that, for $Z_{1}, \ldots, Z_{n}$ IID standard normal, this vector has the same distribution as $\sum_{j=1}^{n} \sqrt{\lambda_{j}} v_{j} Z_{j}$.  

\item Then, the attention problem involves obtaining signals informative of each $Z_{j}$. In this case, there is some $\mu$ such that either: (1) $\lambda_{j} < \mu$, in which case no information about $Z_{j}$ is obtained, or (2) $\lambda_{j} \geq \mu$, in which case information \emph{is} obtained, and the posterior variance of $\sqrt{\lambda_{j}} Z_{j}$ is equal to $\mu$.  
\end{itemize}

\noindent The first bullet makes use of the principal-components representation of normal random variables, which has found extensive use in economics as surveyed above. The parameter $\mu$ is referred to as the ``water level.'' It turns out that the optimal posterior mean admits a finite-location representation, summarized by the Proposition below:

\begin{proposition} \label{lem:RISol}
For any $\rho < \infty$, there exists an optimal signal and a finite mesh
$\{x_{1}, \ldots, x_{n}\}$ such that $b$ is determined by the vector
$(b(x_{1}), \ldots, b(x_{n}))$: for every $x$, there exist deterministic
coefficients $\alpha_1(x),\ldots,\alpha_n(x)$ such that
\[
b(x)=\sum_{i=1}^n \alpha_i(x)b(x_i).
\]
\end{proposition}

\noindent This Proposition reduces the infinite-dimensional attention problem to a finite-dimensional one---where headquarters' prediction of the team state only depends on its prediction of the team state at a finite number of locations.\footnote{This result does \emph{not} say that, given a mesh of $n$ points, the problem is equivalent to one where these $n$ points are used to define a multivariate normal vector. This statement is indeed generally false. Only the posterior mean function $b(x)$ is determined by its realization at finitely many points, although these values will typically depend on the full sample-path realization.} 

This extension to the continuous location case makes use of the \emph{Karhunen-Lo\'eve Theorem},\footnote{While the application of Karhunen-Lo\'eve Theorem to the attention literature is novel to my knowledge, the observation that it enables an extension of water-filling to the Gaussian process case was made by \cite{kolmogorov1956shannon} (see also \cite{pinsker1964information} and \cite{donoho1998datacompression}). The proof of Proposition \ref{lem:RISol} adapts this argument to the present location-space formulation, where mutual information is defined through finite-dimensional meshes rather than on the processes themselves. The additional implication obtained is that the optimal posterior mean can be represented by its values at finitely many locations.} which states that a mean-zero Gaussian process on $[0,1]$ with continuous covariance kernel admits the representation

\begin{equation} 
G(x) = \sum_{j} \sqrt{\lambda_{j}} \phi_{j}(x) Z_{j},  \label{eq:KLExpansion}
\end{equation}

\noindent with the sum ranging over the strictly positive eigenvalues (and converging in $L^{2}$, uniformly in $x$), where $Z_{1}, Z_{2}, \ldots$ is a sequence of IID standard normal random variables and $\lambda_{j} > 0$.\footnote{Without the restriction to Gaussian processes one loses the IID property of the $Z_{j}$ variables; aside from this feature, this result holds for more general stochastic processes as well.} The formulation (\ref{eq:KLExpansion}) is the Karhunen-Lo\'eve decomposition for the process. Each $Z_{k}$ is referred to as a \emph{mode}. The pairs $(\lambda_{j}, \phi_{j}(x))$ are the eigenvalue-eigenfunction pairs of the covariance operator of $G$. In the present setting, this is the covariance of the outcome process, which is expressed above. Not only is this the case, but the functions $\phi_{k}(x)$ form an \emph{orthonormal eigenbasis}. Thus, each $Z_{k}$ variable can be constructed given the outcome realization as: 

\begin{equation*} 
Z_{k} =\frac{1}{\sqrt{\lambda_{k}}}\int_{0}^{1} G(x) \phi_{k}(x) dx
\end{equation*}

The proof of Proposition 1 shows that there exist independent mean-zero normally distributed random variables $S_{1}, \ldots, S_{n}$ such that
\begin{equation*}
b(x)=\sum_{j=1}^{n} S_{j}\psi_j(x),
\end{equation*}
where $\psi_j$ are the eigenfunctions of the team-state process. Each $S_j$ is the posterior mean of the scaled mode $\sqrt{\nu_{j}}Z_{j}$ given headquarters' optimally chosen signal about that mode, where the signal's precision is set so that the information constraint is satisfied.

Proposition \ref{lem:RISol} shows that the attention problem retains the key aforementioned features from the finite-dimensional case: headquarters decides which $Z_{k}$ variables to pay attention to, and chooses $S_{j}$ so that the posterior variance along each attended dimension is equal. Henceforth I denote this posterior variance level as $\mu$.

In particular, there is a closed form expression for the coordination loss in terms of $k$ and the eigenvalues of the covariance operator of the stochastic process. Indeed, the adaptation loss is $\frac{f(k)}{2} \mathbb{E}[W_{0}(x)^{2}]$; the identity that $\int_{0}^{1} \sigma^{2}(x)dx= \sum_{i=1}^{\infty} \lambda_{i}$ allows us to express this term using the eigenvalues directly.\footnote{As $\sigma^{2}(x)=C(x,x)$, this identity is the continuous-location version of the fact that the sum of the eigenvalues is equal to the trace.}  The coordination term is more subtle but also straightforward in light of the above. Recall that the team-state process is equal to $\frac{W_{0}(kx)}{2}$ on the interval $[0,1/k]$ with covariance $\frac{1}{4} C(kx,kx')$. The change of variables formula and the eigenvalue-eigenfunction identity for the task process implies that: 

\begin{equation*} 
\int_{0}^{1/k}\frac{1}{4}C(kx, kx') \phi_{i}(kx')dx' = \frac{1}{4k}\int_{0}^{1} C(kx,u)\phi_{i}(u) du = \frac{\lambda_{i}}{4k} \phi_{i}(kx). 
\end{equation*}

\noindent Hence $(\lambda_{i}, \phi_{i}(x))$ are an eigenvalue-eigenfunction pair for the task process if and only if $(\frac{\lambda_{i}}{4k}, \psi_{i}(x))$ are an eigenvalue-eigenfunction pair for the team process, with $\psi_{i}(x)=\sqrt{k} \phi_{i}(kx)$.\footnote{The scaling by $\sqrt{k}$ ensures $\psi_{i}(x)$ has norm 1, as is the convention.} Putting this together, the organization's objective is to choose $k$ to maximize:

\begin{equation} 
\overbrace{-n \mu}^{(a)} - \overbrace{\sum_{m=n+1}^{\infty} \frac{1}{4k}  \lambda_{m}}^{(b)} - \frac{f(k)}{2}\sum_{m=1}^{\infty} \lambda_{m} ,\label{eq:solvedobj}
\end{equation}

\noindent where $\mu$---referred to as the \emph{water level}---is set so the attention constraint is satisfied: \begin{equation} \rho = \frac{1}{2} \sum_{m=1}^{n} \log \frac{\frac{1}{4k}  \lambda_{m}}{\mu}. \label{eq:attention} \end{equation}

\noindent In (\ref{eq:attention}), $n$ refers to the number of variables that enter into headquarters' attention problem. The \emph{active set} consists of the modes that enter this attention problem; the \emph{inactive set} consists of the modes that do not. I refer to the eigenvalues that do not enter into the attention problem as the \emph{tail} and eigenvalues that do as the \emph{head}. I maintain the following assumption throughout the remainder of the paper: 

\begin{assumption}  \label{assum:infinite}
The Karhunen-Lo\'eve decomposition involves $\lambda_{k} > 0$ for infinitely many $k$. 
\end{assumption}

\noindent While not essential for many of my results, making Assumption \ref{assum:infinite} allows me to ensure that tail eigenvalues always exist and focus on the truly infinite-dimensional case. This assumption is also quite weak; it is necessarily satisfied for any Gaussian process with nondifferentiable sample paths where $C(x,x')$ is smooth.\footnote{Indeed, if $C(x, x')$ is smooth, then smoothness of the eigenfunctions follows from differentiating the identity $\lambda_{k} \phi_{k}(x) = \int_{0}^{1} C(x,x') \phi_{k}(x)$, so nondifferentiability is only possible if $\lambda_{k} > 0$ infinitely often.} Its role is to focus on the key distinguishing feature compared to other work in economics that uses similar machinery.

\section{Complexity and Team Size} \label{sec:teamsize}

I now turn to the organization's optimization over $k$, showing how the attention problem influences team size---and especially that the conclusion from Section \ref{subsec:noattention} no longer holds once the attention constraint is introduced. Writing $\mu$ to be a function of $k$, note that as $k$ varies, $\mu(k)$ adjusts so that (\ref{eq:attention}) is satisfied. Furthermore, since the right-hand side of (\ref{eq:attention}) is monotone in $\mu$, there is a unique value for which equation is satisfied. On the other hand, conjecturing that $\mu$ scales according to $1/4k$, the constraint is satisfied with equality that does not depend on $k$. Since all eigenvalues of the team-state process are obtained from the task-state eigenvalues by the common scaling factor $1/(4k)$, $\mu$ scales in the same way: Thus, there exists $\mu^{*}$ such that: 
\begin{equation*} \mu(k) = \frac{1}{4k} \mu^{*}. \end{equation*}  Making this substitution allows us to differentiate the objective with respect to $k$ while ensuring the attention constraint still holds. Doing so yields a characterization of optimal team size. 

This observation allows us to derive sharp predictions about the relationship between complexity and team size. I focus on two potential interpretations of ``complexity'' in the main model. The first is in terms of the attention capacity $\rho$. The idea is that greater complexity might reflect a lower value of $\rho$, i.e., less processing capacity on behalf of headquarters.

The following result says that under this interpretation, greater complexity is associated with larger teams. Recall that $f'(1) <1/2$ is required for team size to be greater than 1 when headquarters cannot acquire any information about the team-state process:  

\begin{theorem} \label{thm:optimalk}
Suppose $f$ is convex with $0 < f'(1) < 1/2$. There exists $\bar{\rho}$ such that the organization's optimal team size is strictly decreasing in $\rho$ for all $\rho < \bar{\rho}$ and equal to 1 for all $\rho \geq \bar{\rho}$.  
\end{theorem}

\noindent This result, that more attention capacity corresponds to lower team sizes, is consistent with the stated justification for the reorganization of Block highlighted in the introduction. It is also broadly consistent with evidence from \cite{bloom_garicano_sadun_vanreenen_2014_distinct} that information technologies can increase local autonomy by lowering the cost of acquiring and processing relevant information, if smaller teams are interpreted as preserving greater local autonomy. And indeed,  the conclusion of Theorem \ref{thm:optimalk} is straightforward once the simplifications outlined in Section \ref{sec:spectral} are made. As highlighted in Section \ref{subsec:noattention}, the organization's choice of team size essentially solves a simple optimization problem involving a tradeoff between the adaptation and coordination loss: Adaptation loss increases in team size, but coordination loss decreases. Attention has no impact on the costs of increased team size (adaptation), but it does influence the benefits. In particular, the \emph{marginal} benefit is lowered by a multiplicative constant, specifically: 

\begin{equation} 
\frac{n\mu^{*} +\sum_{i=n+1}^{\infty} \lambda_{i}}{\sum_{i=1}^{\infty} \lambda_{i}}.  \label{eq:factor}
\end{equation}

\noindent In particular, $\mu^{*}$ corresponds to the posterior variance of the ``task state process'' eigenvalues for which the corresponding $Z_{j}$ term in the Karhunen-Lo\'eve representation shows up in the organization's attention problem. Hence, it is larger than all $\lambda_{i}$ for $i > n$ and weakly less than all $\lambda_{i}$ for $i \leq n$, so that indeed this constant is less than 1. Furthermore, as $\rho \rightarrow \infty$, $\mu^{*} \rightarrow 0$ and $n \rightarrow \infty$. Thus, $\bar{\rho}$ is the value for $\rho$ such that the corresponding value for $\mu^{*}$ and $n$ satisfies: 

\begin{equation*} 
\frac{n\mu^{*} +\sum_{i=n+1}^{\infty} \lambda_{i}}{\sum_{i=1}^{\infty} \lambda_{i}} = 2f'(1).
\end{equation*}

To discuss the relationship between \emph{horizontal complexity} and team size more directly, I turn to the underlying stochastic process that defines $W_{0}(x)$. While it may be intuitive that the degree of horizontal complexity should be reflected in the covariance operator of the Gaussian process---since this object determines how correlation depends on location---it may not be immediately obvious how precisely to do so. To be clear, the correct notion may very well depend on application or context; my only goal is to illustrate that different predictions may emerge depending on the precise origin of the horizontal complexity of interest. In a survey of the literature on correlated learning using Gaussian process realization as the underlying state, \cite{BardhiCallander2026LearningCorrelatedWorld} associate this with the volatility of the process. While this form of complexity does \emph{not} influence team size, varying individual eigenvalues in the Karhunen-Lo\'eve representation can. 

To make this precise, I consider perturbations of the covariance operator that hold the eigenbasis fixed. Specifically, fixing a covariance operator $C_{0}(x,x')$ with eigenvalue-eigenfunction pairs $(\lambda_{i}, \phi_{i}(x))_{i=1}^{\infty}$, and fixing $i$, consider covariance operators of the form: \begin{equation*} 
C_{i,t}(x,x')=C_{0}(x,x')+t\phi_{i}(x)\phi_{i}(x'), \end{equation*} where $t$ is in a neighborhood of $0$ such that $\lambda_{i} + t \geq 0$. Notice that under $C_{i,t}$, the eigenvalue associated with $\phi_{i}$ is $\lambda_{i} + t$, while all other eigenvalues and the chosen eigenfunctions remain unchanged. Denote the optimal team size under $C_{i,t}$ by $k_{i}^{*}(t)$. 

\begin{theorem} \label{prop:compstat}
Suppose $f(k)$ is convex and fix $\rho\in(0,\bar{\rho})$,
where $\bar{\rho}$ is defined in Theorem \ref{thm:optimalk}.
For the perturbation $C_{i,t}$ defined above, there exists a cutoff
$\bar{\lambda}$ such that
\[
\left.\frac{d k_i^*(t)}{dt}\right|_{t=0}>0
\quad\text{if }\lambda_i<\bar{\lambda},
\]
and
\[
\left.\frac{d k_i^*(t)}{dt}\right|_{t=0}<0
\quad\text{if }\lambda_i>\bar{\lambda},
\]
with derivative zero if $\lambda_i=\bar{\lambda}$.
\end{theorem}
\noindent The case where $\lambda_i$ is exactly at this threshold is a knife-edge
case in which the first-order effect is zero. Since eigenvalues are decreasing
in the index, and since the proof shows that $\lambda_{1} > \bar{\lambda}$ always holds (while $\lambda_{i} \rightarrow 0$ implies all sufficiently late eigenvalues are below $\bar{\lambda}$), the result implies that the leading eigenvalue decreases team size when it grows, while sufficiently late eigenvalues increase team size when they grow. It is worth emphasizing that the conclusion that varying individual eigenvalues can influence team size requires $\rho$ itself to be intermediate: At the two degenerate extremes where the inattention problem does not matter, team size is independent of $\lambda_{i}$. When $\rho=\infty$, this holds because there is simply no reason to have nondegenerate teams. But as illustrated in Section \ref{subsec:noattention}, the process variance enters the organization's objective multiplicatively when $\rho=0$, so that individual eigenvalues do not influence team size there either. 

Theorem \ref{prop:compstat} holds because $\lambda_{i}$ influences team size differently depending on whether it is part of the attention problem. Showing this in general is somewhat involved, but it is straightforward to see in the case where only $Z_{1}$ receives any attention, in which case (\ref{eq:factor}) can be written: 

\begin{equation*}
\frac{e^{-2\rho}\lambda_{1} +\sum_{i=2}^{\infty} \lambda_{i}}{\sum_{i=1}^{\infty} \lambda_{i}}. 
\end{equation*}

\noindent The fact that $\lambda_{1}$ shows up in the attention problem ``dampens'' (but does not eliminate) its influence on coordination loss. Since the impact on adaptation is unchanged, this factor decreases. More generally, whether a given change in volatility generates an increase or decrease in team size depends on how that change influences the attention budget, and in particular whether it enters strongly enough into the attention problem. This ambiguity is consistent with the broader message of \cite{dessein_lo_minami_2022_coordination} that the organizational implications of volatility depend on how volatility interacts with coordination. In the present model, the relevant distinction emerges due to the spectral properties of $C(x,x')$: away from the degenerate cases, increasing $\lambda_{i}$ for sufficiently high $i$ raises team size, whereas such a change can lower it if $\lambda_{i}$ enters headquarters' attention problem. 

While Theorem \ref{prop:compstat} shows that changes in \emph{individual} eigenvalues can change team size in a direction that depends on the attention level, the same techniques yield a derivation for the necessary and sufficient conditions for team size to increase given a change in the state process:

\begin{prop} \label{prop:robustcomplexity}
Fix an orthonormal basis $(\phi_i)_{i=1}^{\infty}$, and let
$\Lambda:[0,\varepsilon] \rightarrow \ell^1_{+}$ be continuously
differentiable as a map into $\ell^1$, where
\[
\Lambda(t)=(\lambda_1(t),\lambda_2(t),\ldots)
\]
is nonzero and weakly decreasing in $i$ for every $t$. Suppose further that $C_{t}(x,x') := \sum_{i=1}^{\infty} \lambda_{i}(t) \phi_{i}(x) \phi_{i}(x')$ is a continuous covariance kernel satisfying the feasibility assumptions for every $t \in [0, \varepsilon]$. Assume further that there is a uniquely optimal interior team size at all $t \in [0,\varepsilon]$. Let $n$ denote the number of modes in headquarters' attention problem at $t=0$\footnote{Recall that, as shown previously, this parameter is independent of $k$.} and define $\frac{1}{4k}\mu^{*}(0)$ as the water level at $t=0$ when team size is $k$. Team size is weakly increasing in $t$ at $t=0$ if and only if: 

\begin{equation} 
\sum_{i=1}^{\infty} \lambda_{i}'(0) \left( \min \left \{1, \frac{\mu^{*}(0)}{\lambda_{i}(0)} \right \} - \frac{n\mu^{*}(0) +\sum_{m=n+1}^{\infty} \lambda_{m}(0)}{\sum_{m=1}^{\infty} \lambda_{m}(0)} \right)\geq 0, \label{eq:NSCondition}
\end{equation}

\noindent and weakly decreasing if the reverse inequality holds. 
\end{prop}

\noindent While stronger than the condition provided in Theorem \ref{prop:compstat}, this conclusion does require solving for the value of $\mu^{*}$ and is hence potentially harder to interpret. At the same time, Proposition \ref{prop:robustcomplexity} facilitates the derivation of \emph{robust} comparative statics: i.e., conditions under which team size increases for \emph{every} attention level.

To do this, notice that the expression simplifies if $\sum_{m=1}^{\infty} \lambda_{m}$ is fixed. In that case, $\sum_{i=1}^{\infty} \lambda_{i}'(0)=0$ and (\ref{eq:NSCondition}) becomes: 

\begin{equation*} 
\sum_{i=1}^{\infty} \lambda_{i}'(0) \left( \min \left \{1, \frac{\mu^{*}}{\lambda_{i}(0)} \right \}\right)\geq 0.
\end{equation*}

\noindent A simple sufficient condition that ensures this condition automatically holds is that there exists $j^{*} \geq 1$ such that $\lambda_{i}'(0) < 0$ for all $i \leq j^{*}$ and $\lambda_{i}'(0) > 0$ for all $i > j^{*}$. Indeed, since $\min \left \{1, \frac{\mu^{*}}{\lambda_{i}(0)}  \right\}$ is weakly increasing in $i$, this condition gives us:

\begin{equation*} 
\sum_{i=1}^{\infty} \lambda_{i}'(0) \left( \min \left \{1, \frac{\mu^{*}}{\lambda_{i}(0)} \right \}\right) \geq \sum_{i=1}^{\infty} \lambda_{i}'(0) \left( \min \left \{1, \frac{\mu^{*}}{\lambda_{j^{*}}(0)} \right \}\right)=\left( \min \left \{1, \frac{\mu^{*}}{\lambda_{j^{*}}(0)} \right \}\right)\sum_{i=1}^{\infty} \lambda_{i}'(0)  =0,
\end{equation*}

\noindent since replacing $ \min \left \{1, \frac{\mu^{*}}{\lambda_{i}(0)} \right \}$ with $ \min \left \{1, \frac{\mu^{*}}{\lambda_{j^{*}}(0)} \right \}$ puts (weakly) more weight on negative terms and (weakly) less weight on positive terms.  \emph{Any} change in the eigenvalue sequence that decreases earlier eigenvalues but increases later eigenvalues, fixing total variance, will increase team size.

What do these comparative statics correspond to economically in terms of the comparative statics of how team size changes with horizontal complexity? One way to think about this is in terms of the \emph{jaggedness} of the sample-path realizations. I illustrate this idea in Figure \ref{fig:kl:three}. Notice that sample-path realizations appear to vary more moving across panels from left to right. These changes are generated by successively decreasing the \emph{eigenvalue decay rate}, using the stationary Ornstein-Uhlenbeck process in the middle panel.

\begin{figure}[t]
    \centering

    \begin{subfigure}[t]{0.32\textwidth}
        \centering
        \includegraphics[width=\linewidth]{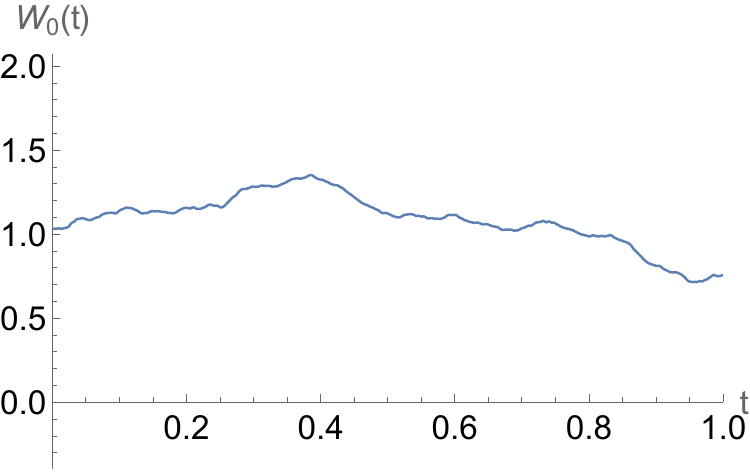}
        \caption{$\lambda_{k}$ decays like $ k^{-3}$}
        \label{fig:kl:1}
    \end{subfigure}\hfill
    \begin{subfigure}[t]{0.32\textwidth}
        \centering
        \includegraphics[width=\linewidth]{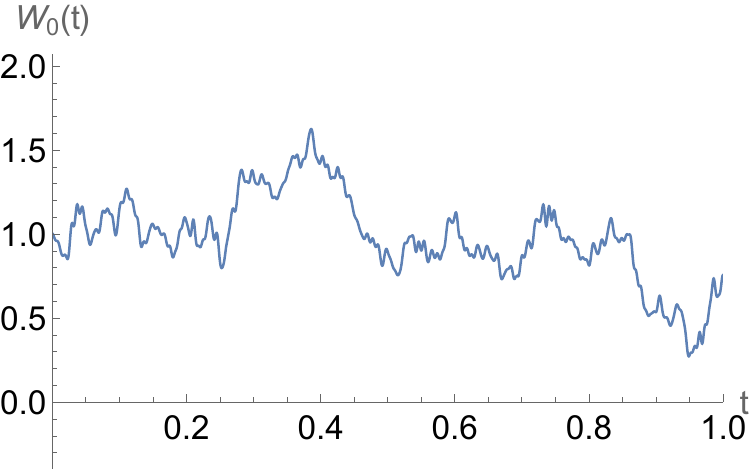}
        \caption{$\lambda_{k}$ decays like $k^{-2}$}
        \label{fig:kl:2}
    \end{subfigure}\hfill
    \begin{subfigure}[t]{0.32\textwidth}
        \centering
        \includegraphics[width=\linewidth]{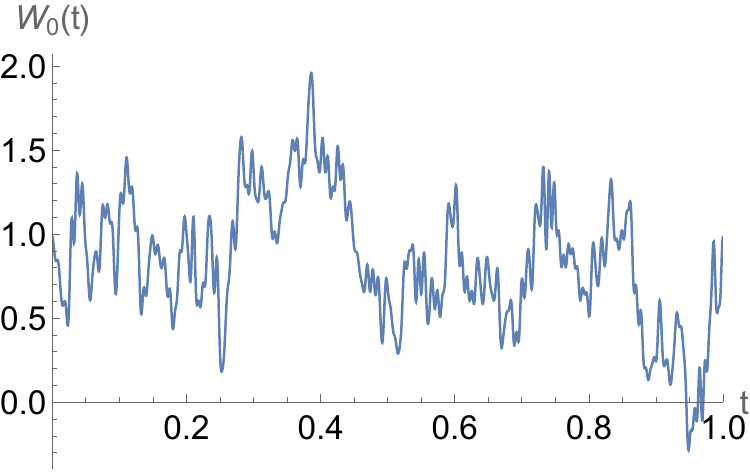}
        \caption{$\lambda_{k}$ decays like $k^{-3/2}$}
        \label{fig:kl:3}
    \end{subfigure}

    \caption{\footnotesize Illustration of the influence of eigenvalue decay on local complexity. In each case, eigenvalues are normalized to keep
$\int \mathbb{E}[W_0(x)^2]dx$ constant, and eigenfunctions corresponding to the stationary Ornstein-Uhlenbeck process are used. 
For illustrative purposes, each plot uses identical realizations for the 
$Z_k$ sequence, so differences in the sample paths are entirely due to the eigenvalues. Panel
(\ref{fig:kl:2}) uses the finite-interval Ornstein-Uhlenbeck eigenvalues
$\lambda_k=1/(1+\omega_k^2)$, where the frequencies $\omega_k$ solve the
Ornstein-Uhlenbeck boundary condition; these eigenvalues decay asymptotically at
rate $k^{-2}$. Panels (\ref{fig:kl:3}) and (\ref{fig:kl:1}) use normalized
powers of these eigenvalues, producing slower and faster decay, respectively. \normalsize }
    \label{fig:kl:three}
\end{figure}

The figure illustrates visually how a decrease in the eigenvalue decay rate makes the sample paths appear more jagged. Appendix \ref{app:jaggedness} provides conditions on the eigenbasis which allow this relationship to be stated formally. One way to quantify jaggedness is to consider the parameter:

\begin{equation} 
\alpha^{*}(x) := \sup \{ \alpha : \limsup_{h \rightarrow 0} \mathbb{E}[(W_{0}(x+h)-W_{0}(x))^{2}]/h^{\alpha} < \infty \}.\label{eq:complexitydef}
\end{equation}

\noindent To interpret this, notice that standard Brownian motion has $\mathbb{E}[(W_{0}(x+h)-W_{0}(x))^{2}] \propto h$, so that $\alpha^{*}(x)=1$. This also holds for the stationary Ornstein-Uhlenbeck process. A natural requirement for more jaggedness is that $\mathbb{E}[(W_{0}(x+h)-W_{0}(x))^{2}]$ approaches 0 \emph{more slowly} as $h \rightarrow 0$---so that the increments have larger variance even for locations that are even closer. Thus, $\alpha^{*}(x)$ being smaller suggests the sample-paths feature more jaggedness. Under a condition that roughly speaking requires higher-index eigenfunctions to oscillate more than lower-index eigenfunctions, $\alpha^{*}(x)$ is strictly increasing in $p$ provided it is less than 2; $\alpha^{*}(x)$ remains equal to 2 once this upper bound is reached, reflecting that this measure does not distinguish further increases in smoothness. The formal condition requires a technical detour that I defer to Appendix \ref{app:jaggedness}. 

The punchline is that changing the decay rate, at least within a parameterized family, gives one way of seeing how increasing jaggedness corresponds to larger teams. Specifically, decreasing the rate of eigenvalue decay while holding variance fixed will satisfy the hypothesis of Proposition \ref{prop:robustcomplexity} following the subsequent discussion, since the condition that $\lambda_{i}'(0)$ crosses 0 once will generally hold for such changes.\footnote{More concretely, this will always hold if $\lambda_{k}(t)=\gamma(t)/k^{s-t}$ where $\gamma(t)$ is such that $\sum_{k} \lambda_{k}(t)$ is fixed. Indeed, $\lambda_{k}'(t)/\lambda_{k}(t) = \log(k) + \frac{\gamma'(t)}{\gamma(t)}$ is strictly increasing in $k$. Since total variance is fixed, $\lambda_{1}'(t) < 0$, since otherwise monotonicity in $k$ would imply that the variance were increasing instead of constant. Since this expression is eventually positive as $k \rightarrow \infty$,  $\lambda_{k}'(t)$ switches signs exactly once.  } But under the additional eigenfunction conditions in Appendix \ref{app:jaggedness}, the same change will also weakly lower $\alpha^{*}(x)$, and strictly so before $\alpha^{*}(x)$ reaches its upper bound of 2. Thus, along such spectral changes, greater jaggedness is associated with weakly larger teams. 

This discussion helps illustrate why not all changes in volatility should have the same impact on team size. In some cases, the volatility of the process may increase, but do so in a way that actually makes nearby state realizations \emph{more correlated in a relative sense}---with the added variance concentrated on the dimension associated with the leading eigenvalue. In other cases, volatility may decrease, but do so in a way that makes nearby states more dissimilar, hence increasing jaggedness and potentially increasing team size. Of course, an increase in the decay rate that \emph{also} increased $\lambda_{1}$ by a significant enough margin could decrease team size. These caveats aside, this discussion underscores that predictions about how team size changes with horizontal complexity depend on the anatomy of such changes. While the formal conditions depend both on how the eigenvalues change and on properties of the associated eigenfunctions, the economic interpretation of these changes in terms of jaggedness is illustrated in Figure \ref{fig:kl:three} and could in principle be empirically measured using other variables to proxy for action correlation.

\section{Other Organizational Margins}

So far, I have outlined the basic relationship between environmental complexity and organizational design that emerges in the presence of attention constraints. The goal is to illustrate how horizontal complexity shows up in other aspects of organization design, since in practice, organizations have other margins for responding to complexity. Specifically, I show how the framework provided speaks to richer hierarchies (beyond team size), investments that either influence the task structure or complexity itself, and the tradeoff between division merger and separation. The role of these extensions is to show how the same spectral properties underlying the main analyses also help organize the analysis of these other margins. Doing so facilitates comparisons to other work in the organizational economics literature, both theoretical and empirical.

\subsection{Vertical Hierarchy}  \label{sec:hierarchies}

I now enrich the model to allow for more than 2 layers in the organizational hierarchy, and provide a formal sense in which the forces that push toward larger teams also favor more hierarchical layers. I do this keeping the baseline model fundamentally unchanged, in the sense that: (1) Headquarters seeks to coordinate at the ``top layer'' of the organization, while (2) Coordination loss depends on how well actions are tailored to the ``ground layer.''

To see how to naturally expand the model to allow for richer hierarchy, first consider the case where one additional layer is added. If ground layer workers are in teams of size $k_{1}$ and teams are in groups of size $k_{2}$, then the covariance operator becomes: 
$\frac{1}{4}C(k_{1}k_{2}x,k_{1}k_{2}x')$---the ``top layer'' process scales the ``middle layer'' process by $k_{2}$, with the middle layer scaling the ground layer by $k_{1}$ (as before). 

Why might richer hierarchies help? An immediate observation is that if $f(k_{1},k_{2})=g(k_{1}k_{2})$---i.e., $k_{1}$ and $k_{2}$ interact multiplicatively---then if $k_{1}$ is chosen optimally, $k_{2}=1$ is optimal: any value for the objective achieved by some choice of $k_{1}$ and $k_{2}$ can also be achieved by choosing $k_{1}$ alone. But consistent with the theme in the previous section, I have in mind that the coordination loss scales according to the inverse team size, whereas the adaptation loss scales linearly: what matters for adaptation is (i) the number of teammates per team and (ii) how many teams are grouped together. Here, I show that this feature also provides scope for more ``vertical'' organizations.

I start with the \emph{short-run} hierarchy design problem. Specifically, I assume that the organization has a ground layer whose size is fixed at some $k_{1}$. While I allow $k_{1}$ to be arbitrary, of particular interest is the case where $k_{1}$ is the optimal team size in the baseline model. On top of this ``lower layer,'' the organization can add a ``higher layer,'' with endogenously-chosen group size $k_{2}$. I call this the short-run problem because $k_{1}$ cannot adjust, so that all that can be done is group teams together---for instance, by introducing a ``hub'' which groups teams together, so that hubs would be chosen without adjusting the routines in the ground layer.  

I obtain a striking result in the case where $f(k_{1},k_{2})=1+\beta(k_{1}-1)+ \gamma (k_{2}-1)$, where $k_{1}, k_{2} \geq 1$. I take $\gamma \geq \beta$: 

\begin{prop} \label{thm:hierarchy}
Let $k_{1}$ be exogenously fixed, with the higher layer endogenously determined. Then provided $k_{2} > 1$, $k_{2} \propto \frac{1}{\sqrt{k_{1}}}$. Furthermore, if $k_{1}$ solves the one-layer problem and $\gamma=\beta$, then   $k_{1} \geq k_{2}$, with strict inequality if $k_{1} > 1$ (i.e., organizations are ``bottom-heavy''). 
\end{prop}

Call $k_{2}^{SR}$ the optimal higher-layer team size when ground-layer teams are fixed to be the optimal team size from the baseline model (i.e., with a single layer), and $k_{2}^{LR}$ the optimal higher-layer team size when ground-layer teams adjust. In other words, in the long run, the optimal team size may itself depend on the presence of the higher layer of the hierarchy. The following result shows how: 

\begin{prop} \label{prop:longrun}
If the first layer is endogenous, then adding a higher layer to the hierarchy decreases the size of the lower layer. Furthermore, $k_{2}^{SR} >1$ if and only if $k_{2}^{LR} > 1$. 
\end{prop}

\noindent Taken together, these results identify some differences between the short-run and long-run outcomes of adding layers to the hierarchy. In the short run, adding a layer may yield gains of the same kind introduced previously. However, the result of adding this layer is to increase the efficiency, meaning that the lower layer need not be as large.  Thus, the lower layer scales down while the higher layer scales up.

Thus, the long-run versus short-run distinction does not influence \emph{whether} it is beneficial to add a layer to the hierarchy, although it does influence the shape of the organization. Not only this, but an immediate corollary of Propositions \ref{thm:hierarchy} and \ref{prop:longrun} is that the shape of the organization responds to complexity in the short run problem, but not the long run problem. In the long run, $k_{1}/k_{2}$ is determined entirely by $\beta$ and $\gamma$. In the short run problem when $k_{1}$ is the optimal one-layer team size, however, $k_{1}/k_{2}$ scales proportionally to $\left(\frac{n\mu^{*}+ \sum_{m=n+1}^{\infty} \lambda_{m}}{2\sum_{m=1}^{\infty} \lambda_{m}} \right)^{1/4}$. Recall that the ratio $\frac{n\mu^{*}+ \sum_{m=n+1}^{\infty} \lambda_{m}}{2\sum_{m=1}^{\infty} \lambda_{m}}$ also determines the size of the ground layer teams in the main model problem; thus, if teams are larger in the single-layer model, then more teams are grouped together with an additional layer. 

I now turn to the problem of allowing for even richer layers of hierarchy, focusing on the firm's long run solution, setting $k=\prod_{i=1}^{j} k_{i}$ in the covariance operator when there are $j$ layers. In this linear case above, there is always a positive gain to adding more layers of hierarchy: It is more efficient to increase group size at a higher layer starting at size 1 than at a lower layer. To speak to ``depth'' precisely, the following proposition assumes that if the firm adds a hierarchy layer, then it must involve $k_{i} \geq 1+\varepsilon$, for $\varepsilon > 0$; alternative microfoundations such as costs could be used to produce similar conclusions.

\begin{prop} \label{prop:generallayers}
When the number of hierarchical layers is endogenous, the optimal number of layers (including the ground layer) is: 

\begin{equation*} 
\max \left\{\left\lfloor \frac{\log\left(\beta/\gamma^{2} \right)+ \log \left(\frac{n \mu^{*} + \sum_{m=n+1}^{\infty} \lambda_{m}}{2\sum_{m=1}^{\infty} \lambda_{m}} \right)}{\log (1+\varepsilon)} -1 \right\rfloor, 1  \right \}. 
\end{equation*}
\end{prop}

\noindent Thus, while team size at a given layer shrinks as more layers are added, Proposition \ref{prop:generallayers} clarifies that the same statistic which determines how attention influences team size---i.e., (\ref{eq:factor})---also determines how attention influences the number of layers, with both team size and the number of layers increasing in this statistic. This mapping provides a way of interpreting findings on technological change and hierarchy depth, instead of or in addition to team size. For instance, \cite{babina_fedyk_he_hodson_2025_firm} finds that investments in artificial intelligence are associated with flatter organizational hierarchies, which in the language of the model is consistent with AI reducing the attention burden faced at higher organizational layers.\footnote{Such changes need not uniquely influence the attention channel: For instance, AI may also change the adaptation problem faced by workers. As an example, \cite{brynjolfsson_li_raymond_2025_generative} show that generative-AI assistance reduces customer requests to speak with a manager, suggesting a task-level change. I defer a description of how my model speaks to the gains from changing the state process in Section \ref{sec:endocomplex} and the adaptation problem in Section \ref{sec:otherdesign}.}

\subsection{Choosing the Complexity of Production} \label{sec:endocomplex}

In practice, organizations may have control over the kinds of problems they face, and hence the form of the relevant stochastic process. As an example, \cite{brynjolfsson_li_raymond_2025_generative} study a generative-AI assistant for customer-support agents, finding that the tool is most valuable for relatively rare problems and that AI assistance reduces customer requests to speak with a manager. Interpreted through the present framework, while such technologies do not necessarily change the \emph{primitive} distribution of tasks, they can reduce the \emph{residual} variation in the locally optimal actions that enters the coordination problem.  Here, I explore an extension in which the organization can choose this effective process and study how the coordination problem shapes that choice.

Specifically, I allow the organization, at a cost, to change the distribution of the Gaussian process $W_{0}(x)$ to be one that is relatively more favorable. To allow for a flexible formulation, I assume that this cost is proportional to the Kullback-Leibler Divergence between the ``baseline'' stochastic process and one that is faced, i.e., $\eta D(Q \| P)$, where $P$ is the initial process (with covariance operator $C(x,x')$, as before) and $Q$ is a mean-zero Gaussian process chosen by the organization. I call $Q$ feasible if its covariance kernel is continuous and positive definite and if the integral of the variance over all locations $x \in [0,1]$ is finite, as is assumed for $P$. 

The form of the KL-Divergence turns out to be particularly tractable. The first observation is that $Q$ must be absolutely continuous with respect to $P$, as otherwise the organization's objective would be $-\infty$ whereas setting $Q=P$ would lead to bounded losses. But any such $Q$ admits a representation where sample paths are $\sum_{j=1}^{\infty} \sqrt{\tilde{\lambda}_{j}} \tilde{\phi}_{j}(x) Z_{j}$, for some eigenvalues $(\tilde{\lambda}_{j})$ and eigenfunctions $(\tilde{\phi}_{j}(x))$. Under this formulation, there is a particular tractable representation for the KL-divergence, simplifying the solution to the organization's problem even further:

\begin{lemma} \label{lem:endoKL}
There exists an optimal feasible process whose Karhunen-Lo\'eve eigenfunctions satisfy $\tilde{\phi}_{k}(x)=\phi_{k}(x)$. For such a process with chosen eigenvalues $(\tilde{\lambda}_{j})_{j=1}^{\infty}$:
\begin{equation*} 
D(Q \| P) = \frac{1}{2} \sum_{j=1}^{\infty}  \left(\frac{\tilde{\lambda}_{j}}{\lambda_{j}}-1- \log \left( \frac{\tilde{\lambda}_{j}}{\lambda_{j}} \right) \right)
\end{equation*}

\end{lemma}

\noindent In words, the choice of $Q$ involves the same eigenfunctions as the original process, but with each \emph{eigenvalue} modified to be more favorable to the organization. Using this result, I can solve for the organization's optimal choice of process: 

\begin{proposition}  \label{thm:endocomplex}
Suppose $k$ is fixed. There exists $\bar{\lambda}$ such that: 

\begin{equation*} 
\tilde{\lambda}_{i} = \frac{\eta-2\mu}{\frac{\eta}{\lambda_{i}} + f(k)}  \text{ when } \lambda_{i} > \bar{\lambda}, \end{equation*} 
and: 
\begin{equation*}
\tilde{\lambda}_{i} = \frac{1}{\frac{1}{\lambda_{i}} + \frac{2}{\eta} \left(\frac{1}{4k} + \frac{f(k)}{2} \right)}  \text{ when } \lambda_{i} < \bar{\lambda}.
\end{equation*}

\noindent On the other hand, if $k$ is chosen optimally, then $k$ satisfies: 

\begin{equation} 
\frac{\mu n}{k}+\frac{1}{4k^{2}} \sum_{i=n+1}^{\infty} \tilde{\lambda}_{i} -\frac{f'(k)}{2} \sum_{i=1}^{\infty} \tilde{\lambda}_{i} =0, \label{eq:endocompk}\end{equation}

\noindent where $\mu$ solves:

\begin{equation*} \rho = \frac{1}{2} \sum_{m=1}^{n} \log \frac{\frac{1}{4k}  \tilde{\lambda}_{m}}{\mu}.\end{equation*}
\end{proposition}

\noindent Proposition \ref{thm:endocomplex} shows how the attention problem interacts with the choice of complexity. First, higher values of $\lambda_{i}$ translate into higher values of $\tilde{\lambda}_{i}$, as expected.\footnote{The proof shows that $\eta-2\mu > 0$ must hold in the solution, so indeed $\tilde{\lambda}_{i} > 0$ when $\lambda_{i} > \bar{\lambda}$.} On the other hand, for a fixed baseline process with eigenvalues $(\lambda_{i})_{i=1}^{\infty}$, an increase in the attention budget either has no impact (for $\lambda_{i} < \bar{\lambda}$) or leads the organization to lower the corresponding eigenvalue $\tilde{\lambda}_{i}$ by less than otherwise. Thus, attention substitutes for investment in this formulation. 

Proposition \ref{thm:endocomplex} also illustrates that investment has the impact of \emph{flattening} the eigenvalue sequence: For any pair of eigenvalues in the head with $\lambda_{i+1} < \lambda_{i}$, $\lambda_{i+1}/\lambda_{i} < \tilde{\lambda}_{i+1}/\tilde{\lambda}_{i}$, and the same holds for any pair of eigenvalues in the tail. We also see that $\tilde{\lambda}_{i}$ is, perhaps surprisingly, bounded above even as $\lambda_{i} \rightarrow \infty$. Thus, large differences for $\lambda_{i}$ in the head translate into much smaller differences when endogenously chosen. The impact on the tail is less dramatic---an observation which, combined with the discussion in Section \ref{sec:teamsize}, suggests that the impact of a constant-factor volatility increase should be concentrated in the tail and hence lead to hierarchy expansion. The next result shows that this is indeed the case.

\setcounter{theorem}{3}

\begin{subtheorem}  \label{prop:gamma}
Suppose the baseline process $P$ is a mean-zero Gaussian process with covariance operator $\gamma C_{0}(x,x')$,  where $\gamma > 0$ and $C_{0}(x,x')$ is fixed with eigenvalues $\lambda_{1}, \lambda_{2}, \ldots$. Then the optimal team size as a function of $\gamma$, $k^{*}(\gamma)$, is increasing in $\gamma$, converging to the zero-attention capacity solution as $\gamma \rightarrow \infty$. Suppose, in addition, there is some $\underline{c}>0$ such that $\lambda_{i+1} \geq \underline{c}\lambda_{i}$, and let $n(\gamma)$ denote the number of variables that enter into the attention problem. Then $\lim_{ \gamma \rightarrow \infty} \tilde{\lambda}_{1}(\gamma)= \lim_{\gamma \rightarrow \infty} \tilde{\lambda}_{n(\gamma)+1}$, where $\tilde{\lambda}_{i}(\gamma)$ denote the eigenvalues of the optimally chosen process $Q$. 
\end{subtheorem}

\noindent Under the provided condition---which essentially ensures that eigenvalues do not decay too quickly even as $\gamma \rightarrow \infty$---the interaction between the attention problem and the process choice problem is quite stark, with all eigenvalues in the attention becoming equalized and converging to the water level. Even though the constant-factor increases in volatility make the eigenvalues more spread out in magnitude, the \emph{opposite} impact is observed for eigenvalues, at least in the head. 

An analogous comparative static holds as the cost of modifying the state process becomes arbitrarily small:

\begin{subtheorem}  \label{prop:eta}
The optimal team size is decreasing in $\eta$, and converges to the zero-attention capacity solution as $\eta \rightarrow 0$. 
\end{subtheorem} 

\noindent The proof idea behind Theorem \ref{prop:eta} closely tracks Theorem \ref{prop:gamma}, with an increase in $\eta$ roughly as having a symmetric effect to a decreasing in $\gamma$. The intuition is that as $\eta$ becomes small, the impact is concentrated on the eigenvalues that are in the attention problem, whereas those outside of it remain less impacted. 

The comparison between Theorem \ref{prop:eta} and Theorem \ref{thm:optimalk} illustrates that the impact of new technologies that improve organizational efficiencies may depend on where these efficiencies are concentrated. When concentrated at the management of headquarters, the main effect is to enable smaller teams by making it easier to monitor the organization. But when this technology makes the tasks themselves more manageable, the implication can be that the remaining variance itself is less concentrated in the attention problem, so that hierarchy should expand as a result. 

This dual-impact is highly reminiscent of the findings of \cite{bloom_garicano_sadun_vanreenen_2014_distinct}, namely that information technology and communication technologies can have opposite impacts on hierarchy. The take-away that my framework delivers is that the impact of any such technology on organizational hierarchy ultimately depends on where the residual complexity lies, i.e., inside or outside of headquarters' attention problem. While a lower value of $\eta$ reflects a better ability to exogenously simplify the state process, this simplification is more concentrated on the dimensions of the state process that already enter into the attention problem, intuitively since these were already more significant to begin with. The conclusion is that the apparently contradictory effects of technological improvements in organizations can be resolved by considering where the remaining horizontal complexity lies. 

\subsection{Decentralization versus Control} \label{sec:divisions}

The discussion of richer hierarchies in Section \ref{sec:hierarchies} assumed that higher layers inherit the same basic structure as lower layers, reflecting the case where teams and subteams do not split the organization itself.  This does not capture a distinct organizational margin: whether to separate a discrete set of divisions as opposed to integrating them under a common hierarchy.  Consider the Manhattan Project as an illustration.  Its activities were dispersed primarily across three sites---Oak Ridge, Hanford, and Los Alamos---each responsible for distinct aspects of the project.  These different locations reflected natural boundaries within the project: tasks within the same site were typically more closely related to one another than to tasks elsewhere,  despite all still remaining part of the broader interdependent production process. 

In this section,  I show how this paper's framework also speaks to the question of when firms should optimally set up boundaries between divisions,  as commonly occurs in practice and as it did with the Manhattan Project.  To maintain simplicity,  I assume that there are two different domains which may exhibit a particular form of correlation.  Implicitly,  there is a natural ``wedge'' between the tasks within each division,  and the question the organization faces is whether to separate divisions along this wedge.

I now describe this extension formally.  The tasks for the first division are located on $A=[0,1]$ while the tasks for the second division are located on $B=[2,3]$.  I start with the assumption that each interval has its own covariance operator---i.e.,  the covariance between $x, x' \in [0,1]$ is $C_{[0,1]}(x,x')$,  while the covariance operator for $x,x' \in [2,3]$ is $C_{[2,3]} (x,x')$.  For simplicity,  assume these covariance operators are supported on the subspaces orthogonal to the constant functions on $A$ and $B$ (although this assumption is not essential).  Let $\tilde{W}_{0}(x)$ denote the resulting stochastic process on $[0,1] \cup [2,3]$ with these covariance operators (with the correlation between $\tilde{W}_{0}(x)$ and $\tilde{W}_{0}(x')$ being 0 when $x$ and $x'$ are in different divisions) and set: 

\begin{equation*} 
W_{0}(x) = \tilde{W}_{0}(x)+ \theta_{A} \mathbf{1}[x \in [0,1]]+\theta_{B} \mathbf{1}[x \in [2,3]]
\end{equation*}

\noindent with $W_{0}(x)=0$ outside of $[0,1] \cup [2,3]$, $\begin{pmatrix} \theta_{A} \\ \theta_{B} \end{pmatrix} \sim \mathcal{N} \left( \begin{pmatrix} 0 \\ 0 \end{pmatrix}, \begin{pmatrix} \sigma^{2} & r \sigma^{2} \\ r \sigma^{2} & \sigma^{2} \end{pmatrix} \right)$ for $r \in [-1,1]$, and $\tilde{W}_{0}$ and $(\theta_{A}, \theta_{B})$ independent (except for the possible correlation between $\theta_{A}$ and $\theta_{B}$).  Thus,  the common shocks between the different domains are correlated---but within a domain,  no single location is uniquely more correlated with another location in another domain.  I allow for the divisions to choose team size independently.  Let $k_{A}$ denote the team size of division $A$ and $k_{B}$ denote the team size of division $B$.   Given the team choice,  the organization's adaptation loss can be defined as in the baseline model,  i.e.,  as the integral of each location's adaptation loss taken across all locations. 

The difference between the ``merged'' case and the ``separated'' case enters through the attention problem.  When divisions $A$ and $B$ are merged, the problem is exactly as described in the main model,  subject to the aforementioned modifications which influence the definition of $W_{1}(x)$.  I let $\rho$ denote the available attention budget as before,  and coordination loss is simply the integral of the difference between the expected team state and the realized team state.  

When divisions $A$ and $B$ are separated,  each division solves the attention problem individually.  Specifically,  headquarters first allocates attention between division $A$ and $B$ freely.   Letting $\rho_{A}$ and $\rho_{B}$ denote the attention budgets of each division,  respectively,  $\rho_{A}$ and $\rho_{B}$ can be any nonnegative amount subject to the constraint that $\tau_{A} \rho_{A} + \tau_{B} \rho_{B} = \rho$.  I assume that $\tau_{A}, \tau_{B} \leq 1$,  where $\tau_{A},  \tau_{B} <1$ reflects the presence of efficiency gains from separation.\footnote{The treatment of a division merger as changing headquarters' attention problem contrasts with \cite{dessein_garicano_gertner_2010_organizing}, where integration creates potential cost savings from standardization at the expense of potentially reducing local adaptation.} Division $j \in \{A, B\}$ chooses a stochastic process $S_{j}$ subject to the constraint that $I(W_{1} \lvert_{j}; S_{j}) \leq \rho_{j}$,  where $W_{1} \lvert_{j}$ denotes the restriction of $W_{1}$ to $j$. \footnote{Note that this abuses notation slightly since $W_{1} \lvert_{j}$ is not supported on $j$; I take $W_{1} \lvert_{A}$ to be supported on $[0,1/k_{A}]$ and $W_{1} \lvert_{B}$ to be supported on $[2,2+1/k_{B}]$.} Division $j$'s coordination loss is the integral of the squared difference between the realized team state and the expected team state given $S_{j}$.  The organization's coordination loss is the sum of the coordination losses for each division.

The following Lemma identifies the difference between the merged and separated problems in terms of the spectral representation: 

\begin{lemma}  \label{lem:mergedKL}
Suppose \(C_{[0,1]}\) has eigenvalues \((\lambda_k^A)_{k=1}^{\infty}\) and \(C_{[2,3]}\) has eigenvalues \((\lambda_k^B)_{k=1}^{\infty}\). Define
\[
\sigma_A^2=\frac{\sigma^2}{4k_A},
\qquad
\sigma_B^2=\frac{\sigma^2}{4k_B}.
\]
Let
\[
\lambda_+
=
\frac{\sigma_A^2+\sigma_B^2}{2}
+
\frac12
\sqrt{(\sigma_A^2-\sigma_B^2)^2+4r^2\sigma_A^2\sigma_B^2},
\]
and
\[
\lambda_-
=
\frac{\sigma_A^2+\sigma_B^2}{2}
-
\frac12
\sqrt{(\sigma_A^2-\sigma_B^2)^2+4r^2\sigma_A^2\sigma_B^2}.
\]
Then the eigenvalues in the Karhunen-Lo\'eve representation for the merged problem are
\[
\lambda_+,\ \lambda_-,\ 
\frac{\lambda_1^A}{4k_A},\frac{\lambda_2^A}{4k_A},\ldots,
\frac{\lambda_1^B}{4k_B},\frac{\lambda_2^B}{4k_B},\ldots .
\]
\end{lemma}
\noindent The new eigenvalues introduced in Lemma \ref{lem:mergedKL} are those that emerge under a standard principal components analysis exercise: They are the ``major principal component'' and the ``minor principal component'' of the common state problems. Note that in the case where $r=0$, these eigenvalues reduce to $\sigma_{A}^{2}$ and $\sigma_{B}^{2}$, the common variance associated with each of the individual divisions. 

\begin{proposition}  \label{thm:mergerbetter}
Suppose $\tau_A=\tau_B=1$. Then, for every $\rho \in \mathbb{R}_{+}$, the organization's losses under merger are weakly lower than under separation. The inequality is strict whenever $r\neq 0$, and $\lambda_{+}$ receives positive attention. In particular, if $\abs{r} \in (0,1)$, the inequality is strict for all sufficiently large finite $\rho$.
\end{proposition}

\noindent While the proof of Proposition \ref{thm:mergerbetter} is involved,  the intuition for why mergers help is straightforward: whenever there is correlation between the common state shocks, learning about one division's state provides information about the other's.  Thus, paying attention jointly allows the organization to avoid duplication,  achieving the same end without expending as much of the attention budget.  Paying attention to both simultaneously is then strictly more efficient. These gains are smaller if most of the attention is allocated to only one division in any case, as occurs when team sizes are asymmetric. It is also immediate that there are no gains in the case of $\rho=0$ and $\rho=\infty$.

This interpretation is related to \cite{dessein_lo_minami_2022_coordination}, which shows that volatility can favor centralization when coordination needs are high. In my model, the need to coordinate across divisions is captured by the presence of correlated division-wide shocks: when $r=0$, merger provides no attention advantage, while when $r\neq 0$, joint processing allows headquarters to adapt to the common component more efficiently.

Allowing $\tau_{A}, \tau_{B} < 1$ reflects efficiency gains from separation. A natural story for why this might be the case is if the merger requires added costs for divisions to coordinate with headquarters; this mechanism would in turn imply that decentralization should be favored when these costs are higher.  This observation relates this extension to findings of \cite{gumpert2022firm} on how ease of transportation between divisions influences organizational hierarchy.  It is also consistent with \cite{mcelheran2014delegation}, finding in the context of IT purchasing in multi-establishment firms that local information advantages favor delegation, while greater firm-wide coordination needs favor centralization.

I obtain the following result on how the organizational hierarchy changes in the presence of the merger, relating this model to the empirical findings of \cite{gumpert2022firm}: 

\begin{proposition}  \label{prop:teamsize}
Suppose $\tau_A=\tau_B=1$ and that the two divisions have the same eigenvalues and adaptation cost. Suppose further that $f$ is convex and that the decentralized optimum involves interior team sizes. If the merged problem has a symmetric optimum $k_A^M=k_B^M$, then the common team size under merger is weakly smaller than the common team size under separation. The inequality is strict whenever the organization does strictly better under the merger as per the conditions in Proposition \ref{thm:mergerbetter}. \footnote{The restriction to symmetric merged optima is substantive only because the merged objective need not be globally convex in $(k_{A},k_{B})$; in fact, it will typically fail to be for intermediate values of attention where $\lambda_{+}$, but not $\lambda_{-}$, is in the attention problem. However, by Lemma \ref{lem:teambound} in the Appendix, optimal team sizes lie in a compact set. If $\rho$ is sufficiently large so that both $\lambda_{+}$ and $\lambda_{-}$ receive attention throughout this set, this problematic region disappears and the merged objective is convex over the relevant region. Hence, symmetry of primitives implies the existence of a symmetric merged optimum in this region. The asymmetric example below shows why no coordinatewise monotonicity result should be expected without symmetry.}
\end{proposition}

This result mirrors the finding of \cite{gumpert2022firm} that longer travel time between headquarters and divisions is associated with more managerial layers. The result that team sizes decrease in the merger mirrors the intuition from the main model---less duplication changes the marginal benefit from larger teams in the same way as an increase in the attention budget,  which favors hierarchy contraction.  But importantly,  this same conclusion need not hold when team sizes are asymmetric. The Appendix presents a numerical example where team sizes are asymmetric and the merger leads to one division using \emph{larger} teams.  Intuitively,  attention that is paid to the ``joint'' modes influences both divisions' coordination loss; so if one division gets most of the attention budget because of asymmetries under separation,  then when attention is paid to a joint mode under centralization,  it is as if some of the budget goes \emph{away} from that division despite the efficiency gain.  While the reasoning behind Proposition \ref{thm:mergerbetter} is not sensitive to exact symmetry, the symmetry condition is nevertheless important for the conclusion to hold.

\subsection{Human Capital and Training}
\label{sec:otherdesign}

Some organizational choices reflect other augmentations of the production technology which are less clearly understood as a direct change of the state process itself. A notable example is \cite{adhvaryu2023organizational}, showing that training increases and hierarchy flattens following the introduction of a new product. 

I accommodate this possibility by allowing headquarters to choose a ``training level'' $\alpha$ at cost $\frac{1}{2} \alpha^{2}$, which influences both the cost of taking an action $a_{x}$ and also the effectiveness. A leading case of interest will be where training makes actions more difficult, as methods are unfamiliar (at least initially), but also increase the effectiveness of a given action. Thus, if the firm chooses training level $\alpha$, then the adaptation loss becomes: 

\begin{equation*} 
\ell(a_{x}, W_{0}(x), \alpha) = f(k) \cdot (c(\alpha) a_{x}^{2} +(\beta(\alpha) a_{x} - W_{0}(x))^{2}),
\end{equation*}

\noindent where $c(\alpha), \beta(\alpha) > 0$. In addition, consistent with the notion that training makes actions more effective, I also assume that $\beta'(\alpha) > 0$.

An illustrative case is when $c(\alpha)=1+\alpha$ and $\beta(\alpha)= \sqrt{1+\alpha}$, in which case $a_{x} = \frac{W_{0}(x)}{2\sqrt{1+\alpha}}$ and $l(W_{0}(x)) =  f(k)\frac{W_{0}(x)^{2}}{2}$. While the problem is unchanged when $\alpha=0$ (no training), choosing positive $\alpha$ allows for less variation in $a_{x}$ without lowering adaptation loss\footnote{The empirical results of \cite{adhvaryu_kala_nyshadham_2023_returns} motivate the assumption that training can improve coordination, as opposed to exclusively lowering privately-incurred action costs. In particular, \cite{adhvaryu_kala_nyshadham_2023_returns} interprets the gains due to improved teamwork from a soft-skills training program as substituting for managerial attention. The present framework captures this by allowing training to decrease the variation in worker actions; since lower residual variation decreases the marginal gain from a higher $\rho$, training and managerial attention naturally operate as substitutes in this case.}.  I refer to this specification of training as \emph{adaptation-loss neutral training}.

Framed in this way, it is clear that there is some possible gain from increasing training. In fact, I can replicate the findings from \cite{adhvaryu2023organizational} using the above illustrative specification: 

\begin{prop}  \label{prop:training}
Consider the model with training and a Gaussian process with covariance operator $\gamma C(x,x')$. Suppose that the objective is $C^{2}$ and that there exists an interior optimum, and further that the second-order condition for a strict local maximum holds. Then optimal team size is locally increasing in training whenever:  

\begin{equation*} 
\frac{\beta'(\alpha)}{\beta(\alpha)}c(\alpha)^{2} - \frac{1}{2} (\beta(\alpha)^{2} +2c(\alpha)) c'(\alpha) > 0;
\end{equation*}

\noindent and locally decreasing in training whenever the reverse inequality holds. 

Moreover, if $\beta(\alpha)=\sqrt{1+\alpha}$, $c(\alpha)=1+\alpha$ and $f''(k) \geq 0$, then a constant-factor increase in complexity raises optimal training and lowers optimal team size.
\end{prop}

\noindent Recall that constant factor increases in volatility do not, by themselves, influence team size or hierarchy. However, in the above specification, choosing ``more training'' allows the firm to avoid incurring more adaptation loss while potentially mitigating an increase in the coordination loss that a constant-factor increase might introduce. Thus, an increase in volatility leads to an increase in training. But under adaptation-loss neutral training, training and hierarchy are substitutes. Intuitively, training provides one mechanism for the organization to lower coordination loss, while hierarchy expansion is another. Given the option of training, the organization need not expand hierarchy as much, and hence can enable workers to focus more on their individual task without relying on larger teams to assist with the attention problem. 

The condition I obtain also shows that such hierarchy contraction in response to complexity is \emph{not possible} when $c'(\alpha)=0$. To the extent that a change in the action costs is a short-run phenomenon, this provides a new channel for the main message of \cite{adhvaryu2023organizational}, that hierarchy and training are ``short run substitutes and long run complements.'' Still, as the main applications of interest are slightly different, I caution against overinterpreting this finding.

\subsection{Link Formation Costs} \label{sect:linkformation}

I briefly mention a very simple channel for constant-factor increases in volatility to directly influence team size, which is to introduce costs associated with choosing larger teams directly. In Appendix \ref{app:discrete}, I discuss a microfoundation in the discrete model which motivates having a cost of choosing teams of size $k$ be $ \kappa \frac{k-1}{2}.$

One could also apply the same cost function for group formation at higher levels of hierarchy. A natural interpretation of such link formation costs is that they reflect the need for communication within a given team. Notice that this interpretation provides a natural prediction that larger teams should emerge when $\kappa$ falls. To the extent that team size is inversely related to autonomy, this observation is consistent with the finding in \cite{bloom_garicano_sadun_vanreenen_2014_distinct} that communication-technology improvements are associated with more centralized decision-making. The following result shows what such costs imply about how team size responds to constant-factor changes in state volatility:

\begin{prop}  \label{prop:costlylink}
Suppose the cost of choosing teams of size $k$ is $\kappa \frac{k-1}{2}$, that $f$ is convex, and that the process $W_{0}$ is a mean-zero Gaussian process with covariance operator $\gamma C_{0}(x,x')$. Optimal team size is weakly increasing in $\gamma$, and strictly increasing at any $\gamma$ where the optimal team size is interior. 
\end{prop}

\noindent Once link formation costs are present, the marginal cost of increasing team size is no longer simply the adaptation costs. The organization's payoff is no longer homogeneous of degree one in volatility. In this case, volatility increases change the marginal benefit from increased team size, while keeping the marginal costs constant.

\section{Conclusion}

This paper has shown that the way organizational hierarchies respond to changes in complexity depends on the nature of the change itself. The spectral decomposition that I outline shows that not all changes in the process variance should have the same impact on team size. In some cases, increases in variance might be associated with larger common shocks, whereas others might be associated with more idiosyncrasies across locations.  These changes in complexity appear qualitatively different, and I have shown that they yield opposite predictions in terms of organizational changes. The language of spectral theory provides a clear way of distinguishing them and illustrating why these different predictions emerge---essentially, because the implication depends on the interaction with the attention problem. 

The view this paper advances is that horizontal complexity that is \emph{locally concentrated} is best understood in terms of ``jaggedness'' which reflects other properties of the stochastic process rather than variance (which may also reflect larger shocks which are common across the organization). A central observation is that certain natural modifications of the state process the organization faces are associated both with more jaggedness and larger teams. At the same time, I illustrated that this framework can also speak to other ways organizations may manage horizontal complexity. 

These simple insights were obtained by appealing to the continuum formulation of an organization. When organizations consist of finitely many individuals, simple comparative statics might not be feasible due to divisibility issues (i.e., that the number of workers must be divisible by team size). My formulation circumvents these, at which point the comparative statics are reduced to simple calculus. The key tool for this derivation was the appeal to the representation of the stochastic process in terms provided by the Karhunen-Lo\'eve theorem, a concept which may very well have further applications in economics.

\newpage

\appendix 

\section{Discrete Approximation} 
\label{app:discrete}

\noindent Here I articulate a discrete-location version of the main model. The purpose is to provide a microfoundation for the objective the organization faces in terms of a problem which is easier to interpret economically. I show that the objective in the main paper is approximated by the discrete version as long as the discretization is sufficiently fine. In particular, I illustrate why the covariance operator with $k >1$ is as claimed in the text, and in addition motivate the functional form by which $l_{x}(a_{x}, W_{0}(x))$ as scaling proportionally to a function of team size. 

In the discrete version, the set of worker locations is $X_{g}^{n}=\{0,\frac{1}{n}, \ldots, \frac{n-1}{n},1\}$. There is a state vector whose distribution is multivariate normal, i.e., $(W_{0}^{n}(x))_{x \in X_{g}^{n}} \sim \mathcal{N}(0,\Sigma^{n})$, where, for $i,j=0,\ldots,n$, the $(i+1,j+1)$th entry of $\Sigma^{n}$ is $C(i/n,j/n)$, with $C$ as defined in the main text.

For simplicity, I assume that the group size $k$ divides $n+1$, a restriction that is harmless for the continuum limit since, for fixed integer $k$, one can consider sequences of $n$ such that $k \mid n+1$. Thus, there are $J_n=(n+1)/k$ teams, and team $j=0,\ldots,J_n-1$ consists of workers $jk,jk+1,\ldots,jk+k-1$. I associate team $j$ with the location $j/n$, so that teams are located at
$X_t^n=\left\{\frac{j}{n}:j=0,\ldots,J_n-1\right\}.$ Note that the last team is at location $\frac{J_{n}-1}{n} = \frac{n+1-k}{kn}$, which converges to $1/k$ as $n \rightarrow \infty$. Using the solution for individual actions from the main text, I define $W_{1}^{n}(j/n)$ to be the average action of team $j$:

\begin{equation*} 
W_{1}^{n}(j/n)=\frac{1}{k} \sum_{r=0}^{k-1} \frac{1}{2} W_{0}^{n}\left( \frac{jk+r}{n} \right).\end{equation*}

\noindent In particular, $W_{1}^{n}$ is a process defined on the $J_n=(n+1)/k$ team locations in $X_t^n$. Then, the covariance between $W_{1}^{n}(x)$ and $W_{1}^{n}(x')$ for $x=j/n$ and $x'=j'/n$ is: 

\begin{equation*} 
\text{Cov} \left( W_{1}^{n}(j/n), W_{1}^{n}(j'/n) \right)= \frac{1}{4k^{2}} \sum_{r=0}^{k-1} \sum_{s=0}^{k-1} C \left(\frac{jk+r}{n}, \frac{j'k+s}{n} \right). 
\end{equation*}

\noindent Now consider this expression in the $n \rightarrow \infty$ limit. Along
any sequence with $j/n \rightarrow x$ and $j'/n \rightarrow x'$, for each fixed
$r,s \in \{0,\ldots,k-1\}$,
\[
\frac{jk+r}{n} \rightarrow kx
\qquad \text{and} \qquad
\frac{j'k+s}{n} \rightarrow kx'.
\]
By continuity of $C$, $C \left(\frac{jk+r}{n}, \frac{j'k+s}{n} \right)
\rightarrow C(kx,kx').$ Since the double sum contains $k^{2}$ terms, the covariance between
$W_1^n(j/n)$ and $W_1^n(j'/n)$ converges to $\frac{1}{4}C(kx,kx').$ This motivates the continuum team-state process $W_1$, supported on
$[0,1/k]$, with covariance kernel $\frac{1}{4}C(kx,kx')$. For example, if teams are of size $2$, the team-state process is supported on an interval of length $1/2$, and the original scale is compressed by a factor of $2$.
Under this normalization, headquarters coordinates with a mass $1/k$ of
team-level objects rather than a unit mass of individual workers, capturing the
assumption that within-team communication is relatively efficient. 

The same limiting team-state process can be obtained for noninteger $k$ by allowing, within any subinterval, some fraction of teams to have the integer size below $k$, while the remainder have the integer size above $k$, in proportions such that the average team size is $k$. Arranging these teams so that these proportions hold locally makes the mass of teams converge to $1/k$ and yields the same limiting covariance process. This construction is analogous to the pointwise-matching approximation in \citet{FuchsGaricanoRayo2015}.

I now consider the microfoundations for the adaptation costs. For this, I consider a slightly different but related setting, where each worker instead chooses a \emph{vector} of actions:

\begin{equation*} 
-a_{x,x}^{2} -(a_{x,x}-W_{0}(x))^{2} + \beta \sum_{x' \in T(x), x' \neq x} -a_{x,x'}^{2} -(a_{x,x'}-W_{0}(x'))^{2},
\end{equation*}

\noindent where $T(x)$ is the team that $x$ belongs to. Thus, the idea is that the worker chooses a ``primary action'' to match their own state $W_{0}(x)$, as well as ``coordinating actions'' to match the state of other team members. Under the assumption that all team members observe the same state, $a_{x,x'} = \frac{W_{0}(x')}{2}$. Now, for a fixed $k$,  the distance between workers in the same team can be made arbitrarily small as the discretization becomes arbitrarily fine. Since $C(x,x')$ is continuous, $\mathbb{E}\left[\left(W_0(x)-W_0(x')\right)^2\right]\rightarrow 0$ as $x'\rightarrow x$. 

Putting this together, the expected within-team loss converges to $(1+ \beta(k-1)) \mathbb{E}[W_{0}(x)^{2}]/2$, so the result is as if the individual loss were scaled by $f(k)=1+\beta(k-1)$. This microfoundation therefore motivates the linear specification; allowing $f(k)$ to vary nonlinearly captures the possibility that the importance of coordination may itself vary with team size---for instance if taking multiple coordinating actions is harder than taking a single action. For noninteger $k$, I interpret the smooth function $f(k)$ in the main text as a reduced-form measure of average adaptation cost rather than mechanically deriving its values between integers from this construction.

Lastly, I discuss the exogenous cost of group size from Section \ref{sect:linkformation}. This can be understood in similar terms as other papers in network design (e.g., \cite{MatouschekPowellReich2025}) whereby adding a link imposes some exogenous cost $\kappa$. A group with $k$ members is a clique of size $k$, so that the number of edges formed is $k\cdot(k-1)/2$; the number of groups is $\frac{n+1}{k}$, yielding a total cost of $\kappa \cdot \frac{(n+1)(k-1)}{2}$, and an average cost of $\kappa \cdot \frac{k-1}{2}$ which is the specification presented in Section \ref{sect:linkformation}.  If the cost of a group of size $k$ were instead $c_{\kappa}(k(k-1)/2)$, the average cost would then be $\frac{1}{k} c_{\kappa}(k(k-1)/2)$.  

\newpage

\section{Proofs}
The following proofs invoke the representation of the process $W_{1}(x)$ in terms of the Karhunen-Lo\'eve theorem \citep{GhanemSpanos1991}: 

\begin{equation*} 
W_{1}(x) = \sum_{j=1}^{\infty} Z_{j}\psi_{j}(x) \sqrt{\nu_{j} },
\end{equation*}

\noindent where $Z_{1}, Z_{2}, \ldots$ are independent standard normal random variables. In particular, $(\psi_{j}(x))_{j=1}^{\infty}$ form an orthonormal eigenbasis of the positive-eigenvalue support subspace of the covariance operator, where $\nu_{j}$ is the eigenvalue corresponding to eigenfunction $\psi_{j}$.

\begin{lemma} \label{lemm:IMApprox}
Let $S$ denote any admissible signal, and let $D= \{d_{1},d_{2}, \ldots, \} \subset [0,1/k]$ be a countable collection of locations whose signal values generate $\mathcal{F}_{S}$. Then for every finite $l$ and $q$, 

\begin{equation*} 
I(\{\sqrt{\nu_{1}}Z_{1}, \ldots, \sqrt{\nu_{l}}Z_{l}\}; \{S(d_{1}), \ldots, S(d_{q}) \}) \leq I(W_{1} ; S).
\end{equation*}
\end{lemma}

\begin{proof}[Proof of Lemma \ref{lemm:IMApprox}] Note that \begin{equation*} \sqrt{\nu_{j}} Z_{j} = \int_{0}^{1/k} W_{1}(x) \psi_{j}(x)dx.\end{equation*} Since $W_{1}$ is mean-square continuous and $\psi_{j}$ is continuous by continuity of $C(x,x')$, this integral can be approximated by Riemann sums. Thus, letting $(M_{m})$ denote a sequence of meshes which arise from partitions whose maximum cell width converges to zero, we let $A^{m}=(A_{1}^{m}, \ldots, A_{l}^{m})$ denote the vector of Riemann approximations such that $A_{j}^{m}$ approximates $\sqrt{\nu_{j}} Z_{j}$; furthermore, each $A_{j}^{m}$ is a deterministic linear combination of values of $\{W_{1}(x) : x \in M_{m}\}$, and $A^{m} \rightarrow^{L^{2}} (\sqrt{\nu_{1}} Z_{1}, \ldots, \sqrt{\nu_{l}} Z_{l})$. 

Now, for each $m$, define the enlarged mesh $\hat{M}_{m,q} = M_{m} \cup \{d_{1}, \ldots, d_{q}\}$. Note that $A^{m}$ can be written as a function of realizations of $W_{1}(x)$ on $\hat{M}_{m,q}$; similarly $S^{q}=\{S(d_{1}), \ldots, S(d_{q})\}$ can be written as a function of $S$ realizations on $\hat{M}_{m,q}$. Thus, by the data processing inequality and (\ref{eq:MIDef}): 

\begin{equation*} 
I(A^{m}; \{S(d_{1}), \ldots, S(d_{q})\}) \leq I(\{W_{1}(x): x \in \hat{M}_{m,q}\}; \{S(x): x \in \hat{M}_{m,q} \}) \leq I(W_{1};S). 
\end{equation*}

\noindent  Now, since $A^{m}$ converges in $L^{2}$, the pair $(A^{m}, S^{q})$ converges in probability---and hence in distribution---to $((\sqrt{\nu_{1}} Z_{1}, \ldots, \sqrt{\nu_{l}} Z_{l}),S^{q})$. Since mutual information is lower semicontinuous under joint weak convergence \citep[Lemma 1.4.3(b)]{DupuisEllis1997}, it follows that  

\begin{equation*}
I(\{\sqrt{\nu_{1}}Z_{1}, \ldots, \sqrt{\nu_{l}}Z_{l}\}; S^{q})  \leq \liminf_{m \rightarrow \infty} I(A^{m}; S^{q}) \leq I(W_{1} ; S),
\end{equation*}

\noindent proving the Lemma.
\end{proof}

\begin{proof}[Proof of Proposition \ref{lem:RISol}]  Headquarters' problem is to find a stochastic process $Y$ satisfying: 

\begin{equation*} 
I(W_{1};Y) \leq \rho,
\end{equation*}

\noindent minimizing the coordination loss, which can be rewritten as:  

\begin{equation*} \mathbb{E} \left[ \int_{0}^{1/k}(W_{1}(x)-\mathbb{E}[ W_{1}(x) \mid Y] )^{2} dx \right]. \end{equation*} 

\noindent Since the Karhunen-Lo\'eve partial sums converge to $W_{1}(x)$ in mean integrated square with conditioning on $Y$ being a contraction in the $L^{2}$ norm, conditioning commutes with the limit so that $\mathbb{E}[W_{1}(x) \mid Y] = \sum_{j=1}^{\infty} \sqrt{\nu_{j}} \mathbb{E}[Z_{j} \mid Y] \psi_{j}(x)$, so that the coordination loss can also be written $\mathbb{E}\left[ \int_{0}^{1/k} \left( \sum_{j=1}^{\infty} \sqrt{\nu_{j}}(Z_{j} - \mathbb{E}[Z_{j} \mid Y]) \psi_{j}(x) \right)^{2} dx \right]$.

Orthonormality of the $\psi_{j}$s implies that $\int_{0}^{1/k} \psi_{j}(x)^{2}dx=1$ and $\int_{0}^{1/k} \psi_{j}(x)\psi_{l}(x)dx=0$. Applying Bessel-Parseval's identity \citep[Theorem 5.9 in][]{brezis2010functional} to unfactor the square inside of the integral, this expression is equal to: 

\vspace{-5mm}

\begin{equation} 
 \mathbb{E}\left[ \sum_{j=1}^{\infty} \nu_{j}(Z_{j} - \mathbb{E}[Z_{j} \mid Y])^{2}  \right], \label{eq:squares}
\end{equation}

\noindent Define:

\begin{equation*} 
b(x) := \mathbb{E}[W_{1}(x) \mid Y], \qquad \tilde{Y}_{j} : =  \sqrt{\nu_{j}} \mathbb{E}[Z_{j} \mid Y], 
\end{equation*}

\noindent noting that it also holds that $\tilde{Y}_{j} = \int_{0}^{1/k} b(x)\psi_{j}(x)dx$. I first consider an auxiliary problem where headquarters chooses functions of the form $\sum_{j=1}^{l} \tilde{Y}_{j} \psi_{j}(x)$, and where the vector $(\tilde{Y}_{1}, \ldots, \tilde{Y}_{l})$ is generated by $(Z_{1}, \ldots, Z_{l})$ and independent signal noise, and is therefore jointly independent of $Z_{l+1}, \ldots$. Denote any such choice by $Y_{l}$. In this auxiliary problem, headquarters seeks to minimize:

\begin{equation*} 
 \mathbb{E}\left[ \sum_{j=1}^{l} \nu_{j}(Z_{j} - \mathbb{E}[Z_{j} \mid Y_{l}])^{2}  \right] + \sum_{j=l+1}^{\infty} \nu_{j}. 
\end{equation*}

\noindent I mention that Mercer's theorem implies that $\sum_{j=1}^{\infty} \nu_{j} < \infty$, which is an upper bound on the objective.  In light of this observation, define $X_{l}(x)= \sum_{j=1}^{l} Z_{j}\psi_{j}(x)\sqrt{\nu_{j}}$.
 Now, I claim that for every $l$, it is possible to find $l$ points such that: 

\begin{equation*} 
A_{l}:=\begin{pmatrix} \psi_{1}(x_{1}) & \ldots& \psi_{l}(x_{1}) \\ \vdots & \ddots & \vdots \\ \psi_{1}(x_{l}) & \ldots & \psi_{l}(x_{l}) \end{pmatrix}
\end{equation*}

\noindent is invertible. The claim follows by induction; since $\psi_{1} \neq 0$, the $l=1$ case is immediate. Now, suppose this matrix is invertible for some $l-1$. Fix some instance of such $l-1$ points with $A_{l-1}$ denoting the corresponding matrix. Schur's formula implies that if no other points such that $A_{l}$ is invertible can be found, then: \begin{equation*} \psi_{l}(x)= (\psi_{1}(x), \ldots, \psi_{l-1}(x))A_{l-1}^{-1} \begin{pmatrix} \psi_{l}(x_{1}) \\ \vdots \\ \psi_{l}(x_{l-1}) \end{pmatrix}, \end{equation*}

\noindent contradicting that $\psi_{1}(x), \ldots, \psi_{l}(x)$ are linearly independent. 

Now, using the definition of $A_{l}$, the following identities obtain:

\begin{equation*} 
A_{l} \cdot \begin{pmatrix} \sqrt{\nu_{1}} Z_{1} \\ \vdots \\ \sqrt{\nu_{l}} Z_{l} \end{pmatrix} =\begin{pmatrix} X_{l}(x_{1}) \\ \vdots \\ X_{l}(x_{l}) \end{pmatrix}, ~~~~ A_{l} \cdot \begin{pmatrix}  \tilde{Y}_{1} \\ \vdots \\ \tilde{Y}_{l} \end{pmatrix} =\begin{pmatrix} Y_{l}(x_{1}) \\ \vdots \\ Y_{l}(x_{l}) \end{pmatrix}. 
\end{equation*}

Since $A_{l}$ is invertible, the evaluation vector $(Y_{l}(x_{1}), \ldots, Y_{l}(x_{l}))$ and the coefficient vector $(\tilde{Y}_{1}, \ldots, \tilde{Y}_{l})$ determine one another. 

I now show that $I(W_{1}; Y_{l})=I(\{\sqrt{\nu_{1}}Z_{1}, \ldots, \sqrt{\nu_{l}}Z_{l}\};\{\tilde{Y}_{1}, \ldots, \tilde{Y}_{l}\})$, so that I will be able to argue that the problem of choosing a process in the auxiliary problem is equivalent to the analogous problem with multivariate normal vectors.  Since $Y_{l}$ is finite-rank and $A_{l}$ is invertible, the values $Y_{l}(x_{1}), \ldots, Y_{l}(x_{l})$ generate $\mathcal{F}_{Y_{l}}$. Thus, in Lemma \ref{lemm:IMApprox}, take the countable generating collection $D$ whose first $l$ elements satisfy $d_{i}=x_{i}$ for $i=1, \ldots, l$ and apply this lemma with $q=l$, so that:

\begin{equation*} 
I(\{\sqrt{\nu_{1}}Z_{1}, \ldots, \sqrt{\nu_{l}}Z_{l}\};\{Y_{l}(x_{1}), \ldots, Y_{l}(x_{l})\}) \leq I(W_{1}; Y_{l}).
\end{equation*}

\noindent Using the invertibility of $A_{l}$, this inequality holds if and only if:

\begin{equation*} 
I(\{\sqrt{\nu_{1}}Z_{1}, \ldots, \sqrt{\nu_{l}}Z_{l}\};\{\tilde{Y}_{1}, \ldots, \tilde{Y}_{l}\}) \leq I(W_{1}; Y_{l}).
\end{equation*}

To obtain the reverse inequality, consider any finite mesh $M= \{x_{1}, \ldots, x_{m}\}$. Write: 
\begin{equation*} 
\begin{pmatrix}W_{1}(x_{1}) \\ \vdots\\\ W_{1}(x_{m}) \end{pmatrix}
=B_{m} \cdot \begin{pmatrix}\sqrt{\nu_{1}}Z_{1} \\ \vdots \\ \sqrt{\nu_{l}}Z_{l}\end{pmatrix}+\begin{pmatrix}R_{1} \\ \vdots\\ R_{m}\end{pmatrix},
\end{equation*}
\noindent where the $i,j$th entry of $B_{m}$ is $\psi_{j}(x_{i})$ and $R_{i} = \sum_{j=l+1}^{\infty} \sqrt{\nu_{j}} \psi_{j}(x_{i}) Z_{j}$. Note further that $\{Y_{l}(x_{1}), \ldots, Y_{l}(x_{m})\}$ is a deterministic function of $\{\tilde{Y}_{1}, \ldots, \tilde{Y}_{l}\}$. The data processing inequality thus implies that: 

\begin{multline*} 
I(\{W_{1}(x_{1}), \ldots, W_{1}(x_{m})\}; \{Y_{l}(x_{1}), \ldots, Y_{l}(x_{m})\}) \\ \leq I(\{\sqrt{\nu_{1}}Z_{1}, \ldots, \sqrt{\nu_{l}}Z_{l}, R_{1}, \ldots, R_{m}\}; \{\tilde{Y}_{1}, \ldots, \tilde{Y}_{l}\}).
\end{multline*}

But since the vector $(R_{1}, \ldots, R_{m})$ is always jointly independent of $(\sqrt{\nu_{1}}Z_{1}, \ldots, \sqrt{\nu_{l}}Z_{l}, \tilde{Y}_{1}, \ldots, \tilde{Y}_{l})$, the chain rule implies: 

\begin{equation*} 
I(\{\sqrt{\nu_{1}}Z_{1}, \ldots, \sqrt{\nu_{l}}Z_{l}, R_{1}, \ldots, R_{m}\}; \{\tilde{Y}_{1}, \ldots, \tilde{Y}_{l}\})=I(\{\sqrt{\nu_{1}}Z_{1}, \ldots, \sqrt{\nu_{l}}Z_{l}\}; \{\tilde{Y}_{1}, \ldots, \tilde{Y}_{l}\}).
\end{equation*}

\noindent Taking the supremum over all meshes and using the previous two displayed expressions yields that $I(W_{1}; Y_{l}) \leq I(\{\sqrt{\nu_{1}}Z_{1}, \ldots, \sqrt{\nu_{l}}Z_{l}\}; \{\tilde{Y}_{1}, \ldots, \tilde{Y}_{l}\})$; since the above showed the reverse inequality, we have this holds with equality. 

Furthermore, since $\tilde{Y}_{1}, \ldots, \tilde{Y}_{l}$ determine $Y_l(x)$,
headquarters achieves the same objective if instead selecting the posterior mean
process
\begin{equation*}
\widehat{Y}_{l}(x) =\sum_{j=1}^{l} S_j\psi_j(x),
\text{ where }
S_j:=\sqrt{\nu_j}\mathbb{E}[Z_j \mid \tilde{Y}_{1}, \ldots, \tilde{Y}_{l}].
\end{equation*}
Indeed, $(S_1,\ldots,S_l)$ is a function of
$(\tilde{Y}_{1}, \ldots, \tilde{Y}_{l})$, so the data processing inequality
implies that
\begin{equation*}
I(\{Z_{1}\sqrt{\nu_{1}}, \ldots, Z_{l}\sqrt{\nu_l}\}; \{\tilde{Y}_{1}, \ldots, \tilde{Y}_{l}\})
\geq
I(\{Z_1\sqrt{\nu_1}, \ldots, Z_l\sqrt{\nu_l}\}; \{S_1, \ldots, S_l\}).
\end{equation*}

Moreover, for each $j$, $\mathbb{E}[Z_j\sqrt{\nu_j}\mid S_1,\ldots,S_l] = \mathbb{E}\left[
\mathbb{E}[Z_j\sqrt{\nu_j}\mid \tilde Y_1,\ldots,\tilde Y_l]
\mid S_1,\ldots,S_l \right]=S_j$ by the law of iterated expectations. Applying the process--vector information equality established above to $Y_{l}$ and $\widehat{Y}_{l}$ gives: 

\begin{multline*} 
I(W_{1}; \widehat{Y}_{l}) = I(\{\sqrt{\nu_{1}}Z_{1}, \ldots, \sqrt{\nu_{l}}Z_{l}\};\{S_{1}, \ldots, S_{l}\}) \\ \leq I(\{\sqrt{\nu_{1}}Z_{1}, \ldots, \sqrt{\nu_{l}}Z_{l}\};\{\tilde{Y}_{1}, \ldots, \tilde{Y}_{l}\}) =I(W_{1}; Y_{l}).
\end{multline*}

\noindent Thus, if $Y_{l}$ satisfies the original information constraint on the signal process, then $\widehat{Y}_{l}$ also satisfies it. Moreover, the iterated-expectations argument above shows that the two processes induce the same posterior mean and hence the same objective. 

Thus, it is without loss in the auxiliary problem to instead choose a finite signal vector $(S_{1}, \ldots,S_{l})$, such that $I(\{\sqrt{\nu_{1}}Z_{1}, \ldots, \sqrt{\nu_{l}}Z_{l}\};\{S_{1}, \ldots, S_{l}\}) \leq \rho$ to minimize \begin{equation*} \mathbb{E}\left[ \sum_{j=1}^{l} \nu_{j}(Z_{j} - \mathbb{E}[Z_{j} \mid S_{1}, \ldots, S_{l}])^{2}  \right],\end{equation*} a finite-dimensional problem. Proposition 1 in \cite{KoszegiMatejka2020} implies that the solution involves $S_{1}, \ldots, S_{l}$ being independent Gaussian variables, where $S_{i}$ is correlated with $Z_{i}$ and no other variable, whose solution is to have posterior variance on the $i$th coordinate equal to: 

\begin{equation*} 
\min \left(1, \frac{\gamma}{2\nu_{i}} \right),
\end{equation*}

\noindent for $\gamma$ equal to the Lagrange multiplier on the mutual information constraint. I claim there exists some finite $\bar{n}$ such that the set of modes receiving information and their optimal signal precisions are unchanged for all $n \geq \bar{n}$. By the Spectral Theorem for compact operators, $\nu_{i} \rightarrow 0$ as $i \rightarrow \infty$. On the other hand, letting $\gamma^{n}$ denote the Lagrange multiplier, $\gamma^{n}$ satisfies:

\begin{equation*} 
\rho = \frac{1}{2} \sum_{i=1}^{n} \max\{0, \log(2 \nu_{i}/\gamma^{n}) \}.
\end{equation*}

This identity implies $\gamma^{n}$ is weakly increasing in $n$. Indeed, there are two cases that can arise when increasing the number of terms on the right-hand side from $n$ to $n+1$: either (i) $\nu_{n+1} \leq \frac{\gamma^{n}}{2}$, in which case $\gamma^{n+1} = \gamma^{n}$, or (ii) the right-hand side would strictly increase if $\gamma^{n+1}=\gamma^{n}$, in which case the constraint is restored at some $\gamma^{n+1}>\gamma^{n}$. Therefore, let $\bar{n}$ be such that $\gamma^{\bar{n}} \geq 2 \nu_{\bar{n}+1}$; in this case, increasing $n$ further does not change the solution, since each $S_{k}$ is chosen to be uncorrelated with $Z_{k}$ for $k$ sufficiently large. Indeed, some $\bar{n}$ must exist, since $\nu_{n} \rightarrow 0$ and $\gamma^{n}$ is weakly increasing in $n$, with $\gamma^{1} = 2 \nu_{1} e^{-2\rho}>0$. 

The characterization above implies that there exists $\bar{n} < \infty$ such that, for every $l \geq \bar{n}$, the solution to the $l$-mode auxiliary problem acquires information only about the first $\bar{n}$ modes. Thus, there exists an optimal signal process that is unchanged for all $l \geq \bar{n}$, and the value of the auxiliary problem, $V_{l}$, satisfies $V_{l} = V_{\bar{n}}$ for all $l \geq \bar{n}$.

I now conclude that the finite-mode solution remains a solution to headquarters' unrestricted attention problem. Let $V_{l}$ denote the negative of the minimum coordination loss in the $l$-mode auxiliary problem (which, notably, includes the loss from unattended tail modes), and let $V$ denote the negative of the minimum coordination loss in the unrestricted problem. Clearly: 

\begin{equation} 
V_{l} \leq V. \label{eq:lblim}
\end{equation}

\noindent Indeed, $V_{l}$ is obtained from an optimization problem that adds constraints to $V$, and can hence only be lower than the unconstrained objective. I now show that $V \leq V_{l} + \sum_{j=l+1}^{\infty} \nu_{j}.$ To see this, consider any admissible $Y$ satisfying $I(W_{1};Y) \leq \rho$, and let $D=\{d_{1}, d_{2}, \ldots \}$ denote a countable collection of points whose signal values generate $\mathcal{F}_{Y}$. Define $Y^{q} = \{Y(d_{1}), \ldots, Y(d_{q})\}$, and note that by Lemma \ref{lemm:IMApprox}, for every finite $q$, $I(\{\sqrt{\nu_{1}} Z_{1}, \ldots, \sqrt{\nu_{l}} Z_{l}\}; Y^{q}) \leq I(W_{1}; Y) \leq \rho$.  Thus, $Y^{q}$ is feasible as a signal in the finite-dimensional attention problem that involves the first $l$ modes, meaning that $\mathbb{E}\left[\sum_{j=1}^{l} \nu_{j}(Z_{j} - \mathbb{E}[Z_{j} \mid Y^{q}] )^{2} \right]$ is bounded below by the optimal first-$l$-mode loss from water filling. But since the sigma-fields $\sigma(Y^{q})$ increase to $\mathcal{F}_{Y}$, the martingale convergence theorem implies $\mathbb{E}[Z_{j} \mid Y^{q}] \rightarrow^{L^{2}} \mathbb{E}[Z_{j} \mid Y]$, meaning that: 

\begin{equation*} 
\mathbb{E}[\sum_{j=1}^{l} \nu_{j}(Z_{j} - \mathbb{E}[Z_{j} \mid Y^{q}])^{2}] \rightarrow \mathbb{E}[\sum_{j=1}^{l} \nu_{j}(Z_{j} - \mathbb{E}[Z_{j} \mid Y])^{2}], 
\end{equation*}

\noindent so that the unrestricted signal $Y$ also cannot outperform water-filling on the first $l$ modes. 

To complete the argument, let $L_{l}^{*}$ denote the minimum residual loss from the first $l$ modes in the finite-dimensional water filling problem. Thus, by the definition of $V_{l}$, $V_{l}=-(L_{l}^{*} + \sum_{j=l+1}^{\infty} \nu_{j})$. On the other hand, in the unrestricted problem, the loss in the first $l$ modes is at least $L_{l}^{*}$ and the loss on the remaining modes is nonnegative.  Since $Y$ was arbitrary, the unrestricted value therefore satisfies: 

\begin{equation*} 
V \leq -L_{l}^{*} = V_{l} + \sum_{j=l+1}^{\infty} \nu_{j}. 
\end{equation*}

\noindent as claimed. Since $\sum_{j=l+1}^{\infty}\nu_j \rightarrow 0$,

\begin{equation*} 
V \leq \liminf_{l \rightarrow \infty} V_{l}.
\end{equation*}

\noindent Combining this inequality with (\ref{eq:lblim}) gives
$V \leq \liminf_{l\to\infty}V_l \leq \limsup_{l\to\infty}V_l \leq V,$
and hence $V_{l} \rightarrow V$. But the characterization above implies that there is
some fixed $\bar{n}$ such that, for all $l\geq \bar{n}$, the solution to the
$l$-mode problem uses only the first $\bar{n}$ modes. Therefore
$V_l=V_{\bar{n}}$ for all $l \geq \bar{n}$, so $V=V_{\bar{n}}$. But the $\bar{n}$-mode finite-dimensional problem has an optimal Gaussian signal $(S_{1}, \ldots, S_{\bar{n}})$ given by the water-filling characterization. Because $\widehat{Y}_{\bar{n}}$ is a finite sum of continuous deterministic eigenfunctions with measurable Gaussian coefficients, it is jointly measurable and separable; the process-vector information equality shows that it satisfies the information constraint. Since it attains $V_{\bar{n}}=V$, it is hence optimal in the unrestricted problem. Set

\begin{equation*} 
(\alpha_{1}(x), \ldots, \alpha_{\bar{n}}(x)):= (\psi_{1}(x), \ldots, \psi_{\bar{n}}(x)) A_{\bar{n}}^{-1}. 
\end{equation*}

It follows that $\widehat{Y}_{\bar{n}}(x)=\sum_{i=1}^{\bar{n}} \alpha_{i}(x) \widehat{Y}_{\bar{n}}(x_{i})$, and in particular that the optimal posterior mean is determined by its value at the finite set $\{x_{1}, \ldots, x_{\bar{n}}\}$, as claimed. \end{proof}

\medskip

\begin{proof}[Proof of Theorem \ref{thm:optimalk}] 

Following the steps outlined in the main text and substituting $\mu(k) = \frac{\mu^{*}}{4k}$ into (\ref{eq:solvedobj}), the relevant objective is: 

\begin{equation*} 
- \frac{1}{4k}\left(n \mu^{*} +\sum_{m=n+1}^{\infty}   \lambda_{m} \right)- \left(\frac{f(k)}{2}\right)\sum_{m=1}^{\infty} \lambda_{i}.
\end{equation*}

\noindent Since $f(k)$ is differentiable and positive, fixing $(\lambda_{i})_{i=1}^{\infty}$, the derivative is proportional to:

\begin{equation*} 
\frac{1}{4k^{2}} \frac{\left(n\mu^{*} +\sum_{m=n+1}^{\infty}   \lambda_{m} \right)}{\sum_{m=1}^{\infty} \lambda_{i}}-\frac{f'(k)}{2}
\end{equation*}

\noindent The function $-\frac{1}{4k}a-f(k)$ satisfies strictly increasing differences and is strictly concave in $k$; thus, strict decreases in $a$ yield strict decreases in the optimal value of $k$. To prove the comparative static, I show that increases in $\rho$ correspond to decreases in $\mu^{*}$ and $\frac{n \mu^{*} + \sum_{m=n+1}^{\infty} \lambda_{m}}{\sum_{m=1}^{\infty} \lambda_{m}}$. Doing so proves the comparative static.

That $\mu^{*}$ decreases in $\rho$ follows immediately from differentiating the entropy constraint with respect to $\rho$: writing $\mu$ as a function of $\rho$ and treating the $\lambda_{m}$s as a constant, 

\begin{equation*} 
1= -\frac{n}{2\mu^{*}} \frac{\partial \mu^{*}}{\partial \rho} \Rightarrow \frac{\partial \mu^{*}}{\partial \rho} < 0. 
\end{equation*}

\noindent Furthermore, $\mu^{*}$ is differentiable provided the active set in the attention problem does not change; if the active set does change, then still $\mu$ is continuous in $\rho$ since the contribution to the entropy term is 0 if $\mu^{*}=\lambda_{i}$ for some $i$. Thus, $\mu^{*}$ is everywhere decreasing. Also, $\mu^{*} \rightarrow 0$ as $\rho \rightarrow \infty$, since the attention constraint must hold with equality and if $\mu^{*} > \underline{\mu} > 0$ then the right-hand side (\ref{eq:attention}) would be bounded even as the left-hand side is unbounded. Thus, the numerator is strictly decreasing in $\rho$, proving that the optimal team size weakly decreases in $\rho$, strictly so whenever it is interior. 

To complete the proof I show that  $n\mu^{*} + \sum_{m=n+1}^{\infty} \lambda_{m} \rightarrow 0$. Since $n \rightarrow \infty$, for every $n^{*}$: 

\begin{equation*} 
\lim_{\rho \rightarrow \infty} n\mu^{*} + \sum_{m=n+1}^{\infty} \lambda_{m} \leq \lim_{\rho \rightarrow \infty} n^{*} \mu^{*} + \sum_{m=n^{*}+1}^{\infty} \lambda_{m}=\sum_{m=n^{*}+1}^{\infty} \lambda_{m}. 
\end{equation*}

\noindent But since this inequality holds for all $n^{*}$, it also holds in the limit as $n^{*} \rightarrow \infty$; and since $\sum_{m=1}^{\infty} \lambda_{m} < \infty$, $\lim_{n^{*} \rightarrow \infty} \sum_{m=n^{*}}^{\infty} \lambda_{m}=0$. The argument provided in the main text then shows that for some $\bar{\rho}$ such that $\frac{1}{4} \frac{\left(n\mu^{*} +\sum_{m=n+1}^{\infty}   \lambda_{m} \right)}{\sum_{m=1}^{\infty} \lambda_{i}}=\frac{f'(1)}{2}$, which satisfies the conditions in the theorem. As $\rho \rightarrow 0$, the conditions of the theorem ensure that the optimal value of $k$ is interior, completing the proof. \end{proof}

\begin{proof}[Proof of Theorem \ref{prop:compstat}] The calculations below are done first holding the active set fixed. If $\lambda_{i} = \mu^{*}$, a local increase in $\lambda_{i}$ places this mode in the active set, whereas a local decrease removes it from the active set. Nevertheless, the one-sided cross-partials still coincide: The active- and inactive-mode formulas coincide when $\lambda_{i}= \mu^{*}$, so that the derivative formula derived below also applies at an active set boundary. 

Provided that the optimal choice of $k$ is interior, it must solve the first order condition (\ref{eq:solvedobj}). Denote this objective as $\text{Obj}(k,\lambda)$. Compute $\frac{\partial\text{Obj}}{\partial k}$ as: 

\begin{equation*} 
-n\mu'(k)+\sum_{m=n+1}^{\infty} \frac{1}{4k^{2}}  \lambda_{m} -\frac{f'(k)}{2}\sum_{m=1}^{\infty} \lambda_{m}
\end{equation*}

\noindent Next I compute $\mu'(k)$. Rewriting the relative entropy constraint: 

\begin{equation*} 
2\rho =   - n \log(\mu(k))  -n\log(4k) +\sum_{m=1}^{n} \log(\lambda_{m})  
\end{equation*} 

Since this holds for all $k$, taking the derivative of both sides with respect to $k$ yields: 

\begin{equation} 
0 =  - n \frac{1}{\mu(k)} \mu'(k) - \frac{n}{k} \Rightarrow \mu'(k) = - \mu(k) \left( \frac{1}{k}  \right). \label{eq:muprimesol}
\end{equation}

\noindent Now consider $\frac{\partial^{2} \text{Obj}}{\partial k \partial \lambda_{i}}$. First suppose $i \leq n$. Substituting in (\ref{eq:muprimesol}) into $\frac{\partial \text{Obj}}{\partial k}$, I have: 

\begin{equation*} 
n \frac{\partial \mu}{\partial \lambda_{i}} \left( \frac{1}{k}  \right) -\frac{f'(k)}{2}
\end{equation*}

\noindent Again, differentiating the entropy constraint implies:

\begin{equation*} 
0 = -n \frac{\frac{\partial \mu}{\partial \lambda_{i}}}{\mu(k)}+ \frac{1}{\lambda_{i}} \Rightarrow \frac{\partial \mu}{\partial \lambda_{i}} = \frac{\mu(k)}{n \lambda_{i}} =\frac{\mu^{*}}{n\lambda_{i}}\left(\frac{1}{4k}\right) . 
\end{equation*}

\noindent Thus, $\frac{\partial^{2} \text{Obj}}{\partial k \partial \lambda_{i}}$ is: 

\begin{equation*} 
\frac{\mu^{*}}{\lambda_{i}}\frac{1}{k}\left(\frac{1}{4k}\right) -\frac{f'(k)}{2}. 
\end{equation*}

\noindent But for $i > n$, $\frac{\partial^{2} \text{Obj}}{\partial k \partial \lambda_{i}}$ is: 

\begin{equation*} 
\frac{1}{4k^{2}} -\frac{f'(k)}{2}. 
\end{equation*}

Now, note that $\text{Obj}(k, \lambda)$ is homogeneous of degree 1 in $\lambda$. But if $k^{*}$ is optimal at $\lambda$, then $\frac{\partial \text{Obj}(k,\lambda)}{\partial k}=0$ at $k^{*}$, and thus for all $t$ that $\frac{\partial \text{Obj}(k,t\lambda)}{\partial k}=0$. Taking the derivative of this expression with respect to $t$ and evaluating at $t=1$ yields the following: 

\begin{equation} 
\sum_{m=1}^{\infty} \frac{\partial^{2} \text{Obj}(k,\lambda)}{\partial k \partial \lambda_{m}} \lambda_{m}=0. \label{eq:Euler}
\end{equation}

\noindent Substitute in the computed cross partials and dividing by $\sum_{m=1}^{\infty} \lambda_{m}$ to obtain the following identity: 

\begin{equation*} 
\overbrace{\frac{n\mu^{*}+ \sum_{m=n+1}^{\infty} \lambda_{m}}{\sum_{m=1}^{\infty} \lambda_{m}}}^{:=R} \frac{1}{4k^{2}} -\frac{f'(k)}{2}=0 \Rightarrow \frac{f'(k)}{2}=\frac{R}{4k^{2}}
\end{equation*}

Furthermore $R<1$, since $\mu^{*}<\lambda_{m} $ for all active modes
$m\leq n$. Therefore, for $i\leq n$,
\begin{equation*}
\frac{\partial^{2} \text{Obj}(k,\lambda)}
{\partial k \partial \lambda_i}
=
\left(\frac{\mu^{*}}{\lambda_i}-R\right)\frac{1}{4k^{2}},
\end{equation*}
while for $i>n$,
\begin{equation*}
\frac{\partial^{2} \text{Obj}(k,\lambda)}
{\partial k \partial \lambda_i}
=
(1-R)\frac{1}{4k^{2}}.
\end{equation*}
Equivalently, for every \(i\),
\[
\frac{\partial^{2} \text{Obj}(k,\lambda)}
{\partial k \partial \lambda_i}
=
\left(\min\left\{1,\frac{\mu^{*}}{\lambda_i}\right\}-R\right)
\frac{1}{4k^{2}}.
\]
Define
\[
\bar{\lambda}:=\frac{\mu^{*}}{R}.
\]
Since \(R<1\), \(\bar{\lambda}>\mu^{*}\). Hence the cross-partial is strictly
positive when \(\lambda_i<\bar{\lambda}\), strictly negative when
\(\lambda_i>\bar{\lambda}\), and equal to zero only when
\(\lambda_i=\bar{\lambda}\). Note furthermore that the region above this cutoff is nonempty; by Assumption \ref{assum:infinite}, the inactive set contains infinitely many strictly positive eigenvalues, and each contributes $\frac{(1-R)\lambda_{i}}{4k^{2}}$ to the left-hand side of (\ref{eq:Euler}); thus, for this identity to hold, some active mode $i \leq n$ must therefore contribute a strictly negative term, so that $\mu^{*}/\lambda_{i} < R$, or equivalently $\lambda_{i} > \bar{\lambda}$. In particular, $\lambda_{1} > \bar{\lambda}$. Since convexity of $f$ implies $\text{Obj}_{kk} < 0$, and $\rho \in (0, \bar{\rho})$ implies $k^{*}$ is interior, it follows that $\frac{\partial k^{*}}{\partial \lambda_{i}}$ has the same sign as $\frac{\partial^{2} \text{Obj}}{\partial k \partial \lambda_{i}}$. Using this observation, one can immediately translate the signs of the cross-partials into the claimed comparative statics in the theorem.   \end{proof}

\begin{proof}[Proof of Proposition \ref{prop:robustcomplexity}] In the proof, I first carry out the calculations holding $n$ locally fixed. If $\lambda_{m}(0)=\mu^{*}(0)$ for some $m$, the same conclusion follows using one-sided derivatives: the active- and inactive-mode formulas coincide at the boundary, as captured by the expression $\min \{1,\mu^{*}/\lambda_{m} \}$. Hence the derivative formula below remains valid even if $n$ changes at $t=0$. Following the steps of Theorem \ref{prop:compstat}, the organization's objective can be written as: 

\begin{equation*} 
-\frac{1}{4k} \frac{n\mu^{*}(t)+\sum_{m=n+1}^{\infty} \lambda_{m}(t)}{\sum_{m=1}^{\infty} \lambda_{m}(t)}-\frac{f(k)}{2}.
\end{equation*}

\noindent The derivative with respect to $k$ is:

\begin{equation*} 
\frac{1}{4k^{2}} \frac{n\mu^{*}(t)+\sum_{m=n+1}^{\infty} \lambda_{m}(t)}{\sum_{m=1}^{\infty} \lambda_{m}(t)}-\frac{f'(k)}{2}.
\end{equation*}

\noindent Define $\tilde{\mu}(t) = \frac{\mu^{*}(t)}{\sum_{m=1}^{\infty} \lambda_{m}(t)}$ and $\tilde{\lambda}_{j}(t) = \frac{\lambda_{j}(t)}{\sum_{m=1}^{\infty} \lambda_{m}(t)}$, so that the first-order condition becomes:

\begin{equation*} 
\frac{1}{4k^{2}}(n\tilde{\mu}(t)+\sum_{m=n+1}^{\infty}\tilde{\lambda}_{m}(t))-\frac{f'(k)}{2}.
\end{equation*}

\noindent Following the arguments from the proof of Theorem \ref{prop:compstat}, the comparative statics with respect to $k$ depend on the sign of this derivative with respect to $t$. Since $f$ does not depend on $t$, only the first term influences the sign of this objective. Furthermore, since $\Lambda$ is $C^{1}$ into $\ell^{1}$, it follows that $\sum_{m=1}^{\infty} \lambda_{m}(t)$ is $C^{1}$ since summation is a continuous linear functional on $\ell^{1}$. Since this sum is nonzero locally around 0, $\Lambda(t)/\sum_{m=1}^{\infty} \lambda_{m}(t)$ is also $C^{1}$ as a map into $\ell^{1}$. Thus, $\sum_{m=n+1}^{\infty} \tilde{\lambda}_{m}(t)$ is also $C^{1}$; that $\tilde{\mu}(t)$ is differentiable follows from the implicit function theorem, using the observation that it is implicitly defined from the attention constraint which, fixing the active set, is a smooth function of $\tilde{\lambda}_{1}(t), \ldots, \tilde{\lambda}_{n}(t)$. Thus, the derivative of the first-order condition with respect to $t$ evaluated at $t=0$ has the same sign as: 

\begin{equation*} 
n\tilde{\mu}'(0) +\sum_{m=n+1}^{\infty} \tilde{\lambda}_{m}'(0). 
\end{equation*}

\noindent   The attention constraint is: 

\begin{equation*} 
2 \rho = \sum_{m=1}^{n} \log \left(\frac{\frac{\sum_{m=1}^{\infty} \lambda_{m}(t)}{4k}\tilde{\lambda}_{m}(t)}{\frac{\sum_{m=1}^{\infty} \lambda_{m}(t)}{4k} \tilde{\mu}(t)} \right),
\end{equation*}

\noindent so that $\tilde{\mu}'(0) = \frac{\tilde{\mu}(0)}{n} \sum_{m=1}^{n} \frac{\tilde{\lambda}_{m}'(0)}{\tilde{\lambda}_{m}(0)}$, and hence the overall derivative at $t=0$ becomes: 

\begin{equation*} 
\tilde{\mu}(0)\sum_{m=1}^{n} \frac{\tilde{\lambda}_{m}'(0)}{\tilde{\lambda}_{m}(0)} +\sum_{m=n+1}^{\infty} \tilde{\lambda}_{m}'(0). 
\end{equation*}

\noindent Since $\frac{\tilde{\mu}(0)}{\tilde{\lambda}_{m}(0)}=\frac{\mu^{*}(0)}{\lambda_{m}(0)}$, this expression is nonnegative if and only if: 

\begin{equation} 
\sum_{m=1}^{\infty} \tilde{\lambda}_{m}'(0) \cdot \min \left(1 ,\frac{\mu^{*}(0)}{\lambda_{m}(0)} \right) \geq 0. \label{eq:maincompnorm}
\end{equation}

\noindent Now, 

\begin{equation*} 
\tilde{\lambda}_{m}'(0) = \frac{\lambda_{m}'(0)}{\sum_{j=1}^{\infty} \lambda_{j}(0)}-\frac{\lambda_{m}(0)}{\left(\sum_{j=1}^{\infty} \lambda_{j}(0)\right)^{2}} \sum_{j=1}^{\infty} \lambda_{j}'(0). 
\end{equation*}

\noindent Using the identity that $\sum_{m=1}^{\infty} \lambda_{m}(0) \cdot \min \left(1 ,\frac{\mu^{*}(0)}{\lambda_{m}(0)} \right) = n \mu^{*}(0)+ \sum_{m=n+1}^{\infty} \lambda_{m}(0)$, multiplying (\ref{eq:maincompnorm}) by $\sum_{j=1}^{\infty} \lambda_{j}(0)$ and substituting yields: 

\begin{equation*} 
\sum_{m=1}^{\infty} \lambda_{m}'(0) \cdot  \left(\min \left(1 ,\frac{\mu^{*}(0)}{\lambda_{m}(0)} \right)  - \frac{n\mu^{*}(0) + \sum_{m=n+1}^{\infty} \lambda_{m}(0)}{\sum_{m=1}^{\infty} \lambda_{m}(0)} \right)\geq 0. 
\end{equation*}

\noindent as claimed in the statement of the proposition.\end{proof}

\begin{proof}[Proof of Proposition \ref{thm:hierarchy}]
As in the single-layer case, the posterior variance $\mu$ scales proportionally to $k_{1}k_{2}$ with a fixed active set, so that we can set $\mu=\mu^{*}/(4k_{1}k_{2})$ into the objective. Fixing $k_{1}$, the optimal $k_{2}$ solves: 

\begin{equation*} 
\min_{\tilde{k}_{2}} \frac{1}{4k_{1}\tilde{k}_{2}} \left(n \mu^{*} + \sum_{m=n+1}^{\infty} \lambda_{m} \right) + \frac{\left(1+ \beta (k_{1} -1) +\gamma (\tilde{k}_{2}-1) \right)}{2} \sum_{m=1}^{\infty} \lambda_{m}.  
\end{equation*}

\noindent Since the objective is convex in $\tilde{k}_{2}$, the first-order condition yields: 

\begin{equation*} 
\frac{1}{4k_{1}k_{2}^{2}} \left(n \mu^{*} + \sum_{m=n+1}^{\infty} \lambda_{m} \right)- \frac{\gamma}{2} \sum_{m=1}^{\infty} \lambda_{m}=0.
\end{equation*}

Hence, 

\begin{equation*} 
k_{2} =  \sqrt{ \frac{1}{k_{1}} \cdot \frac{n \mu^{*} + \sum_{m=n+1}^{\infty} \lambda_{m}}{ 2\gamma \sum_{m=1}^{\infty} \lambda_{m}}}
\end{equation*}

\noindent as claimed. If $k_{1}$ solves the single-layer problem, then $k_{1} = \sqrt{\frac{n \mu^{*} + \sum_{m=n+1}^{\infty} \lambda_{m}}{ 2\beta \sum_{m=1}^{\infty} \lambda_{m}}}$. Thus: 

\begin{equation*} 
k_{2} = \left( \frac{\beta}{\gamma^2} \cdot \frac{n \mu^{*} + \sum_{m=n+1}^{\infty} \lambda_{m} }{2\sum_{m=1}^{\infty} \lambda_{m}} \right)^{1/4}.
\end{equation*}

\noindent In particular, since the fourth root exceeds 1 exactly when its argument does, $k_{2} > 1$ in this case if and only if $\frac{\beta}{\gamma^{2}} \cdot \frac{n \mu^{*} + \sum_{m=n+1}^{\infty} \lambda_{m}}{2\sum_{m=1}^{\infty} \lambda_{m}} > 1$. Now, $k_{1} > k_{2}$ if and only if: 

\begin{equation*} 
\frac{\sqrt{\gamma}}{\beta^{3/4}}\left(\frac{n \mu^{*} + \sum_{m=n+1}^{\infty} \lambda_{m}}{ 2\sum_{m=1}^{\infty} \lambda_{m}} \right)^{1/4}>1.
\end{equation*}

\noindent If $\gamma=\beta$, then this condition becomes: 

\begin{equation*} 
\frac{n \mu^{*} + \sum_{m=n+1}^{\infty} \lambda_{m}}{ 2\sum_{m=1}^{\infty} \lambda_{m}} > \beta.
\end{equation*}

\noindent However, this is precisely the condition for $k_{1} > 1$; if $k_{1}=1$, then the optimization problem for $k_{2}$ is the same and $k_{2}=1$ as well. I conclude that the organization is bottom-heavy. \end{proof}

\begin{proof}[Proof of Proposition \ref{prop:longrun}] The first-order condition for $k_{1}$ yields: 

\begin{equation*} 
k_{1}=\sqrt{ \frac{1}{k_{2}} \cdot \frac{n \mu^{*} + \sum_{m=n+1}^{\infty} \lambda_{m}}{2\beta \sum_{m=1}^{\infty} \lambda_{m}}}.
\end{equation*}

\noindent Thus, $k_{1}^{LR}/k_{1}^{SR}= \sqrt{1/k_{2}}$, so that if $k_{2} > 1$, then $k_{1}^{LR} < k_{1}^{SR}$. Now,

\begin{equation*} 
\frac{k_{1}}{k_{2}}= \sqrt{\frac{k_{1}}{k_{2}}} \sqrt{ \frac{\gamma}{\beta}} \Rightarrow k_{1}/k_{2} = \gamma/ \beta.
\end{equation*}

\noindent Using this, solve for $k_{1}^{LR}$ and $k_{2}^{LR}$: 

\begin{equation*} 
k_{1} =   \left( \frac{\gamma}{\beta^{2}} \cdot \frac{n \mu^{*} + \sum_{m=n+1}^{\infty} \lambda_{m}}{2\sum_{m=1}^{\infty} \lambda_{m}} \right)^{1/3} \text{ and } k_{2} =   \left( \frac{\beta}{\gamma^{2}} \cdot \frac{n \mu^{*} + \sum_{m=n+1}^{\infty} \lambda_{m}}{2\sum_{m=1}^{\infty} \lambda_{m}} \right)^{1/3}.
\end{equation*}

Since $\gamma \geq \beta$, $k_{1} \geq k_{2}$, so that $k_{1} > 1$ if $k_{2} > 1$. Thus, a two-layer hierarchy arises in the long-run if and only if: 

\begin{equation*} 
\frac{\beta}{\gamma^{2}} \cdot \frac{n \mu^{*} + \sum_{m=n+1}^{\infty} \lambda_{m}}{2\sum_{m=1}^{\infty} \lambda_{m}} > 1,
\end{equation*}

\noindent which coincides with the condition for $k_{2}^{SR}  > 1$ as in Proposition \ref{thm:hierarchy}, completing the proof.  \end{proof}

\begin{proof}[Proof of Proposition \ref{prop:generallayers}] Write the first-order conditions for $k_{1}$ and $k_{i}$ with $i \neq 1$. Note that since the first-order conditions are the same,  $k_{2}= k_{3} = \cdots = k_{\ell}$ when there are $\ell$ layers in the organization. Using this observation: 

\begin{equation*} 
k_{1} = \sqrt{ \frac{1}{\beta \cdot k_{i}^{\ell-1}} \frac{n \mu^{*} + \sum_{m=n+1}^{\infty} \lambda_{m}}{2\sum_{m=1}^{\infty} \lambda_{m}}} , ~~~  k_{i} = \sqrt{\frac{1}{\gamma \cdot k_{1} \cdot  k_{i}^{\ell-2}} \frac{n \mu^{*} + \sum_{m=n+1}^{\infty} \lambda_{m}}{2\sum_{m=1}^{\infty} \lambda_{m}}}.
\end{equation*}

\noindent Again obtaining that $k_{1}/k_{i} = \gamma/\beta$: 

\begin{equation*} 
k_{1} = \left( \frac{\gamma^{\ell-1}}{\beta^{\ell}} \cdot \frac{n \mu^{*} + \sum_{m=n+1}^{\infty} \lambda_{m}}{2\sum_{m=1}^{\infty} \lambda_{m}}   \right)^{1/(\ell+1)} ~~~ k_{i} = \left( \frac{\beta}{\gamma^{2}} \frac{n \mu^{*} + \sum_{m=n+1}^{\infty} \lambda_{m}}{2\sum_{m=1}^{\infty} \lambda_{m}}  \right)^{1/(\ell+1)}.
\end{equation*}

\noindent As before, $k_{1} > k_{i}$ since $\gamma > \beta$, and furthermore that $k_{i}^{LR} > 1$ if and only if a second hierarchy layer is beneficial. Therefore, a hierarchy layer is added in the long run if: 

\begin{equation*} 
\left( \frac{\beta}{\gamma^{2}} \frac{n \mu^{*} + \sum_{m=n+1}^{\infty} \lambda_{m}}{2\sum_{m=1}^{\infty} \lambda_{m}}  \right)^{1/(\ell+1)} \geq 1+\varepsilon
\end{equation*}

\noindent Rearranging to express this in terms of $\ell$, this condition becomes:

\begin{equation*} 
\frac{\log\left(\beta/\gamma^{2} \right)+ \log \left(\frac{n \mu^{*} + \sum_{m=n+1}^{\infty} \lambda_{m}}{2\sum_{m=1}^{\infty} \lambda_{m}} \right)}{\log (1+\varepsilon)} \geq \ell+1
\end{equation*}

\noindent Since the inequality for adding a hierarchy layer is harder to satisfy as $\ell$
increases, and since the condition has the form $L\geq \ell+1$, the largest
feasible value of $\ell$ is $\lfloor L-1\rfloor$. This gives the claimed
expression for the number of layers.\end{proof}

\begin{proof}[Proof of Lemma \ref{lem:endoKL}] Note that the organization's payoff depends only on the eigenvalue sequence in the Karhunen-Lo\'eve representation at the chosen process $Q$. Thus, I show that given any candidate $Q$, there exists another process $Q^{*}$ with the same eigenvalue sequence as $Q$ where $Q^{*}$ has the same Karhunen-Lo\'eve eigenfunctions as $P$ satisfying $D(Q^{*}\|P) \leq D(Q \| P)$. Throughout this proof, I assume that $D(Q\|P) < \infty$, since otherwise the organization incurs infinite KL cost, whereas choosing $Q=P$ ensures the organization's objective remains finite. 

Let $C$ denote the covariance operator of the process $P$, and fix any feasible process $Q$ and covariance operator $\tilde{C}$. Let $(\tilde{\lambda}_{i})_{i=1}^{\infty}$ denote the eigenvalue sequence associated with the Karhunen-Lo\'eve representation of $Q$. Let $H_{n}= \text{span} \{\phi_{1}, \ldots, \phi_{n}\}$, and let $\Pi_{n}$ denote the orthogonal projection onto $H_{n}$. Write $Q^{(n)}=(\Pi_{n})_{\#}Q$ and $P^{(n)}=(\Pi_{n})_{\#}P$ for the projected laws on $H_{n}$, which are supported on $H_{n}$ and admit representations as random linear combinations of $\phi_{1}, \ldots, \phi_{n}$. By the data processing inequality:

\begin{equation*} 
D(Q \| P) \geq D(Q^{(n)}\| P^{(n)}).  
\end{equation*}

Since $\lambda_{1}, \ldots, \lambda_{n} > 0$, $P^{(n)}$ is nondegenerate on $H_{n}$. Furthermore, $D(Q^{(n)} \| P^{(n)}) \leq D(Q \| P) < \infty$, so $Q^{(n)}$ is absolutely continuous with respect to $P^{(n)}$. Thus, $Q^{(n)}$ must also be nondegenerate on $H_{n}$, since otherwise it would be supported on a proper linear subspace to which $P^{(n)}$ assigns zero probability. It follows that both covariance matrices are positive definite and the standard finite-dimensional Gaussian KL formula applies. 

Since the image of $\Pi_{n}$ is isometric to $n$-dimensional normal random vectors (i.e., one can identify $H_{n} \cong \mathbb{R}^{n}$ by associating the $i$th coordinate of the vector in $\mathbb{R}^{n}$ with the coefficient on the $i$th basis element $\phi_{i}$), the KL divergence uses the familiar form for multivariate normal random vectors: Given covariance matrix $\Sigma_{Q^{(n)}}$ and $\Sigma_{P^{(n)}}$: 

\begin{equation*} 
D(Q^{(n)}\| P^{(n)}) = \frac{1}{2}[\text{tr}(\Sigma_{P^{(n)}}^{-1} \Sigma_{Q^{(n)}}) - n - \log \det(\Sigma_{P^{(n)}}^{-1} \Sigma_{Q^{(n)}})]. 
\end{equation*}

\noindent Let $\tilde{\lambda}_{1}^{(n)} \geq \cdots \geq \ldots \geq \tilde{\lambda}_{n}^{(n)} > 0$ denote the eigenvalues associated with the process $Q^{(n)}$, enumerated in weakly decreasing order. Now, since the determinant of a product of two matrices is the product of their determinants, and since each determinant is equal to the product of the eigenvalues:

\begin{equation*} 
\log \det(\Sigma_{P^{(n)}}^{-1} \Sigma_{Q^{(n)}}) = \log \left( \prod_{i=1}^{n} \frac{\tilde{\lambda}_{i}^{(n)}}{\lambda_{i}} \right)= \sum_{i=1}^{n} \log \left( \frac{\tilde{\lambda}_{i}^{(n)}}{\lambda_{i}} \right).
\end{equation*} Furthermore, by Ruhe's trace inequality \citep[p. 341]{marshall2011majorization}: 

\begin{equation*} 
\text{tr}(\Sigma_{P^{(n)}}^{-1} \Sigma_{Q^{(n)}}) \geq \sum_{i=1}^{n} \frac{\tilde{\lambda}_{i}^{(n)}}{\lambda_{i}}. 
\end{equation*}

\noindent Thus, $D(Q^{(n)}\| P^{(n)})  \geq D(\bar{Q}^{n} \| P^{(n)})$, where $\bar{Q}^{n}$ is the process using the same eigenfunctions as $P$ but eigenvalues $\tilde{\lambda}_{i}^{(n)}$.

Next, observe that the eigenvalues $\tilde{\lambda}_{i}^{(n)}$ of the covariance operator of $Q^{(n)}$ satisfy: 

\begin{equation*} 
\tilde{\lambda}_{i}^{(n)} \leq \tilde{\lambda}_{i} \text{ and } \tilde{\lambda}_{i}^{(n)} \leq \tilde{\lambda}_{i}^{(n+1)} ,
\end{equation*} for every $i \leq n$, and furthermore, that $\lim_{n \rightarrow \infty} \tilde{\lambda}_{i}^{(n)} = \tilde{\lambda}_{i}$. This follows from the Courant-Fischer min-max principle; see \cite[pp.~491--492, Problem 37, parts 9--10]{brezis2010functional}.

Now define $Q_{n}^{*}$ to be a process with the same Karhunen-Lo\'eve eigenfunctions as $P$, but with eigenvalues $\tilde{\lambda}_{i}^{(n)}$ for $1 \leq i \leq n$ and $\lambda_{i}$ for $i > n$. By construction, $Q_{n}^{*}$ and $P$  coincide in all coordinates in the Karhunen-Lo\'eve representation with $i > n$, while on the first $n$ coordinates they are given by $\bar{Q}^{n}$ and $P^{(n)}$, respectively. Writing the tail as $P^{>n}$, since all Karhunen-Lo\'eve coordinates are independent and since relative entropy is additive under independence: 

\begin{equation*} D(Q_{n}^{*}\| P)= D(\bar{Q}^{n}\|P^{(n)}) + D(P^{>n} \| P^{>n})=D(\bar{Q}^{n} \| P^{(n)}),\end{equation*}

\noindent where both laws factor across the first $n$ Karhunen-Lo\'eve coordinates and the tail coordinates.

Next, I show that $Q$ can always be weakly improved using another Gaussian process with the same Karhunen-Lo\'eve eigenbasis as $P$. Letting $C_{n}^{*}$ denote the covariance operator of $Q_{n}^{*}$, the eigenvalue of $C_{n}^{*}$ associated with the eigenfunction $\phi_{i}$ is $\tilde{\lambda}_{i}^{(n)}$ for $i \leq n$ and $\lambda_{i}$ for $i > n$. The first part is to argue that $Q_{n}^{*}$ converges weakly to $Q^{*}$, the Gaussian process with the same eigenvalue sequence as $Q$ but the eigenbasis of $P$. Let the covariance operator of $Q^{*}$ be denoted by $C^{*}$. Note that the trace norm of $C_{n}^{*}-C^{*}$ is: \begin{equation*} 
\overbrace{\sum_{i \leq n} \abs{\tilde{\lambda}_{i}^{(n)} - \tilde{\lambda}_{i}}}^{(a)}+ \overbrace{\sum_{i > n} \abs{\tilde{\lambda}_{i} - \lambda_{i}}}^{(b)}.
\end{equation*} Now, $\lim_{n \rightarrow \infty} \sum_{i > n} \abs{\tilde{\lambda}_{i} - \lambda_{i}} = 0$, since $ \sum_{i =1}^{\infty} \abs{\tilde{\lambda}_{i} - \lambda_{i}} < \infty$ from the summability of both $\tilde{\lambda}_{i}$ and $\lambda_{i}$ (which holds by assumption). For the term (a), using that $0 \leq \tilde{\lambda}_{i}^{(n)} \leq \tilde{\lambda}_{i}$ and defining $b_{n}^{i}=\mathbf{1}[i \leq n] \abs{\tilde{\lambda}_{i}^{(n)} - \tilde{\lambda}_{i}}$, it therefore holds that $b_{n}^{i} \leq \tilde{\lambda}_{i}$. Furthermore $b_{n}^{i} \rightarrow 0$ for all $i$ as $n \rightarrow \infty$, since $\tilde{\lambda}_{i}^{(n)} \rightarrow \tilde{\lambda}_{i}$ for all $i$. Since $\sum_{i=1}^{\infty} \tilde{\lambda}_{i} < \infty$, dominated convergence thus implies that (a) converges to 0 as well. Since $C_{n}^{*}$ and $C^{*}$ share the same orthonormal eigenbasis, applying Theorem 5.8 in \cite{Kukush2019GaussianMeasures} yields that $Q_{n}^{*}$ converges weakly to $Q^{*}$.\footnote{Note that the condition on off-diagonal terms holds automatically since I take the orthonormal basis to be the common eigenbasis of $C_{n}^{*}$ and $C^{*}$.} Thus, by lower semicontinuity of KL-divergence \citep[Lemma 1.4.3(b)]{DupuisEllis1997}: 
\begin{equation*}
D(Q^{*} \| P ) \leq \liminf_{n} D(Q_{n}^{*} \| P)      
\end{equation*}

\noindent Putting the previous work together: 

\begin{equation*} D(Q^{*} \| P) \leq \liminf_{n} D(Q_{n}^{*}\|P) =\liminf_{n} D(\bar{Q}^{n}\|P^{(n)}) \leq \liminf_{n} D(Q^{(n)} \| P^{(n)})  \leq D(Q\|P).      \end{equation*} Thus, the constructed process $Q^{*}$: (i) has the same covariance eigenvalue sequence as $Q$, with (ii) the covariance eigenbasis of $P$, while (iii) $D(Q^{*} \| P) \leq D(Q \| P)$. This argument thus completes the proof that the eigenbasis of $Q$ can be taken to be the same as $P$. 

It remains to derive the formula. Let $Q^{*,(n)}= (\Pi_{n})_{\#}Q^{*}$ and $P^{(n)}$ as above; the covariance matrices of these processes are exactly $\text{diag}(\tilde{\lambda}_{1}, \ldots, \tilde{\lambda}_{n})$ and $\text{diag}(\lambda_{1}, \ldots, \lambda_{n})$, respectively.  For every $n$, the finite-dimensional formula implies: 

\begin{equation*} 
D(Q^{*,(n)} \| P^{(n)}) = \frac{1}{2} \sum_{j=1}^{n} \left(\frac{\tilde{\lambda}_{j}}{\lambda_{j}}- 1 - \log \left( \frac{ \tilde{\lambda}_{j}}{\lambda_{j}} \right) \right).
\end{equation*}

Furthermore, $\lim_{n \rightarrow \infty} D(Q^{*,(n)} \| P^{(n)}) =D(Q^{*} \| P)$; this follows from the observation that the coordinate sigma-fields generated by the first $n$ coefficients increase to the full Borel sigma-field on $\overline{\text{span}}\{\phi_{j} : \lambda_{j} > 0\}$, together with monotone convergence of relative entropy along increasing sigma-fields \citep[Corollary 5.2.3]{Gray2023Entropy}. Thus, since furthermore $r-1- \log(r) \geq 0$, taking $n \rightarrow \infty$ in the previous equation, 

\begin{equation*} 
D(Q^{*} \| P) =  \frac{1}{2} \sum_{j=1}^{\infty} \left(\frac{\tilde{\lambda}_{j}}{\lambda_{j}}-1- \log \left( \frac{\tilde{\lambda}_{j}}{\lambda_{j}} \right) \right).
\end{equation*}

\noindent The last step is to verify that the covariance operator of the constructed process admits a continuous kernel, as required for feasibility. Let $r_{j} = \tilde{\lambda}_{j}/\lambda_{j}$, and note that the previous formula implies $\sum_{j} (r_{j}-1 - \log(r_{j})) < \infty$. Since $r-1-\log(r) \rightarrow \infty$ as $r \rightarrow \infty$, there must be some $M< \infty$ such that $\tilde{\lambda}_{j} \leq M \lambda_{j}$ for all $j$. Thus: 

\begin{equation*} 
\sup_{x,x'} \sum_{j=N+1}^{\infty} \tilde{\lambda}_{j} \abs{ \phi_{j}(x) \phi_{j}(x')} \leq M \sup_{x,x'} \sum_{j=N+1}^{\infty} \lambda_{j} \abs{ \phi_{j}(x) \phi_{j}(x')}
\end{equation*}

Cauchy-Schwarz implies: 

\begin{equation*} 
\sup_{x, x'} \sum_{j = N+1} \lambda_{j} \abs{ \phi_{j}(x) \phi_{j}(x')} \leq \sup_{x} \sum_{j=N+1}^{\infty} \lambda_{j} \phi_{j}(x)^{2},
\end{equation*}

\noindent while Mercer's theorem implies that this expression converges to 0 as $N \rightarrow \infty$. Thus, $C^{*}(x,x')= \sum_{j=1}^{\infty} \tilde{\lambda}_{j} \phi_{j}(x) \phi_{j}(x')$ is the uniform limit of continuous functions and is therefore continuous. Now, since $D(Q^{*}\|P)<\infty$, every $\tilde{\lambda}_{j}$ is
strictly positive. Indeed, if $\tilde{\lambda}_{j}=0$ for some $j$,
then $Q^{*}$ would be supported on a proper closed subspace to which
$P$ assigns probability zero. Thus, $Q^{*}$ would not be absolutely
continuous with respect to $P$, contradicting the finiteness of
relative entropy. Furthermore, for any distinct $x_{1},\ldots,x_{m}$ and $c_{1},\ldots,c_{m}$ not all zero,
\begin{equation*}
\sum_{i,i'=1}^{m} c_{i}c_{i'}C^{*}(x_{i},x_{i'})=
\sum_{j=1}^{\infty}\tilde{\lambda}_{j} \left(\sum_{i=1}^{m}c_{i}\phi_{j}(x_{i})\right)^{2}.
\end{equation*}

\noindent Since $\tilde{\lambda}_{j}>0$ for all $j$, this expression
can equal zero only if $\sum_{i=1}^{m}c_i\phi_j(x_i)=0$ for every $j$. But then the corresponding quadratic form under $C$ would also vanish; that is,
\begin{equation*}
\sum_{i,i'=1}^{m}c_ic_{i'}C(x_i,x_{i'}) = \sum_{j=1}^{\infty}\lambda_j
\left(\sum_{i=1}^{m}c_i\phi_j(x_i)\right)^2=0,
\end{equation*}
contradicting the positive definiteness of $C$. Thus, $C^{*}$ is
positive definite. Finally, since $Q$ is feasible and $Q^{*}$ has the
same eigenvalues as $Q$, $\sum_{j=1}^{\infty}\tilde{\lambda}_{j}<\infty$,
so $Q^{*}$ has finite integrated variance. Together with the continuity
established above, this proves that $Q^{*}$ is feasible, completing the
proof.\end{proof}

\begin{proof}[Proof of Proposition \ref{thm:endocomplex}] 
First compute the derivative of the multiplier with respect to the chosen $\tilde{\lambda}_{j}$. Abusing notation for the first part of this proof, take the ordering of the eigenvalues to be such that $\tilde{\lambda}_{1} > \tilde{\lambda}_{2} > \cdots$; later this proof will show that the resulting ordering coincides with the ordering from the original eigenvalue sequence. By Proposition \ref{lem:RISol}, for some $n$, the entropy constraint is satisfied for the first $n$ terms in the $\tilde{\lambda}_{j}$ sequence: 

\begin{equation*} 
\rho = \frac{1}{2} \sum_{i=1}^{n} \log \frac{\frac{1}{4k} \tilde{\lambda}_{i}}{\mu}\Rightarrow n \log( \mu)= \sum_{i=1}^{n}\log(\frac{1}{4k} \tilde{\lambda}_{i}) - 2 \rho. 
\end{equation*}

\noindent Write $\mu$ as a function of $\tilde{\lambda}_{j}$. Differentiating both sides of this expression with respect to $\tilde{\lambda}_{j}$: 

\begin{equation*} 
n \cdot \frac{\mu'(\tilde{\lambda}_{j})}{\mu(\tilde{\lambda}_{j})}= \frac{1}{\tilde{\lambda}_{j}}. 
\end{equation*}

Now consider the derivative of the objective, fixing $k$. First suppose that $\frac{\partial \mu}{\partial \tilde{\lambda}_{i}} > 0$. Then: 

\begin{equation*} 
-n\frac{\partial \mu}{\partial \tilde{\lambda}_{j}} -\frac{f(k)}{2} - \eta \left(\frac{1}{2\lambda_{j}}- \frac{1}{2\tilde{\lambda}_{j}} \right)=0. 
\end{equation*} 

Substituting in for the derivative: 

\begin{equation*} 
-\frac{\mu}{\tilde{\lambda}_{j}} -\frac{f(k)}{2} - \eta \left(\frac{1}{2\lambda_{j}}- \frac{1}{2\tilde{\lambda}_{j}} \right)=0. 
\end{equation*}

Solving yields the expression in the text. The next step is to check concavity. Writing $\mu$ as $\mu(\tilde{\lambda}_{j})$ second derivative is proportional to: 

\begin{equation*} 
-\eta +2 \mu(\tilde{\lambda}_{j})-2 \tilde{\lambda}_{j}\mu'(\tilde{\lambda}_{j})=-\eta+2\mu(\tilde{\lambda}_{j})-\frac{2}{n}\mu(\tilde{\lambda}_{j})< -\eta + 2 \mu(\tilde{\lambda}_{j}),
\end{equation*}

\noindent which must be less than 0 in order for $\tilde{\lambda}_{i} > 0$, assuming $\eta > 2 \mu$. Later the proof will show that this holds in the optimal solution.

By contrast, if $\frac{\partial \mu}{\partial \tilde{\lambda}_{i}} =0$, the derivative is instead: 

\begin{equation*} 
-\frac{1}{4k}-\frac{f(k)}{2} - \eta \left(\frac{1}{2\lambda_{j}}- \frac{1}{2\tilde{\lambda}_{j}} \right)=0,
\end{equation*}

\noindent and the expression from the text follows by rearranging. Since the second derivative of the objective is $-\eta/(2\tilde{\lambda}_{j})$, this expression is concave in $\tilde{\lambda}_{j}$ as well, so the first-order condition defines the solution.  

Now, recall that for any $\tilde{\lambda}_{j}$ such that $\frac{\partial \mu}{\partial \tilde{\lambda}_{j}}>0$, $\mu < \frac{\tilde{\lambda}_{j}}{4k}$. Thus, $-\frac{\mu}{\tilde{\lambda}_{j}}>-\frac{1}{4k}$ for any active set eigenvalues; substituting in the solution: 

\begin{equation*} 
\eta/ \lambda_{i} + f(k) < \frac{\eta- 2 \mu}{4k\mu}. 
\end{equation*} Since the left hand side is decreasing in $\lambda_{i}$, the active set eigenvalues are the first $n$ eigenvalues according to the original ordering. Considering inactive set eigenvalues, noting that $\frac{\mu}{\tilde{\lambda}_{j}}\geq \frac{1}{4k}$, substituting in the solution for the solved $\tilde{\lambda}_{j}$ yields:

\begin{equation*} 
4k \mu \geq \frac{\eta}{\frac{\eta}{\lambda_{i}} + 2 \left(\frac{1}{4k} + \frac{f(k)}{2} \right)} \Rightarrow \frac{\eta}{\lambda_{i}}+f(k) \geq \frac{\eta - 2 \mu}{4k\mu}. 
\end{equation*}

\noindent Thus the active and inactive regions are separated by a cutoff with respect to
the original ordering determined by the sequence \((\lambda_i)\). Since both
candidate formulas are increasing in \(\lambda_i\) and agree at the cutoff, the
induced ordering of the chosen \(\tilde{\lambda}_i\)'s is consistent with the
original ordering.

Now consider the optimization problem over $k$. Write the Lagrangian as: 

\begin{multline*} 
\mathcal{L}(k,\lambda)=
- n \mu - \sum_{i=n+1}^{\infty} \frac{1}{4k}\tilde{\lambda}_{i} - \frac{f(k)}{2} \sum_{i=1}^{\infty} \tilde{\lambda}_{i} - \frac{\eta}{2} \sum_{i=1}^{\infty} \left(\frac{\tilde{\lambda}_{i}}{\lambda_{i}}-1- \log \left( \frac{\tilde{\lambda}_{i}}{\lambda_{i}} \right) \right)  \\ + \gamma \left( \rho - \frac{1}{2}\left( \sum_{m=1}^{n} -\log (4k)+\log( \tilde{\lambda}_{m}) - \log \left(\mu(k)\right)
 \right) \right).
\end{multline*}

Since the derivative of the Lagrangian is concave and piecewise differentiable, \cite{milgrom2002envelope} implies that the optimal $k$ satisfies: 

\begin{equation*} 
\frac{1}{4k^{2}} \sum_{i=n+1}^{\infty} \tilde{\lambda}_{i} -\frac{f'(k)}{2} \sum_{i=1}^{\infty} \tilde{\lambda}_{i} + \frac{\gamma n}{2k} =0
\end{equation*}

Using that $\gamma/2 = \mu$ as in the proof of Proposition \ref{lem:RISol} and Proposition 1 in \cite{KoszegiMatejka2020} yields the formula stated in the main text.

I now show that the $\tilde{\lambda}_{i}$s define an admissible Gaussian process, namely that $\sum_{i=1}^{\infty}\tilde{\lambda}_{i}<\infty$ and $D(Q\|P)<\infty$. The former is immediate: the active set contains only finitely many eigenvalues, each with finite $\tilde{\lambda}_i$, while for inactive eigenvalues $\tilde{\lambda}_{i}<\lambda_{i}$ and $\sum_{i=1}^{\infty}\lambda_i<\infty$. To see $D(Q\|P)<\infty$, note first that the active set contributes only finitely many finite terms, since the active eigenvalues satisfy $0<\tilde{\lambda}_i<\infty$. For any inactive eigenvalue, for the appropriate $c(k)$,

\begin{equation*} 
\frac{\tilde{\lambda}_{i}}{\lambda_{i}}-1 -\log \left( \frac{\tilde{\lambda}_{i}}{\lambda_{i}} \right)=\frac{1}{1+c(k) \lambda_{i}}-1 - \log \left( \frac{1}{1+c(k) \lambda_{i}} \right).
\end{equation*}

\noindent Defining $\frac{1}{1+c(k) \lambda_{i}} =r_{i}$, note that for any sufficiently small neighborhood of 1, there exists a constant $C$ such that for all $r$ in that neighborhood,  $r-1-\log(r) \leq C(r-1)^{2}$. Thus, for all sufficiently large inactive $i$, since $\lambda_i\rightarrow 0$:

\begin{equation*} 
\frac{1}{1+c(k) \lambda_{i}}-1 -\log \left( \frac{1}{1+c(k) \lambda_{i}} \right) \leq C  \cdot \left( \frac{c(k) \lambda_{i}}{1 + c(k)\lambda_{i}} \right)^{2} \leq C \cdot c(k)^{2} \lambda_{i}^{2}. 
\end{equation*} 

\noindent Since $\sum_{i=1}^{\infty}\lambda_i<\infty$, it also holds that $\sum_{i=1}^{\infty}\lambda_i^2<\infty$. The remaining finitely many inactive terms are each finite, so the entire inactive set contributes a finite amount to $D(Q\|P)$. Combining this with the finite contribution from the active set gives $D(Q\|P)<\infty$.

The next step is to show that $\eta > 2 \mu$. Note that as $\rho \rightarrow 0$, $\mu$ converges to the largest eigenvalue in the absence of the attention constraint. Now,  $\tilde{\lambda}_{1}=\frac{\eta}{\frac{\eta}{\lambda_{1}}+\frac{1}{2k} +f(k)}$. Since $\mu=\tilde{\lambda}_{1}/(4k)$, the condition becomes that $\eta > \frac{\eta}{2k\frac{\eta}{\lambda_{1}}+1+f(k)}$, which automatically holds since the denominator on the right hand side is always greater than 1. Thus, the formula is valid for $\rho$ sufficiently small by continuity of the solution. On the other hand, that $\frac{\partial \mu}{\partial \rho} < 0$ can be seen from differentiating the attention constraint:  

\begin{equation*} 
n \log(\mu(\rho))=-2\rho -n\log(4k)+\sum_{i=1}^{n} \log \left( \frac{\eta - 2 \mu(\rho)}{f(k)+ \frac{\eta}{\lambda_{i}}} \right) \Rightarrow \mu'(\rho) = - \frac{2(\eta- 2 \mu(\rho))\mu(\rho)}{\eta n}. 
\end{equation*}

\noindent Thus, setting $z(\rho)=\eta- 2 \mu(\rho)$, this implies that $z'(\rho) \propto z(\rho)$. Since $z(0) > 0$, $z(\rho) > 0$ for all $\rho$, and hence that $\eta - 2 \mu(\rho) >0$ for all $\rho$. 

The last step is to argue that the objective is globally concave conditional on the active set. Here, it suffices to show that the Hessian is negative semidefinite when restricting to active set eigenvalues, since cross derivatives involving all other eigenvalues are 0. Using the above calculations and that $\frac{\partial \mu}{\partial \tilde{\lambda}_{i}}= \frac{\mu}{n \tilde{\lambda}_{i}}$, the $i$th diagonal term in the Hessian is $\frac{n-1}{n} \frac{\mu}{\tilde{\lambda}_{i}^{2}} - \frac{\eta}{2 \tilde{\lambda}_{i}^{2}}$, while the $i-j$th cross term is $- \frac{\mu}{n \tilde{\lambda}_{i}\tilde{\lambda}_{j}}$. Thus, fixing any vector $(v_{1}, \ldots, v_{n})$ and setting $y_{i} = \frac{v_{i}}{\tilde{\lambda}_{i}}$, the quadratic form is: 

\begin{equation*}
\left(\mu- \frac{\eta}{2} \right) \sum_{i=1}^{n} y_{i}^{2} - \frac{\mu}{n} \left( \sum_{i=1}^{n} y_{i} \right)^{2} < 0,
\end{equation*}

\noindent where the inequality holds by the argument that $\eta > 2 \mu$.\end{proof}

\begin{proof}[Proof of Theorem \ref{prop:gamma}]
As in the proof of Proposition \ref{prop:robustcomplexity}, the comparative statics depend on how: 

\begin{equation*} 
\frac{4kn\mu(\gamma) + \sum_{i=n+1}^{\infty} \tilde{\lambda}_{i}(\gamma)}{\sum_{i=1}^{\infty} \tilde{\lambda}_{i}(\gamma)} = \frac{4kn  + \overbrace{\sum_{i=n+1}^{\infty} \tilde{\lambda}_{i}(\gamma)/\mu(\gamma)}^{:=T(\gamma)}}{\underbrace{\sum_{i=1}^{n} \tilde{\lambda}_{i}(\gamma)/\mu(\gamma)}_{:=H(\gamma)} + \sum_{i=n+1}^{\infty} \tilde{\lambda}_{i}(\gamma)/\mu(\gamma)}
\end{equation*}

\noindent changes in $\gamma$. Throughout this argument, fix $k$ and consider any $\gamma$ in an open interval over which $n$ is constant. I consider $H(\gamma)$ and $T(\gamma)$ separately. For $H(\gamma)$, recall the entropy constraint is: 

\begin{equation*} 
\frac{1}{2} \sum_{i=1}^{n} \log \left( \frac{\frac{1}{4k} \tilde{\lambda}_{i}(\gamma)}{\mu(\gamma)} \right)= \rho. 
\end{equation*}

\noindent Differentiating the entropy constraint with respect to $\gamma$ gives the identity: 

\begin{equation*} 
\sum_{i=1}^{n} \frac{d}{d\gamma} \log(\tilde{\lambda}_{i}(\gamma))= n \frac{d}{d\gamma} \log(\mu(\gamma))
\end{equation*}

\noindent Note that: 

\begin{equation*} 
\log(\tilde{\lambda}_{i}(\gamma)) = \log(\eta-2\mu(\gamma))+ \log(\gamma \lambda_{i}) - \log(\eta+f(k) \gamma \lambda_{i}), 
\end{equation*}

\noindent so: 

\begin{equation} 
\frac{d}{d\gamma}\log(\tilde{\lambda}_{i}(\gamma))=\frac{d}{d\gamma}\log(\eta-2\mu(\gamma))+ \frac{1}{\gamma} \cdot \overbrace{\frac{\eta}{\eta+f(k)\gamma \lambda_{i}}}^{:=\alpha(\gamma,\lambda_{i})}. \label{eq:tildeideriv}  
\end{equation}

\noindent Returning to the derivative of the entropy constraint, we therefore have: 

\begin{equation*} 
n \frac{d}{d\gamma} \log(\eta - 2\mu(\gamma))+ \frac{1}{\gamma} \sum_{i=1}^{n} \alpha(\gamma,\lambda_{i}) = n \frac{d}{d\gamma} \log(\mu(\gamma)),
\end{equation*}

\noindent and so setting $\bar{\alpha}(\gamma) :=\frac{1}{n}\sum_{j=1}^{n} \alpha(\gamma, \lambda_{j}),$

\begin{equation} 
\frac{d}{d\gamma} \log\left(\frac{\eta - 2 \mu(\gamma)}{\mu(\gamma)} \right) = - \frac{\bar{\alpha}(\gamma)}{\gamma}. \label{eq:alphabarident}
\end{equation}

Now, return to (\ref{eq:tildeideriv}). Subtracting $\frac{d}{d\gamma}\log(\mu(\gamma))$ from both sides therefore yields the key identity: 

\begin{equation*} 
\frac{d}{d\gamma} \log\left( \frac{\tilde{\lambda}_{i}(\gamma)}{\mu(\gamma)} \right) = \frac{1}{\gamma}(\alpha(\gamma,\lambda_{i}) - \bar{\alpha}(\gamma)) \Rightarrow \frac{d}{d\gamma}\frac{\tilde{\lambda}_{i}(\gamma)}{\mu(\gamma)}=\frac{1}{\gamma}\cdot \frac{\tilde{\lambda}_{i}(\gamma)}{\mu(\gamma)} \cdot (\alpha(\gamma,\lambda_{i}) - \bar{\alpha}(\gamma)). 
\end{equation*}

Putting this together, to see how $H(\gamma)$ changes in $\gamma$, it therefore suffices to sign: 

\begin{equation*} 
\sum_{i=1}^{n} \frac{\tilde{\lambda}_{i}(\gamma)}{\mu(\gamma)} \cdot (\alpha(\gamma,\lambda_{i}) - \bar{\alpha}(\gamma)).
\end{equation*}

\noindent However, since $\lambda_{1} \geq \lambda_{2} \geq \cdots \geq \lambda_{n}$, we have $\tilde{\lambda}_{1}(\gamma) \geq \tilde{\lambda}_{2}(\gamma) \geq \cdots \geq \tilde{\lambda}_{n}(\gamma)$ and $\alpha(\gamma, \lambda_{1}) \leq \alpha(\gamma, \lambda_{2}) \leq \cdots \leq \alpha(\gamma, \lambda_{n})$. Thus, by Chebyshev's sum inequality, we have: 

\begin{equation*} 
\frac{d}{d\gamma} H(\gamma) =  \sum_{i=1}^{n} \frac{d}{d\gamma} \frac{\tilde{\lambda}_{i}(\gamma)}{\mu(\gamma)}= \frac{1}{\gamma}\sum_{i=1}^{n} \frac{\tilde{\lambda}_{i}(\gamma)}{\mu(\gamma)} \cdot (\alpha(\gamma,\lambda_{i}) - \bar{\alpha}(\gamma)) \leq 0.
\end{equation*}

\noindent We also note that $H(\gamma) > 4kn$ since by definition each $\frac{1}{4k}\tilde{\lambda}_{i} > \mu(\gamma)$. For $T(\gamma)$, we derive similar inequalities term by-term. Note that:

\begin{equation*} 
\log \left(\frac{\tilde{\lambda}_{i}(\gamma)}{\mu(\gamma)} \right) = \log\left(\frac{\eta}{\mu(\gamma)} \right)+ \log(\gamma)+\log \left( \frac{ \lambda_{i}}{\eta+(f(k)+\frac{1}{2k})\gamma \lambda_{i}} \right). 
\end{equation*}

Differentiating and simplifying, this implies: 

\begin{equation} 
\frac{d}{d\gamma}\log \left(\frac{\tilde{\lambda}_{i}(\gamma)}{\mu(\gamma)} \right)= \frac{d}{d\gamma}\log\left(\frac{\eta}{\mu(\gamma)} \right) + \frac{1}{\gamma} \cdot  \overbrace{\left(\frac{\eta}{\eta+ (f(k) + \frac{1}{2k})\gamma \lambda_{i}} \right)}^{:=\beta(\gamma,\lambda_{i})}. \label{eq:inactivedgamma}
\end{equation}

\noindent Using (\ref{eq:alphabarident}) to find $\frac{d}{d\gamma}\log\left(\frac{\eta}{\mu(\gamma)} \right)= -\frac{\mu'(\gamma)}{\mu(\gamma)}$, the derivative in (\ref{eq:alphabarident}) evaluates and simplifies as:

\begin{equation*} 
 - \frac{\mu'(\gamma)}{\mu(\gamma)} \cdot \left( \frac{\eta}{\eta-2\mu(\gamma)} \right)  = - \frac{ \bar{\alpha}(\gamma)}{\gamma}. 
\end{equation*}

\noindent Putting this together, (\ref{eq:inactivedgamma}) becomes: 

\begin{equation*} 
\frac{d}{d\gamma}\log \left(\frac{\tilde{\lambda}_{i}(\gamma)}{\mu(\gamma)} \right) = \frac{1}{\gamma} \overbrace{\left(\beta(\gamma,\lambda_{i})-\frac{\eta - 2 \mu(\gamma)}{\eta} \bar{\alpha}(\gamma)  \right)}^{(*)}.
\end{equation*}

\noindent I claim $(*)$ is nonnegative, from which it would immediately follow that $T(\gamma)$ is increasing. Define $x^{*}$ as a hypothetical value for $\gamma\lambda$ which is at the water level. Since in this case the active and inactive eigenvalue formulas coincide as per the proof of Proposition \ref{thm:endocomplex}: 

\begin{equation} 
1 =\frac{\eta - 2 \mu}{4k\mu} \frac{x^{*}}{\eta + x^{*}f(k)} \Leftrightarrow  \mu(\eta + x^{*}f(k))=\frac{1}{4k}(\eta - 2\mu) x^{*} \label{eq:threshold}
\end{equation}

\noindent The key step is to show that this identity implies: 

\begin{equation} 
\beta(\gamma,x^{*}/\gamma) =\frac{\eta - 2 \mu(\gamma)}{\eta} \alpha(\gamma,x^{*}/\gamma), \label{eq:thresholdident}
\end{equation}

\noindent so that the claim of interest follows from the observation that $\alpha(\gamma, x^{*}/\gamma) \geq \alpha(\gamma, \lambda_{n}) \geq \bar{\alpha}(\gamma)$ and $\gamma \lambda_{i} \leq x^{*} \leq \gamma \lambda_{j}$ for all $i >n > j$, since at the threshold eigenvalue since $\bar{\alpha}(\gamma)$ is an average and $\alpha(\gamma, \lambda)$ is decreasing in $\lambda$, replacing $\alpha(\gamma,x^{*}/\gamma)$ with $\bar{\alpha}(\gamma)$ increases $(*)$. Substituting in for (\ref{eq:thresholdident}): 

\begin{align*} 
\left(\frac{\eta}{\eta+ (f(k) + \frac{1}{2k})x^{*}} \right)=\frac{\eta - 2\mu(\gamma)}{\eta+f(k)x^{*}} &\Leftrightarrow \eta(\eta+f(k)x^{*}) = (\eta+ (f(k) + \frac{1}{2k})x^{*})(\eta - 2 \mu(\gamma)) \\ &\Leftrightarrow 2\mu(\gamma)(\eta+f(k)x^{*})= \frac{1}{2k} (\eta- 2\mu(\gamma))x^{*},
\end{align*}

\noindent which coincides with (\ref{eq:threshold}) by inspection. 

So $(*)$ must be at least weakly positive if $\lambda_{i}$ places $\tilde{\lambda}_{i}$ at the water level. But $\beta(\gamma,\lambda_{i})$ is decreasing in $\lambda_{i}$, so that at any $\lambda_{i}$ in the inactive set, $\beta(\gamma,\lambda_{i})$ is strictly larger than its value at the water-level. From this, we conclude that $(*)$ is nonnegative, as claimed. 

Returning to the original expression, we therefore see that as $\gamma$ increases provided the active set does not change, $H(\gamma)$ decreases (but is greater than $4kn$) while $T(\gamma)$ increases. At an active-set change, we have the newly added eigenvalue is equal to the water level. Hence, moving it from $T$ into $H$, while increasing $n$ by one, leaves both the numerator and denominator of this expression continuous (and bounded away from zero). Thus, $\frac{4kn+T(\gamma)}{H(\gamma)+T(\gamma)}$ is increasing, as claimed.

The last steps consider the $\gamma \rightarrow \infty$ limit. Specifically, fixing $k$, consider $1 - \frac{4kn(\gamma)+T(\gamma)}{H(\gamma) + T(\gamma)}$ (for this part, allowing $n$ to vary with $\gamma$): 

\begin{equation*} 
\frac{\sum_{i=1}^{n(\gamma)} \left(\frac{\tilde{\lambda}_{i}(\gamma)}{4k\mu(\gamma)}-1\right)}{\sum_{i=1}^{\infty} \frac{\tilde{\lambda}_{i}(\gamma)}{4k\mu(\gamma)}}.
\end{equation*}

\noindent The result that the limiting team size approaches the zero-attention budget solution then holds because this ratio converges to 0 as $\gamma \rightarrow \infty$. The numerator can be bounded using the entropy constraint. For the entropy constraint to be satisfied, for all $i \leq n(\gamma)$, $\frac{\tilde{\lambda}_{i}(\gamma)}{4k\mu(\gamma)} \in [1, e^{2\rho}]$. Thus, find a constant $C_{\rho}$ such that, for all $x \in [1, e^{2\rho}]$, $ x-1 \leq C_{\rho}\log(x)$. In this case: 

\begin{equation*} 
\sum_{i=1}^{n(\gamma)} \left(\frac{\tilde{\lambda}_{i}(\gamma)}{4k\mu(\gamma)}-1\right) \leq C_{\rho} \sum_{i=1}^{n(\gamma)}  \log \left(\frac{\tilde{\lambda}_{i}(\gamma)}{4k\mu(\gamma)} \right)=2C_{\rho} \rho.
\end{equation*}

\noindent Thus, the numerator is bounded. However, the denominator is at least as large as $n(\gamma)$, which I claim satisfies $n(\gamma) \rightarrow \infty$. I show this by contradiction: Suppose there were a sequence of $\gamma_{m} \rightarrow \infty$ where $n(\gamma_{m}) \rightarrow \bar{n}$. The formula for inactive eigenvalues implies: 

\begin{equation*} 
\lim_{\gamma \rightarrow \infty} \tilde{\lambda}_{n(\gamma)+1} = \frac{2k \eta}{1+2kf(k)}.
\end{equation*}

For this to be inactive, $4k\mu$ must be at least as large as this value. On the other hand, the formula for the active eigenvalues imply that these all approach the same level; letting $\mu^{**}$ denote the limiting water level, we have that the limiting active eigenvalues are all $\frac{\eta - 2 \mu^{**}}{f(k)}$. This implies the identity: 

\begin{equation*} 
\frac{\eta - 2\mu^{**}}{4k\mu^{**}f(k)}=e^{2 \rho/\bar{n}} \Rightarrow 4k\mu^{**} = \frac{2k\eta}{1+2e^{2 \rho/\bar{n}}kf(k)}.
\end{equation*}

\noindent But this expression and the previous formula implies that $\mu^{**} < \frac{1}{4k} \lim_{\gamma \rightarrow \infty} \tilde{\lambda}_{\bar{n}+1}$ since $e^{2 \rho/ \bar{n}}>1$. Thus, no such $\bar{n}$ can exist, and thus $n(\gamma) \rightarrow \infty$, completing this part of the proof. 

Finally I show the last part of Theorem \ref{prop:gamma}, using the previous result that $k^{*}(\gamma) \rightarrow k_{0}$ and the added condition that $\lambda_{i+1} \geq \underline{c} \lambda_{i}$. As argued in the proof of Proposition \ref{thm:endocomplex}, $\mu$ is bounded above by $\eta$, so that we can define $\mu^{**}$ to be the limiting water level, passing to a subsequence if necessary. Fix any $m$, and note that since $n(\gamma) \rightarrow \infty$, for $\gamma$ sufficiently large the first $m$ eigenvalues are part of the active set. Thus, we have for $\gamma$ sufficiently large and any fixed $m$: 

\begin{equation*} 
\sum_{i=1}^{m} \log \left( \frac{\tilde{\lambda}_{i}(\gamma)}{4k\mu(\gamma)} \right)\leq 2 \rho
\end{equation*}

Thus, computing the limiting $\tilde{\lambda}_{i}(\gamma)$

\begin{equation*} 
m \log \left( \frac{\eta - 2\mu^{**}}{4k_{0}\mu^{**}f(k_{0})} \right) \leq 2 \rho.
\end{equation*}

But since this holds for all $m$: 

\begin{equation} 
\frac{\eta - 2\mu^{**}}{4k_{0}\mu^{**}f(k_{0})} =1 \Rightarrow 4k_{0}\mu^{**} = \frac{2k_{0} \eta}{1+2f(k_{0})k_{0}}. \label{eq:limitbound}
\end{equation}

\noindent The previous step gives a weak inequality, but exact inequality follows from the restriction that $\tilde{\lambda}_{i}/(4k \mu^{**}) \geq 1$ for any active eigenvalue. Thus, we have that this is the expression for the limiting eigenvalues in the head, and in particular the first chosen eigenvalue. The next step is to bound $\lambda_{n(\gamma)}$, which is complicated by the fact that $n(\gamma) \rightarrow \infty$. However, by the definition of the active eigenvalues, we have: 
\begin{equation*} 
\frac{ \eta - 2 \mu(\gamma)}{\frac{\eta}{\gamma\lambda_{n(\gamma)}}+f(k(\gamma))} > 4k(\gamma)\mu(\gamma) \Rightarrow \frac{\eta - 2\mu(\gamma)}{4k(\gamma) \mu(\gamma)} - f(k(\gamma)) > \frac{\eta}{\gamma \lambda_{n(\gamma)}}.
\end{equation*}

But (\ref{eq:limitbound}) shows that the left-hand side converges to 0 as $\gamma \rightarrow \infty$, so that $\gamma \lambda_{n(\gamma)} \rightarrow \infty$. By the condition that $ \lambda_{i+1} \geq \underline{c}\lambda_{i}$, we therefore have that $\gamma \lambda_{n(\gamma)+1} \rightarrow \infty$ as well. But the formula for $\tilde{\lambda}_{n(\gamma)+1}$ implies: 

\begin{equation*} 
\lim_{\gamma \rightarrow \infty} \tilde{\lambda}_{n(\gamma)+1} = \lim_{\gamma \rightarrow \infty} \frac{1}{\frac{1}{\gamma \lambda_{n(\gamma)+1}}+ \frac{2}{\eta} \left( \frac{1}{4k(\gamma)} + \frac{f(k(\gamma))}{2} \right)}= \frac{1}{\frac{2}{\eta}\left( \frac{1}{4k_{0}} + \frac{f(k_{0})}{2} \right)} = \frac{2 k_{0} \eta}{1+2f(k_{0}) k_{0}}, 
\end{equation*}

\noindent coinciding with the formula derived for $\lim_{ \gamma \rightarrow \infty} \tilde{\lambda}_{1}(\gamma)$, as claimed.  \end{proof}

\begin{proof}[Proof of Theorem \ref{prop:eta}]
The proof essentially follows identical steps as the proof of Theorem \ref{prop:gamma}; I outline the key calculations and refer to the proof of Theorem \ref{prop:gamma} for added details. Define $T(\eta), H(\eta)$ as there and consider all other parameters as a function of $\eta$. First consider $H(\eta)$. Differentiating the entropy constraint yields: 

\begin{equation*} 
\sum_{i=1}^{n} \frac{d}{d\eta} \log(\tilde{\lambda}_{i}(\eta)) = n \frac{d}{d \eta} \log(\mu(\eta)). 
\end{equation*}

\noindent Using the formula for $\tilde{\lambda}_{i}(\eta)$, 

\begin{equation} 
\frac{d}{d\eta} \log(\tilde{\lambda}_{i}(\eta)) = \frac{d}{d\eta} \log(\eta-2\mu(\eta)) - \frac{1}{\eta} \cdot  \frac{\eta}{\eta+f(k)  \lambda_{i}}. \label{eq:lambdaeta}
\end{equation}

\noindent Defining $\alpha(\eta, \lambda_{i})$ and $\bar{\alpha}(\eta)$ analogously to the proof of Theorem \ref{prop:gamma}, the same argument implies:

\begin{equation}
\frac{d}{d \eta} \log \left( \frac{\eta - 2\mu(\eta)}{\mu(\eta)} \right) =\frac{\bar{\alpha}(\eta)}{\eta}. \label{eq:etaclever}
\end{equation}

\noindent Subtracting $\frac{d}{d\eta} \log(\mu(\eta))$ from both sides of (\ref{eq:lambdaeta}): 

\begin{equation*} 
\frac{d}{d\eta} \log \left( \frac{ \tilde{\lambda}_{i}(\eta)}{\mu(\eta)} \right) =\frac{1}{\eta}( \bar{\alpha}(\eta) - \alpha(\eta, \lambda_{i})). 
\end{equation*}

\noindent \noindent Thus, $\frac{d}{d\eta} H(\eta)=\frac{1}{\eta} \sum_{i=1}^{n} \frac{\tilde{\lambda}_{i}(\eta)}{\mu(\eta)} (\bar{\alpha}(\eta) - \alpha(\eta, \lambda_{i}))$. Since $\alpha(\eta, \lambda_{i})$ is increasing in the index while $\frac{\tilde{\lambda}_{i}(\eta)}{\mu(\eta)}$ is decreasing in the index, Chebyshev's sum inequality implies that $\frac{d}{d\eta} H(\eta) \geq 0$.

For $T(\eta)$: 

\begin{equation} 
\frac{d}{d\eta} \log \left( \frac{\tilde{\lambda}_{i}}{\mu(\eta)} \right) = \frac{d}{d\eta} \log \left( \frac{\eta}{\mu(\eta)} \right)- \frac{1}{\eta} \cdot \left( \frac{\eta}{\eta+(f(k)+\frac{1}{2k} ) \lambda_{i}} \right). \label{eq:firstbetaeta}
\end{equation}

\noindent Now, $\frac{d}{d\eta} \log \left( \frac{\eta}{\mu(\eta)} \right)= \frac{1}{\eta} - \frac{\mu'(\eta)}{\mu(\eta)}$. Evaluating the derivative on the left side of (\ref{eq:etaclever}), I find: 

\begin{equation*} 
\left(\frac{1}{\eta} - \frac{\mu'(\eta)}{\mu(\eta)} \right) \cdot \frac{\eta}{\eta - 2 \mu(\eta)} =\frac{\bar{\alpha}(\eta)}{\eta}.
\end{equation*}

\noindent From this and using (\ref{eq:firstbetaeta}), it follows that: 

\begin{equation*} 
\frac{d}{d\eta}  \log \left( \frac{\tilde{\lambda}_{i}}{\mu(\eta)} \right)=\frac{1}{\eta} \left( \frac{\eta - 2 \mu(\eta)}{\eta} \bar{\alpha}(\eta)- \beta(\eta, \lambda_{i}) \right).
\end{equation*}

\noindent Now, I claim that this term is negative. The argument that this term is equal to 0 if $\bar{\alpha}(\eta)$ is replaced by $\alpha(\eta, x^{*})$ where $x^{*}$ is an eigenvalue at the water level is the same, since $\alpha(\eta, \lambda_{i})$ coincides with $\alpha(\gamma,\lambda_{i})$ and $\beta(\eta, \lambda_{i})$ coincides with $\beta(\gamma,\lambda_{i})$. However, it also holds that decreasing $\lambda_{i}$ increases $\beta(\eta, \lambda_{i})$, and that $\bar{\alpha}(\eta)$ is in fact weakly below the value of $\alpha(\eta, \lambda_{i})$ corresponding to the water-level eigenvalue. From this, we have that $T(\eta)$ is decreasing, and the same argument as in the proof of Theorem \ref{prop:gamma} thus implies that the comparative static holds.

To show that the optimal team size converges to the zero-attention solution as $\eta \rightarrow 0$, the same proof strategy replicates although the details differ slightly more: The same upper bound on $\sum_{i=1}^{n(\eta)} \left( \frac{\tilde{\lambda}_{i}(\eta)}{4k\mu(\eta)}-1 \right)$ applies directly, so that the claim follows provided $n(\eta) \rightarrow \infty$. To show this, I again assume a contradiction, namely that there is a sequence $\eta_m\rightarrow0$ along which $n(\eta_m)$ is bounded. Passing to a subsequence, suppose $n(\eta_m)=\bar n$ for all $m$. Note that for any active eigenvalue: 

\begin{equation*} 
\frac{\tilde{\lambda}_{i}}{4k\mu(\eta)}= \frac{\lambda_{i}(\eta - 2 \mu(\eta))}{4k\mu(\eta)(\eta + \lambda_{i}f(k))}.
\end{equation*}

\noindent Now, the entropy constraint implies: 

\begin{equation*} 
\frac{n(\eta)}{2} \log \left( \frac{ \eta - 2 \mu(\eta)}{4k \mu(\eta)} \right) + \frac{1}{2}\sum_{i=1}^{n(\eta)} \log \left( \frac{\lambda_{i}}{\eta + \lambda_{i}f(k)} \right) = \rho. 
\end{equation*}

But since $\frac{\lambda_{i}}{\eta+ \lambda_{i}f(k)}$ converges as $\eta \rightarrow 0$, it follows that as $\eta \rightarrow 0$ since $n(\eta) \rightarrow \bar{n}$, 

\begin{equation*} 
\frac{ \eta - 2 \mu(\eta)}{4k \mu(\eta)} \rightarrow \bar{M} \Rightarrow \frac{\tilde{\lambda}_{i}(\eta)}{4k\mu(\eta)} \rightarrow  \frac{\bar{M}}{f(k)}.
\end{equation*}

Thus applying the limit as $\eta \rightarrow 0$ to the entropy constraint, we have: 

\begin{equation*} 
\bar{n} \log \left( \frac{\bar{M}}{f(k)} \right) = 2 \rho \rightarrow \bar{M} > f(k). 
\end{equation*}

\noindent Now consider $\tilde{\lambda}_{n(\eta)+1}$. It holds that: 

\begin{equation*} 
\frac{\frac{1}{4k} \tilde{\lambda}_{n(\eta)+1}(\eta)}{\mu(\eta)} =\left( \frac{\eta - 2 \mu(\eta)}{4k \mu(\eta)}+ \frac{1}{2k} \right) \cdot \frac{ \lambda_{n(\eta)+1}}{\eta  +  \lambda_{n(\eta)+1}\left( \frac{1}{2k} + f(k) \right)}. 
\end{equation*}

\noindent Now, as $n(\eta) \rightarrow \bar{n}$, we have $\lambda_{n(\eta)+1}= \lambda_{\bar{n}+1}$ along the subsequence, using that there are infinitely many $\lambda_{i}$ such that $\lambda_{i} > 0$. Thus, we can evaluate this limit as: 

\begin{equation*} 
\lim_{\eta \rightarrow 0} \frac{\frac{1}{4k} \tilde{\lambda}_{n(\eta)+1}(\eta)}{\mu(\eta)} = \frac{\bar{M} + \frac{1}{2k}}{\frac{1}{2k}+f(k)}. 
\end{equation*} 

However, since $\bar{M} > f(k)$, $\frac{\bar{M} + \frac{1}{2k}}{\frac{1}{2k}+f(k)} > 1$. But this is a contradiction, since in order for the $n(\eta)+1$st eigenvalue to be inactive, we must have $\frac{\frac{1}{4k} \tilde{\lambda}_{n(\eta)+1}(\eta)}{\mu(\eta)} \leq 1$ for all $\eta$. This contradiction establishes that $n(\eta) \rightarrow \infty$; together with the previous arguments, it thus follows that the optimal team size converges to the zero-attention solution as $\eta \rightarrow 0$. \end{proof}

\begin{proof}[Proof of Lemma \ref{lem:mergedKL}] Assume $r > 0$; the case of $r=0$ is immediate, and the argument is identical when $r <0$. For some $\alpha$ to be determined, let: 

\begin{multline*} 
\tilde{\phi}_{M}(x) = \sqrt{k_{A}}\cos(\alpha) \mathbf{1}[x \in [0,1/k_{A}]] +\sqrt{k_{B}}\sin(\alpha) \mathbf{1}[x \in [2,2+1/k_{B}]], \\ \tilde{\phi}_{m}(x) = -\sqrt{k_{A}}\sin(\alpha) \mathbf{1}[x \in [0,1/k_{A}]] +\sqrt{k_{B}}\cos(\alpha) \mathbf{1}[x \in [2,2+1/k_{B}]]. 
\end{multline*}

\noindent Note that for all $\alpha$, $\tilde{\phi}_{M}$ and $\tilde{\phi}_{m}$ both have norm 1 and are orthogonal, and furthermore orthogonal to all other divisional eigenfunctions, since these functions are linear combinations of the normalized constant functions on the two post-team domains, while the remaining divisional eigenfunctions are orthogonal to those constant functions.  Letting $Z_{+}$ and $Z_{-}$ be independent standard normal variables, consider $W^{\circ}(x)=\sqrt{\lambda_{+}}Z_{+} \tilde{\phi}_{M}(x)+\sqrt{\lambda_{-}}Z_{-} \tilde{\phi}_{m}(x)$. Then for $x \in [0,1/k_{A}]$, $\text{Var}(W^{\circ}(x))=k_{A}(\cos(\alpha)^{2}\lambda_{+}+\sin(\alpha)^{2}\lambda_{-})$. Since $\cos^{2}(\alpha) = \frac{1+ \cos(2\alpha)}{2}$ and $\sin^{2}(\alpha) = \frac{1- \cos(2 \alpha)}{2}$, this is equal to $k_{A}(\frac{\lambda_{+} + \lambda_{-}}{2}+ \frac{\lambda_{+}- \lambda_{-}}{2}\cos(2 \alpha))$. Similarly, if $x \in [2,2+1/k_{B}]$, this is equal to $k_{B}(\frac{\lambda_{+} + \lambda_{-}}{2} - \frac{\lambda_{+} - \lambda_{-}}{2} \cos(2 \alpha)).$ Meanwhile, for $x \in [0,1/k_{A}]$ and $x' \in [2,2+1/k_{B}]$: 

\begin{equation*} 
\text{Cov}(W^{\circ}(x), W^{\circ}(x'))=\sqrt{k_{A}k_{B}}(\lambda_{+} \cos(\alpha)\sin(\alpha)-  \lambda_{-} \sin(\alpha) \cos(\alpha)) = (\lambda_{+} - \lambda_{-})\frac{\sin(2\alpha)}{2}\sqrt{k_{A}k_{B}},
\end{equation*}

\noindent where the last equality uses the double angle identity. Now, 

\begin{equation*} 
\lambda_{+} - \lambda_{-} = \sqrt{(\sigma_{A}^{2}-\sigma_{B}^{2})^{2}+ 4 r^{2} \sigma_{A}^{2} \sigma_{B}^{2}}>0, \lambda_{+} + \lambda_{-}=\sigma_{A}^{2} + \sigma_{B}^{2}.
\end{equation*}

\noindent So, let $\alpha$ be such that: 

\begin{equation*} 
\cos(2 \alpha) =  \frac{\sigma_{A}^{2}- \sigma_{B}^{2}}{\lambda_{+}-\lambda_{-}} ,  \qquad \sin(2 \alpha) = \frac{2r\sigma_{A}\sigma_{B}}{\lambda_{+}- \lambda_{-}}.
\end{equation*}

\noindent Indeed, such an $\alpha$ is guaranteed to exist since $\left(\frac{\sigma_{A}^{2}- \sigma_{B}^{2}}{\lambda_{+}-\lambda_{-}} \right)^{2} + \left(\frac{2r\sigma_{A}\sigma_{B}}{\lambda_{+}- \lambda_{-}} \right)^{2}=1.$ Then by the above, for $x \in [0,1/k_{A}]$, 

\begin{equation*} 
\text{Var}(W^{\circ}(x))=k_{A} \cdot \left(\frac{\sigma_{A}^{2}+\sigma_{B}^{2}}{2}+ \frac{\sigma_{A}^{2} - \sigma_{B}^{2}}{2} \right) = \frac{\sigma^{2}}{4}.
\end{equation*}

\noindent Similarly,  for $x' \in [2,2+1/k_{B}]$, $\text{Var}(W^{\circ}(x'))=\frac{\sigma^{2}}{4}$; and for any such $x, x'$, $\text{Cov}(W^{\circ}(x), W^{\circ}(x'))=r \sigma_{A} \sigma_{B} \sqrt{k_{A}}\sqrt{k_{B}}=r \frac{\sigma^{2}}{4}$. It follows that the process $W^{\circ}(x)$ has the same distribution as the common-shock component of the team-state process,  where the team-state common shocks are
$\frac{\theta_{A}}{2}$ on $[0,1/k_{A}]$ and $\frac{\theta_{B}}{2}$ on $[2,2+1/k_{B}]$ (where division by 2 reflects the worker's actions rather than the task state itself). Therefore, appending $W^{\circ}(x)$ to the Karhunen-Lo\'eve decomposition from the divisional problems when $\sigma=0$ yields a representation of the joint problem, completing the proof of the Lemma. \end{proof}

\begin{proof}[Proof of Proposition \ref{thm:mergerbetter}] I first consider an auxiliary problem where the organization devotes attention $\rho_{\theta}$ to the realizations correlated with $\theta_{A}$ and $\theta_{B}$. Suppose the eigenvalues associated with the two division states are $\nu_{1} \geq \nu_{2}$. If only one variable is in the attention problem, the posterior variance $\mu$ satisfies: 

\begin{equation*} 
\rho_{\theta}= \frac{1}{2} \log \left( \frac{\nu_{1}}{\mu} \right) \Rightarrow \mu = \nu_{1}e^{-2\rho_{\theta}},
\end{equation*}

\noindent so the total posterior variance is therefore $\nu_{1}e^{-2\rho_{\theta}} + \nu_{2}$. If both variables are paid attention to, then by the water-filling algorithm described in Proposition \ref{lem:RISol},

\begin{equation*} 
\rho_{\theta}= \frac{1}{2} \log \left( \frac{\nu_{1}\nu_{2}}{\mu^{2}} \right) \Rightarrow \mu = \sqrt{\nu_{1} \nu_{2}}e^{-\rho_{\theta}}.
\end{equation*}

\noindent Furthermore, the form of the water-filling algorithm implies that the first case applies if and only if $\mu > \nu_{2}$, which in turn occurs if and only if $\rho_{\theta}\leq \frac{1}{2} \log \left(\frac{ \nu_{1}}{\nu_{2}} \right)$.

Focusing on the special case where $\nu_{1}=x(1+d), \nu_{2} = x(1-d)$, note that $d \in (0,1)$ since $x=\frac{\nu_{1} + \nu_{2}}{2}$ and $d = \frac{\nu_{1} - \nu_{2}}{\nu_{1} + \nu_{2}}$. Define: 

\begin{equation*} 
A(x, d,\rho_{\theta})= \begin{cases} 
x(1+d)e^{-2 \rho_{\theta}}+x(1-d) & 0 \leq \rho_{\theta} \leq \frac{1}{2} \log \left( \frac{1+d}{1-d} \right) \\ 2 x\sqrt{(1+d)(1-d)}e^{-\rho_{\theta}} & \rho_{\theta} > \frac{1}{2} \log \left( \frac{1+d}{1-d} \right) .
\end{cases}
\end{equation*}

\noindent Next, I claim  $A(x, d, \rho_{\theta})$ is decreasing in $d$ for every $\rho_{\theta}$. Indeed, $A(x,d, \rho_{\theta})$ is continuous in $d$ since the two cases coincide when $\rho_{\theta} = \frac{1}{2} \log\left(\frac{1+d}{1-d} \right)$. The derivative with respect to $d$ in the first case is $x(e^{-2 \rho_{\theta}} -1)< 0$, and in the second case it is $- \frac{2x d e^{-\rho_{\theta}}}{ \sqrt{(1-d)(1+d)}} < 0$. Since the function $A(x,d, \rho_{\theta})$ is continuous in $d$ and decreasing in each case, it is strictly decreasing overall when $\rho_{\theta}>0$. 

Now, note that that $\lambda_{+} - \lambda_{-}= \sqrt{(\sigma_{A}^{2} - \sigma_{B}^{2})^{2} + 4r^{2} \sigma_{A}^{2}\sigma_{B}^{2}}$, which is greater than $\abs{\sigma_{A}^{2} - \sigma_{B}^{2}}$ whenever $r, \sigma_{A}^{2}, \sigma_{B}^{2} > 0$. Furthermore, let $V_{rest}(\tilde{\rho})$ denote the coordination loss in the organization's problem where $r=0$ given an attention budget $\tilde{\rho}$. Then the organization's coordination loss is: 

\begin{equation} 
\min_{\rho_{\theta} \in [0, \rho]} A\left(\frac{\sigma_{A}^{2}+\sigma_{B}^{2}}{2}, d, \rho_{\theta} \right)+V_{rest}(\rho-\rho_{\theta}), \label{eq:losses}
\end{equation}

\noindent since the organization's payoff in the unconstrained problem is the sum of the losses in each subproblem, with the organization free to set the attention in the subproblem freely. 

Now compare coordination losses in the case of the merged and decentralized organizations---importantly, assuming that team sizes are exogenously fixed to be as in the optimal decentralized case. I allow these to adjust later. Let $\rho_{\theta}^{S}$ denote the solution to (\ref{eq:losses}) in the separated problem and $\rho_{\theta}^{M}$ to denote the solution to (\ref{eq:losses}) in the merged problem. Let $d^{S}$ and $d^{M}$ denote the values of $d$ corresponding to the decentralized and merged problems, respectively.  By the above, $d^{S} < d^{M}$ whenever $r \neq 0$.  Furthermore, 

\begin{align*} 
A\left(\frac{\sigma_{A}^{2}+\sigma_{B}^{2}}{2}, d^{S}, \rho_{S} \right)+V_{rest}(\rho-\rho_{S}) & \geq
A\left(\frac{\sigma_{A}^{2}+\sigma_{B}^{2}}{2}, d^{M}, \rho_{S} \right)+V_{rest}(\rho-\rho_{S})   \\ & \geq A\left(\frac{\sigma_{A}^{2}+\sigma_{B}^{2}}{2}, d^{M}, \rho_{M} \right)+V_{rest}(\rho-\rho_{M}).
\end{align*}

\noindent Furthermore, a strict inequality obtains if $\rho_{M} \neq 0$; if $\rho_{S} > 0$, the first inequality is strict, and $\rho_{S} > 0$ implies $\rho_{M} > 0$. If $\rho_{S}=0$ and $\rho_{M} \neq 0$, then adjusting attention must improve the organization's objective. Furthermore, $\rho_{M} > 0$ if and only if the common shock enters into the organization's objective in the merged problem. Thus, since the organization does better when merged but assuming team sizes are fixed to be as in the decentralized problem, allowing team sizes to readjust can only improve the organization's objective, since the adaptation losses for any fixed team sizes are the same for the merged and the decentralized cases.    \end{proof}

The results on team sizes under mergers make use of the following Lemma: 

\begin{lemma}  \label{lem:teambound}
Define $k_{0}$ as: 

\begin{equation*} 
k_{0} = \argmin \frac{1}{4k}+ \frac{f(k)}{2}.
\end{equation*}

\noindent Then each division's optimal team size is weakly less than $k_{0}$, either when separated or merged. 
\end{lemma}

\noindent This result follows immediately from the form of the objective in the baseline model, and hence for the separated case as well, since additional attention only decreases the marginal benefit from larger team size. The proof below shows that this conclusion also holds in the merged case. 

\begin{proof}[Proof of Lemma \ref{lem:teambound}]
I show that  if $R$ is the coordination loss, then: 

\begin{equation*} 
- \frac{\partial R}{\partial k_{j}} \leq \frac{\sigma^{2}+ \sum_{i} \lambda_{i}^{j}}{4k_{j}^{2}},
\end{equation*}

\noindent which implies that the optimal $k_{j}$ must be smaller than $k_{0}$, since the marginal cost of $k$ is the same as the no-information case, so that if marginal benefits from increasing $k$ are lower at every $k$, optimality cannot hold at any $k > k_{0}$. I show this for all possible cases regarding the solution to the attention problem; the envelope theorem then implies the same bound after attention is optimally allocated across all modes.

If neither $\lambda_{+}$ or $\lambda_{-}$ are in the attention problem, then this is immediate since $\lambda_{+} + \lambda_{-} = \frac{\sigma^{2}}{4k_{A}}+ \frac{\sigma^{2}}{4k_{B}}$. So suppose that at least one of the modes is attended to, and as in the proof of Proposition \ref{thm:mergerbetter} let $\rho_{\theta}$ denote the attention to these terms. Let $\sigma_{A}^{2}$ and $\sigma_{B}^{2}$ be as in the statement of Lemma \ref{lem:mergedKL}. If both $\lambda_{+}$ and $\lambda_{-}$ are in the attention problem, then their residual variance is $2\mu=2 \sqrt{\sigma_{A}^{2}\sigma_{B}^{2}(1-r^{2})}e^{-\rho_{\theta}}$; the derivative with respect to $\sigma_{A}^{2}$ can be written as $\mu/\sigma_{A}^{2}$; but since by assumption $\mu \leq \sigma_{A}^{2}$, this is less than 1. Since in the no-information problem the derivative of the coordination loss with respect to $\sigma_{A}^{2}$ is exactly 1, the desired inequality holds.

Now suppose only $\lambda_{+}$ is in the attention problem. In this case the residual variance from these terms is $\lambda_{+}e^{-2\rho_{\theta}}+\lambda_{-}$, which is: 

\begin{equation*} 
\frac{\sigma_{A}^{2} + \sigma_{B}^{2}}{2}(1+e^{-2\rho_{\theta}})-\frac{1}{2}(1-e^{-2\rho_{\theta}}) \sqrt{(\sigma_{A}^{2}-\sigma_{B}^{2})^{2} + 4 \sigma_{A}^{2} \sigma_{B}^{2} r^{2}}. 
\end{equation*}

\noindent The derivative is:  

\begin{equation*} 
\frac{1+e^{-2\rho_{\theta}}}{2} - \frac{1-e^{-2 \rho_{\theta}}}{2} \cdot \frac{\sigma_{A}^{2} - \sigma_{B}^{2}+2 \sigma_{B}^{2} r^{2}}{\sqrt{(\sigma_{A}^{2}-\sigma_{B}^{2})^{2} + 4 \sigma_{A}^{2} \sigma_{B}^{2} r^{2}}}.
\end{equation*}

\noindent But $\abs{\sigma_{A}^{2} - \sigma_{B}^{2}+2 \sigma_{B}^{2} r^{2}} \leq \sqrt{(\sigma_{A}^{2}-\sigma_{B}^{2})^{2} + 4 \sigma_{A}^{2} \sigma_{B}^{2} r^{2}}$; indeed, squaring both sides of this inequality reduces to $4 \sigma_{B}^{4}(1- r^{2})r^{2} \geq 0$. If $\sigma_{A}^{2}-\sigma_{B}^{2}+2 \sigma_{B}^{2} r^{2} \geq 0$, then the ratio lies between $0$ and $1$. Therefore the derivative is bounded above by
\begin{equation*}
\frac{1+e^{-2 \rho_{\theta}}}{2}<1,
\end{equation*}

\noindent since by assumption $\rho_{\theta} > 0$. If $\sigma_{A}^{2}-\sigma_{B}^{2}+2 \sigma_{B}^{2} r^{2} \leq 0$, then the ratio lies between $-1$ and $0$, so the derivative is bounded above by
\begin{equation*}
\frac{1+e^{-2 \rho_{\theta}}}{2}+\frac{1-e^{-2 \rho_{\theta}}}{2}=1.
\end{equation*} Thus, the derivative is bounded above by one, completing the proof. 
\end{proof}

\begin{proof}[Proof of Proposition \ref{prop:teamsize}] For any optimal symmetric solution $k_{A}=k_{B}$, say $k_{reg}$ when the regime is $reg \in \{dec, mer\} $  (decentralized and merged, respectively), with the same adaptation costs and eigenvalues, write the organization's loss as a function of $reg$ and $k_{reg}$ as: 

\begin{equation*} 
L^{reg}(k_{reg})=\frac{\alpha_{\rho, reg}}{4k_{reg}} + \frac{f(k_{reg})}{2} \left(2 \sigma^{2} + \sum_{j=1}^{\infty} \lambda_{j}^{A} + \lambda_{j}^{B} \right),
\end{equation*}

\noindent for some fixed $\alpha_{\rho, reg}.$ Note $\frac{\partial^{2} L^{reg}}{\partial \alpha_{\rho,reg} \partial k_{r}} < 0$. Since by assumption $k_{dec}$ is interior, the first order condition is satisfied: 

\begin{equation*} 
- \frac{\alpha_{\rho,dec}}{4 k_{dec}^{2}}+ \frac{f'(k_{dec})}{2}\left(2 \sigma^{2} + \sum_{j=1}^{\infty} \lambda_{j}^{A} + \lambda_{j}^{B} \right)=0. 
\end{equation*}

\noindent So, \begin{equation*} \frac{\partial L^{mer}}{\partial k} \lvert_{k=k_{dec}} = - \frac{\alpha_{\rho,mer}}{4k_{dec}^{2}} + \frac{f'(k_{dec})}{2}\left(2 \sigma^{2} + \sum_{j=1}^{\infty} \lambda_{j}^{A} + \lambda_{j}^{B} \right)= \frac{\alpha_{\rho,dec}- \alpha_{\rho,mer}}{4k_{dec}^{2}} \geq 0. \end{equation*} 

Convexity of $L^{mer}(k)$ in $k$ implies that the symmetric merged minimizer is weakly below $k_{dec}$, strictly if $\alpha_{\rho,mer} < \alpha_{\rho,dec}$. \end{proof}

\noindent \textit{Example Showing a Division's Team Size can Increase under Mergers with Asymmetries} Take $f(k)=1+\beta(k-1)$ with $\beta=1/10$, which I note involves the $\rho=0$ team size being $\sqrt{5} \approx 2.236$. Suppose $\lambda_{1}^{A}=2$, $\sigma^{2}=1$, and $\lambda_{1}^{B}=1/2$; in this example, I take all other eigenvalues to be 0, although this is not essential. 

Take $\rho=1/2$, and solve for the optimal team size in the unmerged case. The first step, in particular, is to verify that in this example that all attention is allocated to division A. This immediately implies that division B's optimal team size is $\sqrt{5}$. Given any such $\mu^{*}$ such that $\mu^{*} < \sigma, \lambda_{1}^{A}$, the first-order condition implies that the optimal team size of division A is $\sqrt{\frac{10}{3} \mu^{*}}$. When $\rho=1/2$, $\mu^{*}$ solves the entropy constraint $1/2= \frac{1}{2} \left[ \log \left( \frac{2}{\mu^{*}} \right) + \log \left( \frac{1}{\mu^{*}} \right) \right]$, so $\mu^{*}=\sqrt{2/e} \approx 0.85776$. This corresponds to $\mu\approx 0.1268$, which is larger than $1/(4 \cdot \sqrt{5}))$.  Thus, not only does this identify the solutions for team sizes---1.6909 for division A and $\sqrt{5}$ for division $B$---but also verifies that the solution involves division A receiving all attention. 

Now consider the merged case with $r=4/5$.  The individual division eigenvalues are $\frac{1}{2k_{A}}$ for division $A$ and $\frac{1}{8k_{B}}$ for division B,  with the merged eigenvalues being: 

\begin{equation*} 
\lambda_{+} = \frac{1}{8k_{A}}+ \frac{1}{8k_{B}} + \frac{1}{2}\frac{\sqrt{25k_{A}^{2} + 14k_{A}k_{B}+25k_{B}^{2}}}{20k_{A}k_{B}},  \lambda_{m} = \frac{1}{8k_{A}}+ \frac{1}{8k_{B}} - \frac{1}{2} \frac{\sqrt{25k_{A}^{2} + 14k_{A}k_{B}+25k_{B}^{2}}}{20k_{A}k_{B}}. 
\end{equation*}

\noindent But $\lambda_{m}$ is decreasing in $k_{A}$ and $k_{B}$ for all $k_{A}, k_{B} \geq 1$,  so that the highest possible value for $\lambda_{m}$ is $0.05$.   Similarly,  $\lambda_{1}^{B}/(4k_{ah}) \leq 1/8$.  So consider the problem when $\lambda_{+}$ and $\lambda_{1}^{A}/(4k_{A})$ are in the attention problem.  The objective is: 

\begin{equation*} 
2 \mu(k_{A}, k_{B}) + \lambda_{m}(k_{A}, k_{B}) + \frac{\lambda_{1}^{B}}{4k_{B}} + \frac{3}{2}(1+\beta(k_{A}-1))+\frac{3}{4}(1+\beta (k_{B}-1))
\end{equation*}  

\noindent Minimizing this numerically, I find that $k_{A} \approx 1.7331$ and $k_{B} \approx 2.00095$.   Thus, division A has larger teams under the merger. Furthermore, at this solution,  $\mu \approx 0.160403 < 1/(2k_{A})$,  so that indeed both of these eigenvalues enter into headquarters' attention problem.  Since optimal team sizes are bounded above by the no-information team size, $\sqrt{5}$, the inequalities that define the active set hold so that the above objective is the relevant one for the merged problem, implying that this is the global optimum.

\begin{proof}[Proof of Proposition \ref{prop:training}] 
In the model with training, $a_{x}$ is chosen to maximize $-c(\alpha) a_{x}^{2} -(\beta(\alpha) a_{x} - W_{0}(x))^{2}$, so that: 

\begin{equation*} 
a_{x} = \overbrace{\frac{\beta(\alpha)}{\beta(\alpha)^{2}+c(\alpha)}}^{:=h(\alpha)} W_{0}(x). 
\end{equation*}

\noindent Therefore, the baseline stochastic process is $\frac{\beta(\alpha)}{\beta(\alpha)^{2}+c(\alpha)}W_{0}(x)$. Substituting in to $\ell_{x}(a, W_{0}(x),\alpha)$: 

\begin{equation*} 
l(W_{0}(x))= f(k) \cdot \overbrace{\frac{c(\alpha)}{\beta(\alpha)^{2} + c(\alpha)}}^{:=r(\alpha)} \cdot W_{0}(x)^{2} . 
\end{equation*}

\noindent Therefore, $\alpha$ and $k$ are chosen to maximize:

\begin{equation*} 
-\frac{1}{2}\alpha^{2} - n\frac{\mu^{*}}{k} \cdot h(\alpha)^{2} - \sum_{m=n+1}^{\infty} \frac{1}{k} \lambda_{m} \cdot h(\alpha)^{2} - f(k) r(\alpha) \sum_{m=1}^{\infty} \lambda_{m}. 
\end{equation*}

\noindent Consider a constant-factor increase in the variance of $W_{0}(x)$. The optimal value of $k$ satisfies a first-order condition that does not depend on $\gamma$. Thus, the comparative statics of team size depend on (i) the cross-partial of the first order condition with respect to $\alpha$, and (ii) how $\gamma$ changes $\alpha$.  

The first-order condition with respect to $k$ is: 

\begin{equation*} 
 n\frac{\mu^{*}}{k^{2}} \cdot h(\alpha)^{2} +\sum_{m=n+1}^{\infty} \frac{1}{k^{2}} \lambda_{m} \cdot h(\alpha)^{2} = f'(k) r(\alpha) \sum_{m=1}^{\infty} \lambda_{m}. 
\end{equation*}

Differentiating with respect to $\alpha$ yields: 

\begin{equation*} 
 n\frac{\mu^{*}}{k^{2}} 2h(\alpha) \cdot h'(\alpha) +\sum_{m=n+1}^{\infty} \frac{1}{k^{2}} \lambda_{m} 2h(\alpha) \cdot h'(\alpha) - f'(k) r'(\alpha) \sum_{m=1}^{\infty} \lambda_{m}. 
\end{equation*}

\noindent The first-order condition delivers an expression for $n\frac{\mu^{*}}{k^{2}} +\sum_{m=n+1}^{\infty} \frac{1}{k^{2}} \lambda_{m}$ which we can substitute in to the cross-partial:

\begin{equation*} 
\frac{f'(k)r(\alpha) \sum_{m=1}^{\infty} \lambda_{m}}{h(\alpha)^{2}} \cdot 2h(\alpha)h'(\alpha) - f'(k) r'(\alpha) \sum_{m=1}^{\infty} \lambda_{m} \propto 2\frac{r(\alpha)}{h(\alpha)} h'(\alpha) - r'(\alpha),
\end{equation*}

\noindent where proportionality holds since $f'(k), \sum_{m=1}^{\infty} \lambda_{m} > 0$. Using that $r(\alpha)/h(\alpha) = c(\alpha)/\beta(\alpha)$, and clearing the common denominator of $h'(\alpha)$ and $r'(\alpha)$, this expression has the same sign as: 

\begin{multline*} 
\frac{c(\alpha)}{\beta(\alpha)}  \left((\beta(\alpha)^{2} + c(\alpha))\beta'(\alpha) -(2\beta(\alpha)\beta'(\alpha)+c'(\alpha))\beta(\alpha) \right) \\- \frac{1}{2}  \left((\beta(\alpha)^{2} + c(\alpha))c'(\alpha)-(2\beta(\alpha)\beta'(\alpha)+c'(\alpha))c(\alpha) \right). 
\end{multline*}

\noindent Simplifying this expression yields the equation in the proposition statement. 

Now consider how $\alpha$ changes with respect to $\gamma$. Write headquarters' objective as $V(k, \alpha;\gamma)=-\frac{1}{2}\alpha^{2} - \gamma \cdot U(k,\alpha)$; in particular, the first-order conditions can be written:  \begin{equation*} [k]: U_{k}=0,  ~~~~ [\alpha]: \alpha  + \gamma U_{\alpha}=0
\end{equation*} Then, differentiating the first-order conditions with respect to $\gamma$:

\begin{equation*} 
\begin{pmatrix}  U_{kk} & U_{k\alpha} \\ \gamma U_{\alpha k}  & 1+ \gamma U_{\alpha \alpha} \end{pmatrix} \cdot  \begin{pmatrix}  \frac{d k }{d\gamma } \\ \frac{d \alpha}{d\gamma} \end{pmatrix}  =  \begin{pmatrix} 0 \\  -U_{\alpha} \end{pmatrix}.
\end{equation*}

Thus, applying Cramer's rule, $\frac{d \alpha}{d \gamma} =  -\frac{ U_{kk} U_{\alpha}}{U_{kk}(1+\gamma U_{\alpha \alpha}) - \gamma U_{k \alpha} U_{\alpha k}}$. Note that the denominator is the determinant of the Hessian of $V$ divided by $\gamma$; thus, assuming the second-order condition for a strict local maximum, so that the Hessian is negative definite, it follows that $\frac{d\alpha}{d\gamma}$ has a sign opposite of $U_{\alpha}$.  

I now show that necessary conditions hold when $\beta(\alpha) = \sqrt{1+\alpha}$ and $c(\alpha)=1+\alpha$. Here, $r(\alpha)=1/2$ and $h(\alpha)^{2} = \frac{1}{4(1+\alpha)}$; thus: 
\begin{multline*} 
U(k,\alpha) = \left(n \frac{\mu^{*}}{k} + \sum_{m=n+1}^{\infty} \frac{\lambda_{m}}{k} \right) \frac{1}{4(1+\alpha)} + \frac{f(k)}{2} \sum_{m=1}^{\infty} \lambda_{m} \\ \Rightarrow U_{\alpha} = -\left(n \frac{\mu^{*}}{k} + \sum_{m=n+1}^{\infty} \frac{\lambda_{m}}{k} \right) \frac{1}{4(1+\alpha)^{2}}  < 0,
\end{multline*}  so that $\frac{d \alpha}{d \gamma} > 0$ by the above. Furthermore, inspecting the objective: 

\begin{equation*} 
U_{kk} = \frac{n\mu^{*} + \sum_{m=n+1}^{\infty} \lambda_{m}}{2k^{3}(1+ \alpha)} + \frac{\sum_{m=1}^{\infty} \lambda_{m}}{2}f''(k), U_{\alpha \alpha} = \frac{n\mu^{*} + \sum_{m=n+1}^{\infty} \lambda_{m}}{2k(1+ \alpha)^{3}} , U_{k\alpha} = \frac{n\mu^{*} + \sum_{m=n+1}^{\infty} \lambda_{m}}{4k^{2}(1+ \alpha)^{2}}.
\end{equation*}

On the other hand, negative definiteness is equivalent to (i)  $U_{kk} > 0$ and  (ii) $U_{kk}(1+\gamma U_{\alpha \alpha}) - \gamma U_{k \alpha}^{2} > 0$; the former follows immediately from the previous, and the latter follows from noting that $U_{kk} > 0, \gamma  > 0$ and: 

\begin{equation*} U_{kk}U_{\alpha\alpha}  -U_{k\alpha}^{2} > \frac{(n\mu^{*} + \sum_{m=n+1}^{\infty} \lambda_{m})^{2}}{4k^{4}(1+ \alpha)^{4}}-\frac{(n\mu^{*} + \sum_{m=n+1}^{\infty} \lambda_{m})^{2}}{16k^{4}(1+ \alpha)^{4}}  > 0. \end{equation*}

\noindent Thus, the necessary conditions for the comparative statics hold in this specification, completing the proof. \end{proof}

\begin{proof}[Proof of Proposition \ref{prop:costlylink}]
Write the objective as $\gamma U(k) - \kappa \frac{k-1}{2}$, where 
\begin{equation*} U(k) = -\left(\frac{1}{4k}\left(n\mu^{*}+\sum_{m=n+1}^{\infty}\lambda_{m}\right) + \frac{f(k)}{2}\sum_{m=1}^{\infty}\lambda_{m}\right);\end{equation*} note $U$ is strictly concave. Let $\tilde{k}_{0} = \arg\max_{k} U(k)$, i.e., the optimal team size absent link-formation costs. For any $\gamma >0$, the optimal team size lies in $[1,\tilde{k}_{0}]$: for $k > \tilde{k}_{0}$, lowering $k$ to $\tilde{k}_{0}$ raises $\gamma U(k)$ and lowers the link-formation cost. On $[1,\tilde{k}_{0}]$, $U$ is increasing, so for $\gamma' > \gamma$ the difference $(\gamma'-\gamma)U(k)$ is increasing in $k$; the objective therefore has increasing differences in $(k,\gamma)$ on the relevant region, and hence the optimal team size is weakly increasing in $\gamma$ by \cite{Topkis1998}. If the optimum at $\gamma$ is interior, it satisfies the first-order condition $\gamma U'(k^{*})=\kappa/2 > 0$. An increase to $\gamma' > \gamma$ requires $U'$ to fall to $\kappa/2\gamma' < \kappa/2\gamma$ at the new optimum, which remains interior since it is weakly larger; strict concavity of $U$ then implies the new team size is strictly larger.
\end{proof}

\newpage

\section{Relating Eigenvalue Decay to Jaggedness} \label{app:jaggedness}

This appendix provides sufficient conditions under which, for a fixed eigenbasis, jaggedness is decreasing in the decay rate of the covariance eigenvalues. The conclusion requires conditions ensuring that higher-index eigenfunctions tend to vary more rapidly across nearby locations and that this variation does not vanish locally.

\begin{align*} 
\mathbb{E}[(W(x+h)-W(x))^{2}] &= \mathbb{E} \left[ \left(\sum_{k=1}^{\infty}\sqrt{\lambda_{k}} (\phi_{k}(x+h)-\phi_{k}(x)) Z_{k} \right)^{2} \right ]  \\ & =  \sum_{k=1}^{\infty} \lambda_{k}(\phi_{k}(x+h)- \phi_{k}(x))^{2},
\end{align*}

\noindent where the second line uses that the $Z_{k}$ sequence is IID standard normal.  Take as given that there is some $\overline{B}>0$ and $p >1$ such that: 

\begin{equation*} 
 \lambda_{k} \leq  \overline{B} \cdot k^{-p}, 
\end{equation*}

\noindent and some $q, \overline{C}_{x} >0$ such that: 

\begin{equation*} 
\abs{\phi_{k}(x+h) - \phi_{k}(x)} \leq   \overline{C}_{x} \cdot \min\{1,  hk^{q} \}; 
\end{equation*}

\noindent for all $k$ and sufficiently small $h$.  Thus: 

\begin{equation} 
\mathbb{E}[(W(x+h)-W(x))^{2}/h^{\alpha}] \leq \overline{B} \cdot \overline{C}_{x}^{2} \left( \sum_{k=1}^{k^{*}(h)} h^{2} k^{2q-p} +\sum_{k=k^{*}(h)+1}^{\infty} k^{-p}  \right)/h^{\alpha},   \label{eq:sum}
\end{equation}

\noindent where $k^{*}$ is the largest $k$ such that $h k^{q} \leq 1$.   Thus,

\begin{equation*}   \lim_{h \rightarrow 0}  \left( \sum_{k=1}^{k^{*}(h)} h^{2} k^{2q-p} +\sum_{k=k^{*}+1}^{\infty} k^{-p}  \right)/h^{\alpha} < \infty \Rightarrow 
\limsup_{h \rightarrow 0} \mathbb{E}[(W(x+h)-W(x))^{2}/h^{\alpha}] < \infty,
\end{equation*}

\noindent I pause to highlight the part of this argument that illustrates why $\alpha$ might vary in general: ignoring the $k > k^{*}+1$ terms---or, more precisely, assuming a finite-dimensional eigenbasis--- the right hand side of (\ref{eq:sum}) is finite for $\alpha=2$. However, while $\lambda_{k} \rightarrow 0$, typically $\abs{\phi_{k}'(x)}$ becomes large as $k \rightarrow \infty$. Thus, as $h \rightarrow 0$ and $k^{*}(h) \rightarrow \infty$, these ``higher-frequency'' terms may still contribute nontrivially to the overall sum. Roughly speaking, dividing by $h^{\alpha}$ instead of $h^{2}$ for $\alpha < 2$ ensures that these contributions have diminished impact. 

Substituting in for $k^{*}(h)$: 

\begin{equation*} 
\lim_{h \rightarrow 0}  \left( \sum_{k=1}^{\lceil h^{-1/q}  \rceil } h^{2} k^{2q-p} +\sum_{k=\lceil h^{-1/q}  \rceil}^{\infty} k^{-p}  \right)/h^{\alpha}
\end{equation*}

\noindent I consider each sum separately; for simplicity, assume that $h^{-1/q}$ is an integer, as this does not influence the argument. For the second, I note for some $C_{t}$: 

\begin{equation*} 
\sum_{k= h^{-1/q} }^{\infty} k^{-p}  \leq \int_{ h^{-1/q}  -1}^{\infty} x^{-p}dx =\frac{(h^{-1/q}-1)^{1-p}}{p-1} \leq  C_{t} \cdot h^{(p-1)/q}. 
\end{equation*}

For the first term, first assume that $(p-1)/q < 2$, and note that for some $C_{l}$: 

\begin{equation*} 
 \sum_{k=1}^{ h^{-1/q}   } h^{2}k^{2q-p}  \leq h^{2}  \left(1+\int_{1}^{ h^{-1/q}  } x^{2q-p} dx \right) \leq C_{l} h^{2} \cdot h^{(p-2q-1)/q}= C_{l}h^{(p-1)/q}. 
\end{equation*}

\noindent Putting this together: 

\begin{equation*} 
\lim_{h \rightarrow 0}  \left( \sum_{k=1}^{ h^{-1/q}   } h^{2} k^{2q-p} +\sum_{k= h^{-1/q}  }^{\infty} k^{-p}  \right)/h^{\alpha} \leq \lim_{h \rightarrow 0}  (C_{t}+C_{l})h^{(p-1)/q-\alpha},
\end{equation*}

\noindent  where the right-hand side is finite whenever $\alpha \leq (p-1)/q$. Thus, provided $(p-1)/q < 2$, the $\lim \sup$ is finite for every $\alpha \leq (p-1)/q$, so that $\alpha^{*}(x) \geq (p-1)/q$. The reverse inequality---that the ratio diverges for every $\alpha > (p-1)/q$, so that $\alpha^{*}(x) \leq (p-1)/q$---follows from the matching lower-bound argument below. When $(p-1)/q = 2$, identical steps show that: 

\begin{equation*} 
 \sum_{k=1}^{ h^{-1/q}   } h^{2}k^{2q-p} \leq C_{l} h^{2} \log(1/h). 
\end{equation*}

For $(p-1)/q> 2$: 

\begin{equation*} 
 \sum_{k=1}^{ h^{-1/q}   } h^{2}k^{2q-p} \leq  h^{2}\left(1+\int_{1}^{\infty} x^{2q-p}dx \right) \leq C_{l}h^{2}. 
\end{equation*}

\noindent Thus, in these cases, for any $\alpha <2$, the right hand side of (\ref{eq:sum}) is finite. (In the boundary case where $(p-1)/q=2$, this term diverges for $\alpha=2$ but converges for all smaller $\alpha$; when $(p-1)/q>2$, this term converges). 

This argument provides the basic intuition for why eigenvalue decay is related to jaggedness---if eigenvalues decay more slowly, then the more ``erratic'' eigenfunctions have more weight, which demands a smaller $\alpha$ for the overall limit to converge. The reason this cannot be made into a full proof \emph{in general} is that it is typically not possible to have a matching lower bound on $\abs{\phi_{k}(x+h) - \phi_{k}(x)}$. For instance, it may be that this is 0 infinitely often at some $x$, even though $q$ is the smallest value such that the upper bound holds. However, it is possible to provide such a lower-bound in examples.

A sufficient condition that ensures that the bound is tight is that (i) $p$ also lower-bounds the decay rate of the eigenvalues, i.e., $\lambda_{k} \geq \underline{B} \cdot k^{-p} $ and (ii) for some $0<a <b<\infty$:\footnote{As above, assume $h^{-1/q}$ is an integer to save on notation.} 

\begin{equation*} 
\sum_{k=ah^{-1/q}}^{bh^{-1/q}} (\phi_{k}(x+h)-\phi_{k}(x))^{2} \geq c_{x} h^{-1/q}. 
\end{equation*}

\noindent I first show that these conditions do indeed imply the lower bound is tight, which in turn shows that $\alpha^{*}(x) =\min\{2,\frac{p-1}{q}\}$; in particular, $\alpha^{*}(x) \leq 2$ holds provided $\phi_{k}'(x) \neq 0$ for some $k$, since under this case $\frac{(\phi_{k}(x+h)- \phi_{k}(x))^{2}}{h^{2}} \cdot \frac{1}{h^{\alpha-2}} \rightarrow \infty$ as $h \rightarrow 0$ for any $\alpha > 2$. And indeed,  

\begin{align*} 
\mathbb{E}[(W(x+h)-W(x))^{2}] &\geq   \sum_{k=ah^{-1/q}}^{bh^{-1/q}} \lambda_{k} (\phi_{k}(x+h)-\phi_{k}(x))^{2}  \\ & \geq \underline{B}((bh^{-1/q})^{-p}) \sum_{k=ah^{-1/q}}^{bh^{-1/q}} (\phi_{k}(x+h)-\phi_{k}(x))^{2} \\ & \geq \underline{C}_{x} h^{p/q}h^{-1/q}=\underline{C}_{x} h^{(p-1)/q}.
\end{align*}

\noindent Thus, for any $\alpha > (p-1)/q$: 

\begin{equation*} 
\frac{\mathbb{E}[(W(x+h)-W(x))^{2}]}{h^{\alpha}} \rightarrow \infty. 
\end{equation*} Since jaggedness is defined using the supremum over $\alpha$ such that $\limsup_{h \rightarrow 0} \frac{\mathbb{E}[(W(x+h)-W(x))^{2}]}{h^{\alpha}} < \infty$, combining the upper and lower bounds yields: 

\begin{equation*} 
\alpha^{*}(x) = \min\{2, \frac{p-1}{q} \}.
\end{equation*}

As a result, since the eigenbasis and hence $q$ is fixed, faster eigenvalue decay weakly increases $\alpha^{*}(x)$ and hence cannot increase jaggedness. This change is strict whenever $1 < p < 1+2q$; taking $p \leq 1$ violates the assumption that the integrated variance is finite, while $\alpha^{*}(x)=2$ for all $p \geq 1+2q$. To show that this condition is indeed not vacuous, I verify it for the eigenbasis corresponding to Brownian motion: 

\begin{equation*} 
\phi_{k}(x) = \sqrt{2} \sin \left((k-1/2)\pi x \right). 
\end{equation*}

\noindent Using that $\sin(x)-\sin(y)=2 \cos((x+y)/2)\sin((x-y)/2)$: 

\begin{equation*} 
\phi_{k}(x+h)- \phi_{k}(x) =2\sqrt{2} \cos((k-1/2)\pi(x+h/2)) \sin\left((k-1/2)\pi h/2 \right). 
\end{equation*}

\noindent Letting $a=1/3$ and $b=2/3$ for $h$ sufficiently small, $(k-1/2)\pi h/2 \in [\pi/8, \pi/3]$. Since $\sin((k-1/2)\pi h/2)$ is bounded away from 0 in this region, there exists a constant $\underline{c}_{s}$ such that $\sin\left((k-1/2)\pi h/2 \right)^{2} > \underline{c}_{s}$. 

For the $\cos$ term, note that for any $x \in (0,1)$ there exists a $\bar{h}_{x}$ sufficiently small such that whenever $h < \bar{h}_{x}$, $x+h/2 \in [\bar{h}_{x}, (1- \bar{h}_{x}) ]$. Let $\eta=\pi \bar{h}_{x}/4$. Suppose that, at some $k$ for some $m$ $\abs{(k-1/2)\pi(x+h/2)-m\pi-\pi/2}< \eta$. Note that the increment between successive $(k-1/2)\pi(x+h/2)$ terms lies in $[4\eta, \pi-4\eta]$. Thus, if $\abs{(k-1/2)\pi(x+h/2)-m\pi-\pi/2}< \eta$, then there cannot be any $m'$ such that $\abs{(k+1/2)\pi(x+h/2)-m'\pi-\pi/2}< \eta$; indeed, this would imply that $(k-1/2)\pi(x+h/2)$ is in an interval of width $2\eta$ centered at $(m+1/2)\pi$, while  $(k+1/2)\pi(x+h/2)$ is at least $\bar{h}_{x}\pi> 2\eta$ larger. Thus, $(k+1/2)\pi(x+h/2)$ cannot be within $\eta$ of $m\pi+\pi/2$. However, to be in the next $2\eta$-length interval around $\pi/2$ modulo $\pi$, it would need to be at least $\pi-2 \eta$ larger (noting that this is exactly what the increase would need to be if at the right endpoint of the region, then taking it to the left endpoint of the next region). However, it is no larger than $\pi-4\eta$, so that this is impossible.  It follows that in the sum: 

\begin{equation*} 
\sum_{k=1}^{\infty} \cos((k-1/2)\pi(x+h/2))^{2},
\end{equation*}

\noindent at least one of every two consecutive terms is at least $\eta$ away from $\pi/2$ mod $\pi$, so that they are at least $\sin(\eta)^{2}$. So, setting $\underline{k}=\lceil a/h \rceil $ and $\overline{k}=\lfloor b/h \rfloor$, putting this together, conclude that

\begin{equation*} 
\sum_{a/h\leq k \leq b/h} (\phi_{k}(x+h)-\phi_{k}(x))^{2} \geq 8\underline{c}_{s} \frac{\overline{k}-\underline{k}-1}{2}\sin^{2}(\eta) \geq c_{x}/h,
\end{equation*}

\noindent which is the required condition for $q=1$. Thus, for the Brownian eigenbasis, and eigenvalues bounded above and below by constant multiples of $1/k^{p}$, $\alpha^{*}(x) = \min\{2, p-1\}$ for $x \in (0,1)$. In particular, slower decay strictly increases jaggedness for $1 < p < 3$, with the same caveat that jaggedness is fixed for $p \geq 3$.

\newpage

\bibliographystyle{ecta}
\bibliography{orgcomplexity}

\end{document}